\documentclass[runningheads,orivec,a4paper,envcountsame]{llncs}
\pdfoutput=1

\spnewtheorem{definition}{Definition}{\bfseries}{\rmfamily}
\makeatletter\let\c@definition\relax\makeatother
\newaliascnt{definition}{theorem}

\usepackage[T1]{fontenc}
\usepackage{times}
\usepackage[numbers]{natbib}

\usepackage{xcolor}
\usepackage[colorlinks=true,linkcolor=violet,citecolor=violet,urlcolor=teal]{hyperref}
\usepackage{pgfplots}
\usepackage{svg}
\usepackage{bbding}
\usepackage{stmaryrd}
\usepackage{graphicx,multirow,tabularx,makecell}
\usepackage{microtype}
\usepackage[compact]{titlesec}
\usepackage{paralist}

\usepackage{booktabs}   

\RequirePackage{symbolic-bisim}
\RequirePackage{adjustbox}

\RequirePackage[textsize=tiny,disable]{todonotes} 
\RequirePackage{mathpartir}
\RequirePackage{cleveref}
\crefname{section}{Sec.}{Sec.(s)}
\crefname{appendix}{Appx.}{Appx.(s)}
\crefname{figure}{Fig.}{Fig.(s)}
\crefname{theorem}{Thm.}{Thm.(s)}
\crefname{definition}{Def.}{Def.(s)}
\crefname{proposition}{Prop.}{Prop.(s)}
\crefname{corollary}{Cor.}{Cor.(s)}
\crefname{lemma}{Lem.}{Lem.(s)}
\crefname{example}{Ex.}{Ex.(s)}
\crefname{remark}{Rem.}{Rem.(s)}
\crefname{equation}{}{}

\usepackage{versions}
\excludeversion{exclproof}
\excludeversion{technical}

\usepackage[misc]{ifsym}

\newcommand\cutout[1]{}

\begin{document}

\title{Pushdown Normal-Form Bisimulation: A Nominal Context-Free Approach to Program Equivalence
\thanks{This publication has emanated from research supported in part by a grant from Science Foundation Ireland under Grant number 13/RC/2094\_2; and the Cisco University Research Program Fund, a corporate advised fund of Silicon Valley Community Foundation.}}
\titlerunning{Pushdown Normal-Form Bisimulation}
%
\author{Vasileios Koutavas\inst{1} \and Yu-Yang Lin\inst{1} \and Nikos Tzevelekos\inst{2}}
\institute{Trinity College Dublin \and Queen Mary University of London}
%
%
\maketitle
\begin{abstract}
We propose Pushdown Normal Form (PDNF) Bisimulation to verify contextual equivalence in higher-order functional programming languages with local state.
Similar to previous work on Normal Form (NF) bisimulation, PDNF Bisimulation is sound and complete with respect to contextual equivalence.
However, unlike traditional NF Bisimulation, PDNF Bisimulation is also decidable for a class of program terms that reach bounded configurations but can potentially have unbounded call stacks and input an unbounded number of unknown functions from their context.
Our approach relies on the principle that, in model-checking for reachability, pushdown systems can be simulated by finite-state automata designed to accept their initial/final stack content.
We embody this in a stackless Labelled Transition System (LTS), together with an on-the-fly saturation procedure for call stacks, upon which bisimulation is defined.
To enhance the effectiveness of our bisimulation, we develop up-to techniques and confirm their soundness for PDNF Bisimulation.
We develop a prototype implementation of our technique which is able to verify equivalence in examples from practice and the literature that 
were out of reach for previous work.
\end{abstract}

 \section{Introduction}
 \label{sec:intro}
 \newcommand\Ed{\mathcal{E}}
\newcommand\Edb{\Ed}
\enlargethispage{10mm}

The problem of contextual equivalence for programming languages aims at determining whether two program terms exhibit the same operational behaviour within any given program context~\cite{Morris68}.
Although an undecidable problem, relatively recent work is pushing the frontier of decidable equivalence verification
in languages incorporating functional, higher-order paradigms, where the behaviour of a term can depend on external unknown code provided by the context as an argument \cite{hector,coneqct,syteci,KoutavasLT22,KoutavasLT23}.

Normal-Form (NF) bisimulation is a technique that treats unknown code (provided as higher-order arguments) symbolically.
The technique was originally defined for characterising L\'evy-Longo tree equivalence for the lazy lambda calculus~\cite{Sangiorgi:lazylambda} and adapted to languages with call-by-name~\cite{LassenHNF}, call-by-value~\cite{LassenENFB}, nondeterminism~\cite{LassenNondet}, aspects~\cite{JagadeesanPitcherRiely:aspects}, recursive types~\cite{LassenL07}, polymorphism~\cite{LassenLevy:nfbisimpoly}, control with state~\cite{StorvingLassenPOPL07}, state-only \cite{BiernackiLP19}, and control-only~\cite{BiernackiLengletNFB2}. More recently, it was used to create equivalence verification techniques for call-by-value functional languages with and without state~\cite{KoutavasLT22,KoutavasLT23}.

However, even NF bisimulations are prone to unbounded behaviour that needs to be explored to verify equivalence. A main source of such behaviour is the potential repeated nested calls between term and context which lead to unbounded stacks of term continuations being created by the bisimulation {exploration}. 
Such behaviour is common when programming with callback functions, as is the case in instances of the Observer Pattern \cite{design-patterns-book}, shown in the following ML example which models event listeners inspired by JavaScript. Similar examples have been showcased in the literature of program equivalence~\cite{DreyerNB10}.

\begin{example}\label{ex:motiv}
    $M,N : (\Unit \arrow \Unit) * (\Unit \arrow \Unit) \arrow \Unit \arrow \Unit$\\
  \begin{tabular}{@{}l@{\hspace{2ex}}|@{\hspace{2ex}}l@{}}
\begin{lstlisting}[boxpos=m]
${\ma{M=}}$ let createElement (onstart,onend) =
       ref flag = false in
       let event () =
         flag$\,$:=$\,$true; onstart ();
         flag$\,$:=$\,$false; onend (); !flag
       in event
     in createElement
\end{lstlisting}
%
  &
\begin{lstlisting}[boxpos=m]
${\ma{N=}}$ let createElement
         (onstart,onend) =
       let event () =
         onstart ();
         onend (); 0
       in event
     in createElement
\end{lstlisting}
%
\end{tabular}
\end{example} 
{In a proof of equivalence of $M$ and $N$, we can get an unbounded sequence of nested calls to \lstinline{event}, caused by the unknown functions \lstinline{onstart} and \lstinline{onend}.
  This makes the equivalence non-trivial as \lstinline{flag} may change values an arbitrary amount of times. However, each assignment of \lstinline{flag} to true is matched by one setting it back to false because each call of \lstinline{onstart} is matched by a return of this function. In other words, calls and returns of \lstinline{onstart} are well-bracketed and, hence, updates of \lstinline{flag} to true and then false are also well-bracketed.}

  Reasoning with such examples, for instance using Normal-Form bisimulation, requires the creation of an infinite candidate relation (due to the unbounded stack of nested calls) and then prove it a bisimulation~\cite{BiernackiLP19}. Although effective for hand-crafted proofs, such an approach would not work for a verification tool of equivalence, which would need to explore all tuples in the candidate relation.
In previous equivalence verification techniques such as~\cite{KoutavasLT22}, stacks of effectively pure functions were bounded with up-to techniques which however were unable to finitise the exploration\,---\,and thus prove equivalence\,---\, of stateful examples such as the one above.

In this work we propose
\emph{Pushdown Normal Form (PDNF) Bisimulation} {to finitise the exploration of such examples. This is}
an alternative NF bisimulation for a higher-order functional programming language with local state (\cref{sec:lang}) that
abstracts away stacks without losing precision, by relying on the fact that
traces of such interactions form a \emph{context-free language} and, when model-checking for reachability, they
can be simulated {precisely} by finite-state automata designed to accept their initial/final stack content~\cite{BouajjaniEM97,FinkelWW97}.

{We develop PDNF bisimulation on a behavioural LTS of a core-ML language.}
Contrary to the LTS in~\cite{KoutavasLT22} (reviewed in \cref{sec:lts}), the LTS we design here (\cref{sec:the-stackless-lts}) is \emph{stackless} and the definition of PDNF bisimulation incorporates a so-called \emph{saturation procedure}~\cite{BouajjaniEM97,FinkelWW97}, albeit performed on the fly as the bisimulation {exploration} evolves (\cref{sec:pdnf-bisim}).
Our approach follows exact-stack analyses used in Control-Flow Analysis (\cite{CFA2-jour,PDCFA,AAC} and in particular~\cite{P4F}), which similarly remove the need for an explicit continuation stack without losing precision.

This approach allows us to adapt a decidability result from nominal pushdown automata~\cite{ChengK98,MurawskiRT17} to program equivalence of higher-order stateful languages.
PDNF bisimulation equivalence is decidable between program terms that reach bounded stackless configurations, even though they may input an unbounded number of unknown functions from their context
and their corresponding {suppressed} stacks may be unbounded 
(\cref{sec:decidability}). This result is further amenable to up-to techniques, for example considering configurations up to garbage collection.

We establish that PDNF bisimulation is fully abstract for contextual equivalence by relating it to the NF bisimulation of~\cite{KoutavasLT22} (\cref{sec:soundness,sec:completeness}).
Furthermore, we increase the strength of our tool in proving equivalences, and similarly to~\cite{BiernackiLP19,KoutavasLT22},
we develop a number of bisimulation up-to techniques and prove their soundness for PDNF bisimulation (\cref{sec:up-to}).
These are powerful rules that allow us to reduce the size of the relation that we examine for bisimulation.
In particular, apart from simple techniques such as up to garbage collection, name permutation, and beta reductions, we develop up to separation and name reuse.
These two techniques, besides being sound,  are also complete in the sense that if after applying them an inequivalence is found, this is a real inequivalence and no backtracking is needed by the bisimulation verification procedure.

We modified the \Hobbit{} tool of~\cite{KoutavasLT22} to implement a bounded equivalence checker called \SLHobbit{} {(\cref{sec:imp})}, which remains bounded complete, i.e.\ it finds all inequivalences given sufficiently large bounds {and divergence detection.}
Our tool and this work
have the advantage of being able to finitise, and therefore prove, the otherwise infinite NF bisimulation exploration of equivalences that are beyond the reach of \Hobbit{} (\cref{sec:motivating-example}).
Of course, not all cases can be finitised this way, as
contextual equivalence in a Turing complete language is undecidable.
Finally, we discuss related and future work (\cref{sec:outro}).

\section{Motivating Examples}\label{sec:motivating-example}

We presented \Cref{ex:motiv} as a motivating instance of equivalence that can be resolved via PDNF bisimulation.
To give an intuitive understanding of the method, we next look at a simplified version of \Cref{ex:motiv}  and how its NF bisimulation game~\cite{BiernackiLP19,KoutavasLT22} {becomes} infinite because of nested context calls.
%
%
For the next example, and to prepare for the developments in the main body of the paper,
  we shall
follow more closely the style of presentation in~\cite{KoutavasLT22} and
describe the interactions between a term and its context, and the ensuing LTS, using terminology taken from  game semantics~\cite{AJM,HO,Nickau}.\footnote{The terms ``proponent'', ``opponent'' and ``move'' will be all the terminology we use from game semantics in this paper;
 the term ``game'' will almost exclusively refer to the bisimulation game, which {is traditionally} between  \emph{Challenger} and \emph{Defender}.}
In particular,
we shall refer to the examined term as
the \emph{Proponent}, whereas its syntactic context will be the \emph{Opponent}.
The two parties, i.e.\
proponent and opponent, can interact by issuing \emph{moves}, which are simply calls to functions provided by the opposite party and their corresponding returns. The bisimulation game is based on the matching of these moves.

\begin{figure}[t]
  \hspace{-3mm}\includegraphics[scale=0.4]{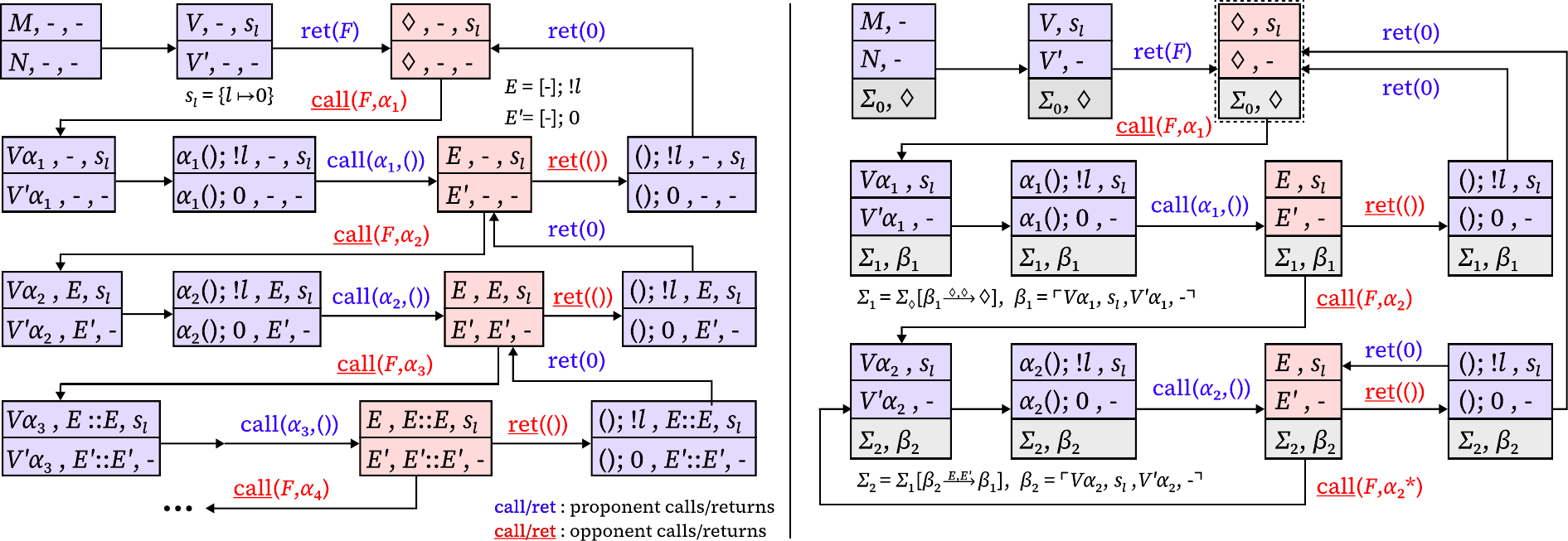}
  \caption{NF bisimulations for terms $M$ and $N$ (\cref{ex:dummy}); standard/stacked (left) and pushdown/stackless (right). We use ``$\diamond$'' to denote the top-level continuation and entry point, and set $\Sigma_0$ to be the empty continuation graph. We write ``-'' for the empty store and continuation stack.}
  \label{fig:hard}
\end{figure}

\begin{example}\label{ex:dummy}
Consider the following equivalent terms, the NF bisimulation game of which is depicted in \cref{fig:hard} (left).
\medskip

\begin{tabular}{@{\qquad}l@{\qquad}|@{\qquad}l}
\begin{lstlisting}[boxpos=m]
$M=$ let x = ref 0 in
${\color{white}M=}$ fun f -> f(); !x

$N=$ fun f -> f(); 0
\end{lstlisting}
  &
\begin{lstlisting}[boxpos=m]
$V=$ fun f -> f(); !l
$(\text{for location }$l$)$

$V'=$ fun f -> f(); 0
\end{lstlisting}
\end{tabular}
\medskip

\noindent
The game involves
pairs of configurations\footnote{For expository reasons, the notation used here for configurations, their components and the LTS is a simplified version of the one used later on when these notions are formally defined.} of the form $(\Phi,K,s)$, where $\Phi$ is either a term of the language (in configurations where proponent plays next) or a continuation (when opponent plays next), $K$ is a continuation stack, and $s$ is a local store. Initially,
proponent returns an abstract function $F$ representing respectively
the functions $V$ and $V'$ (move $\lpropret{F}$).
Next, the bisimulation game can engage in a series of moves as on the left below, where opponent repeatedly calls $F$ with (fresh) arguments $\alpha_1,\alpha_2,\dots$,
\[
  \lopapp{F}{\alpha_1}\,\lpropapp{\alpha_1}{()}\,\lopapp{F}{\alpha_2}\,\lpropapp{\alpha_2}{()}\,\dots\quad
  K = E{::}E{::}\ldots,\; K'=E'{::}E'{::}\ldots
\]
thus leading to unbounded continuation stacks $K$ and $K'$ respectively as on the right.

In \cref{fig:hard} (right) we can see the PDNF bisimulation game.\footnote{The loop transition at the bottom right, labelled $\lopapp{F}{\alpha_2^*}$, represents a transition \emph{for each fresh} $\alpha_2$. This representation is informal and used here for economy to demonstrate finiteness. In fact, the pushdown NF bisimulation would not be finite but, instead, \emph{orbit-finite} (cf.\ \Cref{ex:revisited}).} We observe that now configurations are \emph{stackless} pairs $(\Phi,s)$, and that we have incorporated an additional \emph{environment} component $(\Sigma,\beta)$. The latter is an over-approximation of the (combined) stack structure which records:
\begin{itemize}
\item
  The opponent call that is currently being evaluated, by means of an \emph{entry point} $\beta$: this is simply the pair of configurations $(C_1,C_2)$ that the call led to and, in this case, there is only one such pair
  with terms $V\alpha$ and $V'\alpha$ (and corresponding stores).
\item
  The possible sequencings of entry points $\beta$, using a \emph{continuation graph} $\Sigma$.
  Edges in $\Sigma$ are of the form $\beta'\trans{\Ed_1,\Ed_2}\beta$ which denote that, starting from $\beta$, we are led to an opponent call 
  with continuations $\Ed_1,\Ed_2$ and resulting entry point $\beta'$.
\end{itemize}
We assume by convention that there is a top-level opponent call with $\beta=\diamond$ which starts the bisimulation game. We can see in \cref{fig:hard} that every path in the graph on the left has a corresponding path in the graph on the right. In other words, PDNF bisimulation is sound.
On the other hand, the graph on the right has infeasible paths. For example we can form a path from the initial vertex to the highlighted top-level one with trace:
\[
\lpropret{F}\,\lopapp{F}{\alpha_1}\,\lpropapp{\alpha_1}{()}\,\lopapp{F}{\alpha_2}\,\lpropapp{\alpha_2}{()}\,\lopret{()}\,\lpropret{0}
\]
which breaks the stack discipline (we reach the top level while doing more pushes than pops). However, such spurious paths are harmless and do not affect completeness of our method: all pairs of configurations that are spuriously reached can also be reached by real paths. In fact, the  (highlighted) vertex reached by the path above was already reached after the first move $\lpropret{F}$.
\end{example}



  \section{Language and Semantics}
  \label{sec:lang}
  \newcommand\Nam{\mathsf{ANam}}
\newcommand\Perm{\mathsf{Perm}}
\newcommand\Names{\mathsf{Names}}
\newcommand\orb[1]{\mathsf{orb}(#1)}

\begin{figure*}[t]
  \footnotesize
  \[\begin{array}{r@{\;\;}r@{\;\;}c@{\;\;}l}
    {\Loc:} & l,k &
    \multicolumn{2}{l}{
      \hspace{3em}
    {\Var:}\;\;  x,y,z       
      \hspace{3em}
    {\Con: }\;\;  c 
      \hspace{3em}
    {\Nam: }\;\;  \alpha 
    }\\
    {\Typ:} & T           & \mis & \Bool \mor \Int \mor \Unit \mor T \arrow T \mor T_1 * \ldots * T_n \\
      {\Exp: } & e,M           & \mis & v \mor (\vec e)\mor \arithop{\vec e} \mor \app e e \mor \cond{e}{e}{e} \mor \new l v e \mor \deref l \mor l \asgn e
 \mor \elet{(\vec x)}{e}{e} \\                                             
    {\Val:} & u,v       & \mis & c \mor x\mor \alpha_{T\to T}\mor \lam[f_{T\to T}] x e   \mor (\vec v)                                                                            \\ 
    {\EC:} & E        & \mis & \hole_T \morcondensed (\vec v,E,\vec e) \morcondensed \arithop{\vec v,E,\vec e} \morcondensed \app E e \morcondensed \app v E \morcondensed l \asgn E \morcondensed \cond E e e \morcondensed \elet{(\vec x)}{E}{e} \\
    {\Cxt:} & D           & \mis & \hole_{i,T} \mor e \mor (\vec D) \mor \arithop{\vec D} \mor \app D D \mor l \asgn D \mor \cond D D D \mor \lam[f_{T\to T}] x D
                                    \\
      &&& \mor \new l D D\mor \elet{(\vec x)}{D}{D}       \\
    \end{array}\]
  \[\begin{array}{@{}l@{\;\,}l@{\;\,}lll}
    \redconf{s}{\arithop{\vec c}}    & \redbase & \redconf{s}{w}                      & \text{if } \mathop{op}^{\textsf{arith}}(\vec c) = w & 
    \multirow{2}{*}{\fbox{\;\;\parbox{.2\linewidth}{$\redconf{s}{e}\in\Exp\times\St$\\  
${\St}=\Loc \overset{\mathsf{fin}}{\rightharpoonup} \Val$}}}
\\
    \redconf{s}{\app{(\lam[f] x e)}  v} & \redbase & \redconf{s}{e\sub{x}{v}\sub{f}{\lam[f] x e}}  \\
    \redconf{s}{\elet{(\vec x)}{(\vec v)}{e}} & \redbase & \redconf{s}{e\sub{\vec x}{\vec v}}  \\
    \redconf{s}{\new l v e}          & \redbase & \redconf{s\stupd{l}{v}}{e}          & \text{if } l \not\in\dom{s}\\
    \redconf{s}{\deref l}            & \redbase & \redconf{s}{v}                      & \text{if } s(l)=v \\
    \redconf{s}{l \asgn v}           & \redbase & \redconf{s\stupd{l}{v}}{()}   \\
    \redconf{s}{\cond{c}{e_1}{e_2}}  & \redbase & \redconf{s}{e_i}                    & \multicolumn{2}{l}{\text{if } (c,i) \in \{(\true,1), (\false,2)\}} \\
    \redconf{s}{E\hole[e]}           & \red     & \redconf{s'}{E\hole[e']}            & \text{if } \redconf{s}{e} \redbase \redconf{s'}{e'}
  \end{array}\vspace{1mm}\]
  \hrule\vspace{-1mm}
  \caption{Syntax and reduction semantics of the language \lang.}\label{fig:lang}
\end{figure*}

We work with \lang, a simply-typed call-by-value lambda calculus with local state~\cite{KoutavasLT22}.
The syntax and operational semantics are shown in \cref{fig:lang}.
Expressions (\Exp) include the standard lambda expressions with recursive functions ($\lam x e$), together with location creation ($\new l v e$), dereferencing ($\deref l$), and assignment ($l \asgn e$), as well as standard base type constants ($c$) and operations ($\arithop{\vec e}$).
Locations are mapped to values, including function values, in a store (\St). We write \emptyS for the empty store
and let $\fl{X}$ denote the set of free locations in syntactic or semantic object $X$.
Values consist of boolean, integer, and unit constants, functions and arbitrary length tuples of values. Functions consist of standard functions ($\lam x e$) as well as \emph{abstract} ones ($\alpha$) sourced from a typed-indexed set of countably infinite sets of \emph{abstract names} $\Nam=\biguplus_{T,T'}\Nam_{T\to T'}$.
These correspond to environment (unknown) functions and are used in the open-term LTS used in NF bisimulation. Given an object $X$, we write $\an{X}$ for the set of abstract names appearing in $X$.


The language \lang is simply-typed with typing judgements of the form $\typing{\Delta}{\Lambda}{e}{T}$, where $\Delta$ is a type environment (omitted when empty), $\Lambda$ a store typing and $T$ a type (\Typ).
Abstract functions are explicitly typed in terms, and we assume that said typing is consistent within terms.
The rules of the type system are standard and omitted here. We call an expression $e$ \emph{closed} when $\Lambda\vdash e:T$.

The reduction semantics is by small-step transitions between configurations containing a store and an expression, $\redconf{s}{e} \red \redconf{s'}{e'}$,
defined using single-hole evaluation contexts ($\EC$) over a base relation $\redbase$.
Holes $\hole_T$ are annotated with the type $T$ of closed values they accept, which we may omit to lighten notation.
Stores map locations to closed values; the latter are uniquely typed and, thus, each store $s$ yields a store typing $\Lambda_s$.
Beta substitution of $x$ with $v$ in $e$ is written as $e\sub{x}{v}$.
We write $\redconf{s}{e}\trm$ to denote $\redconf{s}{e} \red^* \redconf{t}{v}$ for some $t$, $v$.
We write $\vec X$ to mean a syntactic sequence, and assume standard syntactic sugar from the lambda calculus.
In our examples we assume an ML-like syntax and implementation of the type system, which is also the concrete
syntax of our tool (same syntax as that used in \Hobbit~\cite{KoutavasLT22}).
We write $\bot$ for a diverging computation.


Contexts $D$ contain multiple, non-uniquely indexed holes $\hole_{i,T}$, where $T$ is the type of value that can replace the hole (each index can have one associated type).
A context is called \emph{canonical} if its holes are indexed $1,\dots,n$, for some $n$.
Given a canonical context $D$ and a sequence of typed expressions $\Lambda\vdash\vec e:\vec T$, notation $D[\vec e]$ denotes the context $D$ with each hole $[\cdot]_{i,T_i}$ replaced by $e_i$.
We omit hole types and indices where possible.
We assume the Barendregt convention for locations, thus replacing context holes avoids location capture (note $\mathsf{ref}$ is a binder).
Standard contextual equivalence~\cite{Morris68} follows.

\begin{definition}[Contextual Equivalence]\label{def:cxt-equiv}
  Expressions $\vdash e_1:T$ and $\vdash e_2:T$
with $\an{e_1}=\an{e_2}=\emptyset$
  are \emph{contextually equivalent}, written as $e_1 \cxteq e_2$, when for all contexts $D$ such that $\vdash D\hole[e_1]: \Unit$ and $\vdash D\hole[e_2]: \Unit$ we have
$
    \redconf{\emptyS}{D\hole[e_1]}\trm ~\text{iff}~
    \redconf{\emptyS}{D\hole[e_2]}\trm
$.
\end{definition}

We finally consider environments $\Gamma \in \Nat \xrightharpoonup{\mathsf{fin}} \Val$ which map natural numbers to closed values.
The concatenation of two such environments $\Gamma_1$ and $\Gamma_2$, written $\Gamma_1,\Gamma_2$ is defined when $\dom{\Gamma_1}\cap\dom{\Gamma_2}=\emptyset$.
We write $(\maps{i_1}{v_1}, \ldots, \maps{i_n}{v_n})$ for a concrete environment mapping $i_1, \ldots, i_n$ to $v_1, \ldots, v_n$, respectively.
Environment $\Gamma$ can be used to fill in holes of context $D$ with matching indices; we refer to the result as $D[\Gamma]$. 
When indices are unimportant we omit them and treat $\Gamma$ environments as lists.


\subsubsection*{Names and permutations.}

It is useful to introduce notation that allows us to easily reason on locations, abstract names and environment indices, which we collectively refer to as \emph{names}:
\[
  \Names = \Loc\cup\Nam\cup\Nat
\]%
These appear in the syntax and semantics of our language
in a \emph{nominal} way: the identity of a given name is immaterial\,--\,what is relevant is how the name compares to other names in its environment.
Technically speaking, our constructions are founded on \emph{nominal sets}~\cite{PittsNominal}. Below, we refer to elements in our syntax and semantics as \emph{objects}.

\begin{definition}[Permutations]\label{def:perms}
  We consider permutations of store locations ($\pi_l$), abstract names ($\pi_\alpha$, which are type-reserving) and environment indices ($\pi_i$), respectively. 
  We combine these in permutations $\pi$ of the form $\pi_l\uplus\pi_\alpha\uplus\pi_i$, which
  we compose as functions (e.g.\ we may write $\pi\circ\pi'$). 
  We restrict our attention to \emph{finitary} permutations $\pi$, i.e.\ such that the set $
      \supp{\pi}=\{x\in\Names\mid \pi(x)\not= x\}$ be finite. We let $\Perm$ be the set of all finitary permutations.
Given names $x,x'$ we write $(x\ x')$ for the permutation that swaps $x$ with $x'$ (and fixes all other names).
\end{definition}

Given an object $X$ and (finitary) permutation $\pi$, we write $\pi\cdot X$ for the result of \emph{applying $\pi$ on $X$}.
The result of applying a permutation on an object $X$ is done as expected, e.g.\ applying
a permutation $\pi_l$ to a store $s$, the former acts on both the domain and range of
the latter. When applying a permutation $\pi_i$, we treat environment index $i$ differently than other instances of the natural number $i$.

\begin{definition}[Nominal set and orbit-finiteness]
  Given an object $X$, its \emph{support} $\supp{X}$ is the least $S\subseteq\Names$ such that permutations fixing all $x\in S$ also fix $X$:
  \[
\forall\pi\in\Perm.\ (\forall x\in S.\,\pi(x)=x)\implies \pi\cdot X=X.
\]
We henceforth assume that all objects have finite support.
$X$ is called \emph{equivariant} if $\supp{X}=\emptyset$, in which case $\pi\cdot X=X$ for all $\pi$.
A set $\mathcal X$ of objects is called a \emph{nominal set} if it is closed under permutation, i.e.\ $\pi\cdot X \in\mathcal X$ for all $\pi\in\Perm$ and $X\in\mathcal X$.
Given $x\in\Names$, we say that \emph{$x$ is fresh for $X$}, and write
$x\fresh X$, when $x\notin\supp{X}$.
\\
Given object $X$, its \emph{orbit} is defined by: \
$
\orb{X}= \{\pi\cdot X\mid \pi\in\Perm\}
$.
Nominal set $\mathcal X$ is \emph{orbit-finite} if its  set of orbits
$
\{\orb{X}\mid X\in\mathcal X\}
$
is finite.
\end{definition}

Note that in finite objects (e.g.\ terms of \lang), the support of an object typically coincides with the set of free names featuring in it. In such a case, writing e.g.\ $\vec\alpha\fresh \vec X$ will stand for $\forall i,j.\,\alpha_i\notin\an{X_j}$.
Orbit-finiteness is central in computability with nominal sets~\cite{BojanczykKLT13,BojanczykKL14} and can be seen as the analogue of finiteness in nominal sets.


  \section{Stacked LTS and NF Bisimulation}
  \label{sec:lts}
  \newcommand{\realtrans}[2][]{\xmapsto[#1]{#2}}
\newcommand\hsig[1]{\mathsf{sig}(#1)}
\renewcommand{\conf}[5]{\ma{\langle #2 \mathop{;} #3 \mathop{;} #4 \mathop{;} #5\rangle}}
\renewcommand{\oldpconf}[5]
{\ma{\color{blue}\langle\color{black} #2 \mathop{;} #3 \mathop{;} #4 \mathop{;} #5 \color{blue}\rangle\color{black}}
}
\renewcommand{\oldoconf}[5]
{\ma{{\color{red}\langle} #2 \mathop{;} {#3} \mathop{;} #4  \mathop{;} #5 {\color{red}\rangle}}
}
\renewcommand{\pconf}[4]
{\ma{\color{blue}\llparenthesis\color{black} #2 \mathop{;} #3 \mathop{;} #4 \color{blue}\rrparenthesis\color{black}}
}
\renewcommand{\oconf}[4]
{\ma{\color{red}\llparenthesis\color{black} #2 \mathop{;} {#3} \mathop{;} #4 \color{red}\rrparenthesis\color{black}}
}

We next recall the LTS and NF bisimulation presented in~\cite{KoutavasLT22}.
As mentioned in \cref{sec:motivating-example}, the LTS is based on game semantics and uses opponent and proponent call and return transitions:
proponent transitions are the moves of an expression interacting with its context;
opponent transitions are the moves of the context surrounding the expression.
These transitions are over \emph{proponent, opponent and divergence configurations}, respectively:
\[\oldpconf{A}{\Gamma}{K}{s}{e},\;\;\oldoconf{A}{\Gamma}{K}{s}{\Ed}
\;\;\text{ and }\;\;\botconf
  .\]
$\botconf$ is a special configuration which is used in order to represent expressions that cannot perform given transitions (cf.~Remark~\ref{rem:bots}).
In other configurations:
\begin{itemize}
  \item $\Gamma$ is an environment indexing proponent functions known to opponent;
  \item $K$ is a stack of \emph{continuations} $\Ed$, created by opponent calls; a {continuation} is either an {evaluation context $E$} or the constant $\diamond$ (for the \emph{top-level}, empty continuation);
  \item $s$ is the store containing proponent locations;
  \item $\Ed$ is a {continuation},
    and is either {the most recent evaluation context $E$} or $\diamond$;
  \item $e$ is the expression reduced in proponent configurations.
  \end{itemize}
Given a configuration $C$, we write $C.\Gamma$ ($C.K$, etc.) for the first  (second, etc.) component of $C$; if $C=\botconf$ then by convention $C.\_=\bot$.
We shall use $\Phi$ to range over $\Ed$ and $e$.

Compared to~\cite{KoutavasLT22}, we have made minor technical modifications in the structure of configurations to ensure uniformity between this and the stackless LTS of the next section.
In particular we 
(1) drop sets of abstract names from configurations, replacing them with a mutual freshness condition for new abstract names in the bisimulation (\cref{def:bisim});
(2) separate the most recent evaluation context from those stacked in $K$ in opponent configurations.
To ensure that $\Ed=\diamond$ corresponds to a top-level configuration we require that opponent configurations $\oldoconf{A}{\Gamma}{K}{s}{\Ed}$ satisfy the condition on the left below,
while the push operation on stacks $K$ is defined as on the right.
\[
\begin{array}{l@\qquad|@{\qquad}l}
\Ed=\diamond\iff K=\noe
&
  \Ed,K =\begin{cases} E,K &\text{if }\Ed=E\\
          \noe &\text{if } \Ed=\diamond\text{ and } K={\noe}\\
          \text{undefined} &\text{otherwise}
          \end{cases}
\end{array}
\]
The LTS uses \emph{moves} of the forms:
$\eta\, ::=\, \lopapp{i}{D[\vec\alpha]}\mid \lopret{D[\vec\alpha]}\mid
  \lpropapp{\alpha}{D}\mid \lpropret{D}
$.
Underlined moves are \emph{opponent moves}, and the rest are \emph{proponent moves}.
Contexts $D$ are picked from the following restricted grammar (of values with higher-order holes):
\[
D^\bullet ::= c\mid\hole_{i,T\to T}\mid (\vec D_\bullet)
\]
Given such a context $D^\bullet$, we can derive its \emph{hole signature} $\hsig{D^\bullet}\in\Nat^*$ setting:
\[
\hsig{c}=\varepsilon,\quad
\hsig{\hole_{i,T}}=i,\quad
\hsig{(D_1^\bullet,\dots,D_n^\bullet)}=\hsig{D_1^\bullet},\dots,\hsig{D_n^\bullet}.
\]
We stipulate that $\hsig{D^\bullet}$ must be a non-repeating sequence (i.e.\ every hole appears exactly once).
Given a value $v$, we can extract its \emph{ultimate pattern} \cite{LassenL07}, which is a pair $(D^\bullet,\Gamma)$,
and extend $\ulpatt$ to types through the use of abstract function names:
\begin{align*}
  (D^\bullet,\Gamma)\in\ulpatt(v) &\iff v=D^\bullet[\Gamma]\land\dom{\Gamma}=\{i\mid i\in\hsig{D^\bullet}\}
\\
  (D^\bullet,\Gamma) \in\ulpatt(T) &\iff {}\vdash D^\bullet\hole[\Gamma] : T\land
\Gamma:\{i\mid i\in\hsig{D^\bullet}\}\to\Nam 
\end{align*}
In the latter case, we write $(D^\bullet,\Gamma)$ simply as $(D^\bullet,\vec\alpha)$, where $\vec\alpha=\Gamma(i_1),\dots,\Gamma(i_k)$ and $i_1,\dots,i_k=\hsig{D^\bullet}$.
For economy, we will henceforth denote these contexts by $D$.

\begin{definition}[Stacked LTS]\label{def:lts}
The LTS is defined by the rules in \cref{fig:lts}. We write $C\realtrans{\eta}C'$ if $C\trans{\eta}C'$ without using the \textsc{Response} rule. We write $C\downarrow$ if $C=\oldoconf{A}{\Gamma}{\noe}{s}{\diamond}$.
\end{definition}

\begin{figure*}[t] 

  \[\begin{array}{@{}lll@{}}
    \irule![PropCall][propappf]{
      (D,\Gamma') \in \ulpatt(v)
    }{\nbox{
      \oldpconf{A}{\Gamma}{K}{s}{E[\app {\alpha} v]}
      \trans{\lpropapp{\alpha}{D}} 
      \oldoconf{A}{\Gamma,\Gamma'}{K}{s}{E}
    }}
    \\
      \irule![PropRet][propretf]{
      (D,\Gamma') \in \ulpatt(v)
    }{
      \oldpconf{A}{\Gamma}{\Ed,K}{s}{v}
      \trans{\lpropret{D}}
      \oldoconf{A}{\Gamma,\Gamma'}{K}{s}{\Ed}
    }
    \\
     \irule![OpCall][opappf]{
       \nbox{
        \vec\alpha\fresh\Gamma,K,s,\Ed\land (D,\vec \alpha) \in \ulpatt(T)\\
       {\hspace{-4mm}}\land\Lambda_s \vdash \Gamma(i) : T \arrow T'\land
        \app {\Gamma(i)} {D\hole[\vec \alpha]} \funred e
      }
    }{
      \oldoconf{A}{\Gamma}{K}{s}{\Ed}
      \trans{\lopapp{i}{D\hole[\vec \alpha]}}
      \oldpconf{A\uplus\vec\alpha}{\Gamma}{\Ed,K}{s}{e}
    }
    \\
    \irule![OpRet][opretf]{
        \vec\alpha\fresh\Gamma,K,s,E\land (D,\vec \alpha) \in \ulpatt(T)
    }{
      \oldoconf{A}{\Gamma}{K}{s}{E\hole_T}
      \trans{\lopret{D\hole[\vec\alpha]}}
      \oldpconf{A\uplus\vec\alpha}{\Gamma}{K}{s}{E\hole[D[\vec\alpha]]}
    }
    \\
    \irule![Tau][tau]{
      \redconf{s}{e} \red \redconf{s'}{e'}
    }{
      \oldpconf{A}{\Gamma}{K}{s}{e}
      \trans{\tau}
      \oldpconf{A}{\Gamma}{K}{s'}{e'}
   }
    \\
    \irule![Response][dummy]{
      \nbox{\eta \not= \tau\text{ and }
      C\not\wtrans{\eta}\text{ from other rules}}
    }{
      C \trans{\eta} \botconf
    }
    \\[-1mm]
  \end{array}\]
  \hrule
  \caption{The stacked Labelled Transition System (following~\cite{KoutavasLT22}). We let $vu\funred e$ mean $e=\alpha u$ when $v=\alpha$; and $e=e'[u/x][\lam[f] x e'/f]$ when $v=\lam[f] x e'$.}\label{fig:lts}
\end{figure*}

We next introduce a notion of boundedness on terms that examines the sizes of all their possible descendant configurations in their LTS, ignoring stacks $K$.

\begin{definition}\label{def:cxt-free-term}
  Define a size function $\|\cdot\|$ for expressions inductively as:
  \begin{align*}
    \|c_\Int\| &= |c|+1\ & \|x\| &=1 &\|\alpha_{T\to T}\|&=1
    & \|\lam[f_{T\to T}] x e\| &= 1+\|e\|\\
\|c_T\| &= 1\,(T\neq\Int) &
 \|\deref l\| &=1 & \|l \asgn e\| &= 1+\|e\| &\|\new l v e\| &= 1+\|v\|+\|e\|
  \end{align*}
  and $\|\kappa(e_1,\dots,e_n)\|=1+\|e_1\|+\cdots+\|e_n\|$ for all other $n$-ary syntactic constructs $\kappa$. Extend this to
continuations by $\|E\|=\|E[x]\|$ and $\|\diamond\|=1$, and to
  configurations by:
\begin{align*}
\|\botconf\| &=1 & \|\conf{A}{\Gamma}{K}{s}{\Phi}\| &= \max\left(\sum\nolimits_{i\in\dom{\Gamma}}\|\Gamma(i)\|,\sum\nolimits_{l\in\dom{s}}\|s(l)\|,\|\Phi\|\right)
 \end{align*}
Call expression $\vdash e:T$ \emph{context-free with bound $k$} if the set 
$
\{ \|C\| \mid \exists t.\ \oldpconf{\emptyA}{\emptyG}{\emptyK}{\emptyS}{e}\transs{t}C\}
$
is upper-bounded by $k\in\mathbb N$. Let $e$ be \emph{context-free} if it is context-free with some bound $k$.
\end{definition}

Thus, an expression is context-free when its stacked LTS can be represented as a nominal pushdown system~\cite{ChengK98,MurawskiRT17}. We shall show equivalence is decidable for these expressions, {using the PDNF bisimulation in the following section}.

We next present NF bisimulation. We write that \emph{move $\eta$ introduces $\vec\alpha$} if $\eta$ is a proponent move and $\vec\alpha$ is empty, or $\eta\in\{\lopapp{i}{D[\vec\alpha]},\lopret{D[\vec\alpha]}\}$ for some $i,D$.
Moreover, $\wtrans{\eta}$ means ${\trans{\tau}}{}^*$, when $\eta=\tau$; and $\wtrans{\tau}\trans{\eta}\wtrans{\tau}$ otherwise.

\begin{definition}[NF Bisimulation]\label{def:bisim}
    Configurations $C_1,C_2$ are called \emph{compatible} whenever $\botconf\in\{C_1,C_2\}$, or
  $C_1,C_2$ have same polarity and $\dom{C_1.\Gamma}=\dom{C_2.\Gamma}$.
Relation \bisim{R} between compatible configurations is a \emph{weak simulation} when
  for all $C_1 \bisim*{R} C_2$:
  \begin{compactitem}
    \item if $C_1\downarrow$ then $C_2\downarrow$,
   \item if $C_1 \realtrans{\eta} C_1'$ with $\eta$ introducing $\vec\alpha\fresh C_2$ then  $C_2 \wtrans{\eta} C_2'$ and $C_1' \bisim*{R} C_2'$.
\end{compactitem}
  If \bisim{R}, $\bisim{R}^{-1}$ are weak simulations then
  \bisim{R} is a \emph{weak bisimulation}.
  Similarity $(\simil)$ and bisimilarity $(\bisimil)$ are the largest weak 
  simulation and bisimulation, respectively.
  \defqed
\end{definition}

\begin{remark}\label{rem:bots}
 Following~\cite{KoutavasLT22}, 
any proponent configuration that cannot match a standard bisimulation transition challenge can trivially respond to the challenge by transitioning into $\botconf$ by the \iref{dummy} rule in \cref{fig:lts}.
By the same rule, this configuration can trivially perform all non-$\tau$ transitions.
  While in \emph{loc.~cit.} there is an {explicit termination transition from top-level, non-$\botconf$ configurations, here we choose instead to use the termination predicate $\downarrow$ to signify the end of a complete trace}.
{To obtain determinacy, we also impose trivial transitions to $\botconf$ to take place only if same-labelled transitions are not possible by the other rules.
The differences made in this section are inessential
  leaving the full abstraction result of~\cite{KoutavasLT22} unaffected.}
\end{remark}
\begin{definition}[NF Bisimilar Expressions]
  Expressions $\vdash e_1: T$ and $\vdash e_2: T$
with $\an{e_1}=\an{e_2}=\emptyset$
  are \emph{NF bisimilar}, written $e_1 \bisimil e_2$, when 
  $\oldpconf{\emptyA}{\emptyG}{\emptyK}{\emptyS}{e_1} \bisimil
  \oldpconf{\emptyA}{\emptyG}{\emptyK}{\emptyS}{e_2}$.
  \defqed
\end{definition}


\begin{theorem}[Full abstraction~\cite{KoutavasLT22}]\label{theorem:SC}
  $e_1 \bisimil e_2$ iff $e_1 \cxteq e_2$.  
\qed\end{theorem}



%
%
%

Finally,  similarity is a nominal set and closed under $\tau$-transitions (cf.~\cite{KoutavasLT22}).
\begin{lemma}\label{lem:closures}
  Given $C_1\simil C_2$, 
  if $C_i\wtrans{\tau}C_i'$ or
  $C_i'\wtrans{\tau}C_i$ (for $i=1,2$) then $C_1'\simil C_2'$.
  Moreover, for all $\pi\in\Perm$, $\pi\cdot C_1\simil \pi\cdot C_2$.
\qed\end{lemma}


  \section{Stackless LTS and Pushdown NF Bisimulation}
  \label{sec:stackless-lts}
\renewcommand{\pconf}[4]
{\ma{\color{blue}\llparenthesis\color{black} #2 \mathop{;} #3 \mathop{;} #4 \color{blue}\rrparenthesis\color{black}}
}
\renewcommand{\oconf}[4]
{\ma{\color{red}\llparenthesis\color{black} #2 \mathop{;} {#3} \mathop{;} #4 \color{red}\rrparenthesis\color{black}}
}
\newcommand{\cconf}[4]{\ma{\llparenthesis  #2 \mathop{;} #3 \mathop{;} #4 \rrparenthesis}}
\newcommand\botcconf{\ma{\llparenthesis\bot\rrparenthesis}}
\newcommand{\betais}[1]{\ulcorner #1 \urcorner}
\newcommand{\betaisR}[2]{\llcorner #1 \lrcorner_{#2}}
\newcommand{\ocalls}[1]{\ma{\mathsf{ocalls}(#1)}}
\newcommand\transNR[1]{\trans[{\sf NR}]{#1}}
\newcommand\wtransNR[1]{\wtrans[{\sf NR}]{#1}}
\newcommand\transsNR[1]{\transs[{\sf NR}]{#1}}
\newcommand\similNR{\simil_{\sf NR}}
\newcommand\bisimilNR{\ma{\approx_{\sf NR}}}
\newcommand\realtransNR[1]{\realtrans[{\sf NR}]{#1}}
\newcommand\sigb{_{\Sigma,\beta}}
\newcommand\sigbr{_{\Sigma^{-1},\beta^{-1}}}
\newcommand\sigbrp{_{\Sigma^{-1},\beta'^{-1}}}
\newcommand\sigbb{_{\Sigma',\beta'}}
\newcommand\leftdiverge{\bisim*R_{\uparrow\bot}}
\newcommand\Kont{\mathsf{Kont}}
\newcommand\bisimilPD{\bisimil_{\mathsf{PD}}}
\newcommand\toSig[2][\Sigma]{\mathrel{\trans{#2}_{#1}}}
\newcommand\toSigs[2][\Sigma]{\mathrel{\transs{#2}_{#1}}}
\newcommand\toSigsD[2][\Sigma]{\mathrel{\transs{#2}_{#1}}\!\diamond}

\subsection{The Stackless LTS}\label{sec:the-stackless-lts}
In the LTS of \Cref{fig:lts} the stack is manipulated solely by rules
{\sc OpCall} (push $\Ed$ on the stack) and {\sc PropRet} (pop last $\Ed$). Thus, in a bisimulation game where both component configurations are not blocked (i.e.\ not $\botconf$), the stack operations of the two components are synchronised. For instance, if we currently are in configuration pair $(C_1,C_2)$ and a challenge is made in $C_1$ that pushes some $\Ed_1$ on $C_1.K$, then the response in $C_2$ must push some $\Ed_2$ on $C_2.K$. The bisimulation game can thus be seen as using a single stack of pairs $(\Ed_1,\Ed_2)$. This in turn allows us to remove the stack component from configurations and attach it to bisimulations as an \emph{environment} component (which we can then apply abstractions on). This is the intuition behind the stackless LTS that we present next.


We start off with stackless configurations, which  will  be triples of the forms:
\[
\text{Proponent: }\;  \pconf{A}{\Gamma}{s}{e}\,;\quad\text{Opponent: }\;
\oconf{A}{\Gamma}{s}{\Ed}\;\text{ or }\;\oconf{A}{\Gamma}{s}{\chi}\,;\quad
\text{Divergence: }\; \botcconf\,.
\]
Here the constant $\chi$
stands for an unspecified continuation and it is used merely as a placeholder so that we can later apply a substitution of the form $\sub{\chi}{\Ed}$. This will become clearer in \Cref{def:SLbisim}; for now we can think that a configuration of the form $\oconf{A}{\Gamma}{s}{\chi}$ may silently reduce to $\oconf{A}{\Gamma}{s}{\Ed}$ for selected previously encountered continuations $\Ed$.

\begin{definition}[Stackless LTS]\label{def:SLlts}
  The stackless LTS has the exact same rules as those in \cref{fig:lts}, with the exception that configurations are now stackless, $K$ is dropped from the side conditions, and rules {\sc OpCall} and {\sc PropRet} have the following transitions while maintaining the same side-conditions (see~\cref{sec:stackless-lts-full}):
  \begin{align*}
     \oconf{A}{\Gamma}{s}{\Ed}
      \trans{\lopapp{i}{D\hole[\vec \alpha]}}
      \pconf{A\uplus\vec\alpha}{\Gamma}{s}{e}
      &&
      \pconf{A}{\Gamma}{s}{v}
      \trans{\lpropret{D}}
      \oconf{A}{\Gamma,\Gamma'}{s}{\chi}
  \end{align*}
  {We write $C{\realtrans{\eta}}C'$ if $C{\trans{\eta}}C'$ without using the \textsc{Response} rule,
  and $C{\downarrow}$ if $C=\oconf{A}{\Gamma}{s}{\diamond}$.}
\end{definition}



\begin{example}\label{ex:SL-dummy}
  Recall below the equivalent terms $M$ and $N$ from our introductory \cref{ex:dummy}:
  \medskip

\noindent
  \begin{tabular}{@{}l@{\quad}|@{\quad}l@{\quad}|@{\quad}l@{}}
\begin{lstlisting}[boxpos=m]
$M=$ let x = ref 0 in
${\color{white}M=}$ fun f -> f(); !x

$N=$ fun f -> f(); 0
\end{lstlisting}
  &
\begin{lstlisting}[boxpos=m]
$V_l=$ fun f -> f(); !l
$(\text{for location }$l$)$

$V'=$ fun f -> f(); 0
\end{lstlisting}
    &
      \begin{lstlisting}[boxpos=m]
$e_{l,\alpha}=$ $\alpha$(); !l
$(\text{for abstract name }\alpha)$

$e_\alpha'=$ $\alpha$(); 0
\end{lstlisting}
\end{tabular}
\medskip

\noindent
Their stackless LTS's include transitions:
\begin{align*}
&\pconf{\emptyA}{\emptyG}{\emptyS}{M}\transs{\tau}\pconf{\emptyA}{\emptyG}{s_l}{V_l}\trans{\lpropret{\hole}}\oconf{\emptyA}{\Gamma_l}{s_l}{\chi}=C_\chi \text{ with }\Gamma_l=\maps{1}{V_l},\, s_l=\{l\mapsto 0\}\\
&C_\chi\sub{\chi}{\diamond}=  \oconf{\emptyA}{\Gamma_l}{s_l}{\diamond}\trans{\lopapp{1}{\alpha_1}}\pconf{\alpha_1}{\Gamma_l}{s_l}{e_{l,\alpha_1}}\trans{\lpropapp{\alpha_1}{()}}\oconf{\alpha_1}{\Gamma_l}{s_l}{E_l}\text{ with }E_l=\hole;{!l}
  \\
  &\oconf{\alpha_1}{\Gamma_l}{s_l}{E_l}\begin{cases}
\trans{\lopret{()}}\pconf{\alpha_1}{\Gamma_l}{s_l}{E_l[()]}\transs{\tau}\trans{\lpropret{0}}\oconf{\alpha_1}{\Gamma_l}{s_l}{\chi}\\
\trans{\lopapp{1}{\alpha_2}}\pconf{\alpha_1,\alpha_2}{\Gamma_l}{s_l}{e_{l,\alpha_2}}\trans{\lpropapp{\alpha_2}{()}}\oconf{\alpha_1,\alpha_2}{\Gamma_l}{s_l}{E_l}
    \end{cases}\\[2mm]\hline
  &\pconf{\emptyA}{\emptyG}{\emptyS}{N}\transs{\tau}\pconf{\emptyA}{\emptyG}{\emptyS}{V'}\trans{\lpropret{\hole}}\oconf{\emptyA}{\Gamma'}{\emptyS}{\chi}=C_\chi' \text{ with }\Gamma'=\maps{1}{V'}\\
  &C_\chi'\sub{\chi}{\diamond}=  \oconf{\emptyA}{\Gamma'}{\emptyS}{\diamond}\trans{\lopapp{1}{\alpha_1}}\pconf{\alpha_1}{\Gamma'}{\emptyS}{e'_{\alpha_1}}\trans{\lpropapp{\alpha_1}{()}}\oconf{\alpha_1}{\Gamma'}{\emptyS}{E'}\text{ with }E'=\hole;0
  \\
  &\oconf{\alpha_1}{\Gamma'}{\emptyS}{E'}\begin{cases}
\trans{\lopret{()}}\pconf{\alpha_1}{\Gamma'}{\emptyS}{E'[()]}\transs{\tau}\trans{\lpropret{0}}\oconf{\alpha_1}{\Gamma'}{\emptyS}{\chi}\\
\trans{\lopapp{1}{\alpha_2}}\pconf{\alpha_1,\alpha_2}{\Gamma'}{\emptyG}{e'_{\alpha_2}}\trans{\lpropapp{\alpha_2}{()}}\oconf{\alpha_1,\alpha_2}{\Gamma'}{\emptyG}{E'}
    \end{cases}
\end{align*}
Note that the above are not complete descriptions of the LTS's as
we have not explored all possible continuations that can instantiate the abstract continuation $\chi$. As we shall see next, in fact, only a restricted set of relevant instantiations needs to be considered.
  \end{example}

\subsection{Pushdown Normal Form Bisimulation}\label{sec:pdnf-bisim}
Having disengaged stacks from configurations, we need to engineer bisimulations to account for stack discipline during the bisimulation game. 
  Following~\cite{P4F}, we employ an abstraction  which
makes part of \emph{saturation algorithms}~\cite{BouajjaniEM97,FinkelWW97} that finitise pushdown systems with a finite number of control states. We abstract  stacks by so-called \emph{continuation graphs}, which consist of:
\begin{itemize}
  \item \textbf{Vertices ($\beta$, etc.):} these represent pairs of Proponent function entry points, that is, {Proponent} configuration pairs $(C_1,C_2)$ reached after an Opponent call (i.e.\ after a push). We shall write $\beta=\betais{C_1,C_2}$.
\item \textbf{Edges ($\beta'\xrightarrow{\Ed_1,\Ed_2}\beta$):} these are directed and labelled with pairs of continuations, and can be seen as procedure summaries.
  \todo{VK: does the term "procedure summary" come from some particular paper?}
  An edge $\betais{C_1',C_2'}\xrightarrow{(\Ed_1,\Ed_2)}\betais{C_1,C_2}$ means that playing the bisimulation game from $(C_1,C_2)$ we can reach a pair of Opponent configurations with evaluation contexts $\Ed_1,\Ed_2$ respectively from which, in turn, we can fire \textsc{OpCall} transitions and reach $(C_1',C_2')$.
\end{itemize}
Our bisimulation game  will now involve tuples of the form $(C_1,C_2,\Sigma,\beta)$ including a pair of stackless configurations along with a continuation graph $\Sigma$ and an encoding $\beta$ of the pair of entry points that is currently evaluated. 
At each \textsc{OpCall} step in the bisimulation game,
containing transitions $C_i\trans{\lopapp{i}{D[\vec\alpha]}}C_i'$ (for $i=1,2$),  
we shall extend $\Sigma$ by adding a new edge (if not already present) and update the current $\beta$:
\[
(C_1,C_2,\Sigma,\beta)\;  \xmapsto{(\lopapp{i}{D[\vec\alpha]})}\; (C_1',C_2',\Sigma[\beta'\xrightarrow{\Ed_1,\Ed_2}\beta],\beta')
\]
where $\beta'=\betais{C_1',C_2'}$ and $\Ed_i=C_i.\Ed$ (for $i=1,2$).

The above scenario accounts for points in the bisimulation game where a push operation needs to be performed. On the other hand, for pop operations, we need to turn to \textsc{PropRet} steps. Given a current 
tuple $(\hat C_1,\hat C_2,\Sigma,\beta')$ with $\beta'=\betais{C_1',C_2'}$, and assuming that $\hat C_1,\hat C_2$ are about to perform
 $\hat C_i\trans{\lpropret{D}}C_i$ (for $i=1,2$),  
 the set
\[
  \Sigma(\beta') = \{\,(\Ed_1,\Ed_2,\beta)\mid \beta'\xrightarrow{\Ed_1,\Ed_2}\beta\,\}
\]
contains all the push operations that have led to the pair of entry points  $(C_1',C_2')$ that we are currently evaluating. Though only one of them is the operation that has led to the current pair $(\hat C_1,\hat C_2)$, it is sound for the bisimulation game to pop back \emph{any of the operations} in $\Sigma(\beta')$. Thus, we can have:
\[
(\hat C_1,\hat C_2,\Sigma,\beta')\;  \xmapsto{(\lpropret{D})}\; (C_1[\Ed_1/\chi],C_2[\Ed_2/\chi],\Sigma,\beta)
\]
for any $(\Ed_1,\Ed_2,\beta)\in\Sigma(\beta')$.
We next make concrete this high-level presentation by formally introducing continuation graphs and defining the ensuing notion of bisimulation.

  To simplify presentation, we will
abuse notation and utilise $X$ to stand for some object $X$ or $\bot$ (for $X$ an environment $\Gamma$, a store $s$, a continuation $\Ed$ or a stack $K$). The constant $\bot$ denotes a dummy component of a divergent configuration. For continuations and stacks in particular, we extend the push operation by setting:
\[
  \Ed,K = \begin{cases}
    \Ed,K &\text{if {$\Ed\not=\bot$ and $K\not=\bot$} and $\Ed,K$ is defined {according to \cref{sec:lts}}}\\
    \bot &\text{if $\Ed=\bot$}\\
    \text{undefined} &\text{otherwise}
    \end{cases}
\]
Below we write $\Kont$ for the set of continuations $\Ed$, i.e.\
$\Kont=\EC\uplus\{\diamond,\bot\}$.

\begin{definition}\label{def:betaSigma}
  Define \emph{entry points} and \emph{continuation graphs} as follows:
  \begin{align*}
    \mathsf{EPoint} \ni \ \beta &::=\ \diamond\ \mid\ \betais{C_1,C_2}\\
    \mathsf{CGrph}\ni \ \Sigma &\subseteq_{\rm fin.orb.}^{\neq\emptyset} \mathsf{EPoint}\times\Kont\times\Kont\times\mathsf{EPoint}
  \end{align*}
  where $C_1,C_2$ are non-Opponent
  configurations, $\{C_1,C_2\}\neq\{\botcconf\}$, and each $\Sigma$ must satisfy the conditions (note we write $\beta'\toSig{\Ed_1,\Ed_2}\beta$ for $(\beta',\Ed_1,\Ed_2,\beta)\in\Sigma$ and let $\dom{\Sigma}$ contain all such $\beta'$):
  \begin{compactitem}
\item{} \emph{Reachability}.~For all $\beta\in\dom{\Sigma}$ there are evaluation stacks $K_1,K_2$ such that $\beta\toSigs{K_1,K_2}\diamond$, where $\toSigs{}$ is the transitive closure of $\toSig{}$. In particular, $\cdot\toSigs{}\diamond$ is defined inductively by:\vspace{-2mm}
  \[
    \irule{
    }{
      \diamond\toSigs{\noe\,,\noe}\diamond
    }
    \qquad
    \irule{
      \beta\toSigs{K_1,K_2}\diamond
      \\
      \beta'\toSig{\Ed_1,\Ed_2}\beta
    }{
      \beta'\toSigs{(\Ed_1,K_1),(\Ed_2,K_2)}\diamond
    }
\vspace{-2mm}  \]
\item \emph{Top and Divergence}.~For all $\beta'\toSig{\Ed_1,\Ed_2}\beta$ and $j\in\{1,2\}$:
  \begin{align*}
&    \diamond\in\dom{\Sigma}\land (\beta'=\diamond\implies\beta=\diamond) \land
    (\beta=\diamond\iff\Ed_j=\diamond) \\
&{}\land    (\beta'.j=\bot\iff \Ed_j=\bot) \land
(\beta.j=\bot\implies
    \Ed_j=\bot)
  \end{align*}
  where we write $\beta.1=\bot$ just if $\beta=\betais{\botcconf,C}$ (and similarly for $\beta.2=\bot$).
\item \emph{Nominal closure}.~For  all $\beta'\toSig{\Ed_1,\Ed_2}\beta$ and permutations $\pi$, \
$\pi\cdot\beta'\toSig{\pi\cdot\Ed_1,\pi\cdot\Ed_2}\pi\cdot\beta$.
\end{compactitem}
%
\end{definition}

\begin{remark}
Given any continuation graph $\Sigma$, 
the top condition  along with the fact that $\Sigma$ cannot be empty imply that $\Sigma$ contains the loop:
\[
  \Sigma_\diamond = \
\raisebox{-5pt}{  \begin{tikzpicture}
    \node[] (A) at (0,0) {$\diamond$};
    \path[->] (A) edge [loop right, looseness=20] node [right] {$\diamond,\diamond$} (A);
  \end{tikzpicture}}
\]
which is itself a continuation graph.
Note also that the divergence condition ensures that, for any edge $\cdot\toSig{\Ed_1,\Ed_2}\cdot$, we cannot have $\Ed_1=\Ed_2=\bot$.
\end{remark}
\begin{remark}
It is worth commenting on the nominal closure condition. The condition imposes that continuation graphs be closed under permutations so that e.g.\ extending a graph with an edge $\beta'\toSig{\Ed_1,\Ed_2}\beta$ in fact extends it with the whole orbit $\orb{\beta'\toSig{\Ed_1,\Ed_2}\beta}$.
The addition of elements of the orbit is sound and complete as, the behaviour that led $\beta$ to reach $\beta'$ while pushing $(\Ed_1,\Ed_2)$, can also be used by $\pi\cdot\beta$ to reach the corresponding $\pi\cdot\beta'$ while pushing $(\pi\cdot\Ed_1,\pi\cdot\Ed_2)$, for any permutation $\pi$.
The latter allows us to saturate continuation graphs in a finite amount of steps in examples like the ones we saw in \Cref{ex:motiv,ex:dummy}. More foundationally, the saturation of $\Sigma$'s under permutation amounts to treating the pushdown stack (of pairs $(\Ed_1,\Ed_2)$) and its abstraction $\Sigma$ \emph{nominally}, i.e.\ by means of representatives.
In effect, we are working with pushdown nominal automata~\cite{ChengK98,MurawskiRT17}, and that bring about decidability for context-free expressions (\Cref{thm:decide}).
\end{remark}

We next show that $\Sigma$ can only produce valid, compatible stacks (cf.\ \Cref{app:simple-lemmas}).

\begin{lemma}\label{lem:Ks}
  For any $\beta\toSigs{K_1,K_2}\diamond$, we have that $K_1,K_2$ are defined and:
  \begin{compactitem}
    \item $\beta=\diamond$ and $K_1=K_2=\noe$, or
    \item $\beta.1,\beta.2\neq\bot$ and 
      $|K_1|=|K_2|$, or
      \item $\exists j\in\{1,2\}.\ \beta.j=\bot \land K_j=\bot\neq K_{3-j}$
\end{compactitem}
where $|K|$ is the length of $K$ (if $K\neq\bot$). \qed
\end{lemma}

Continuation graphs will be updated in the bisimulation game using the following two operations.
{By definition, continuation graph updates that satisfy the \emph{Top and Divergence} condition produce valid continuation graphs.}
\begin{definition}\label{def:sigma-ext-restr}
  We can \emph{extend} $\Sigma$ with an edge $\beta'\trans{\Ed_1,\Ed_2}\beta$ or \emph{restrict} it to its reachable subgraph starting from $\beta\in\dom{\Sigma}$ as follows (note $\pi\cdot(x\trans{\eta}y)\defeq(\pi\cdot x)\trans{\pi\cdot\eta}(\pi\cdot y)$):
\begin{align*}
  \Sigma[\beta'\toSig{\Ed_1,\Ed_2}\beta] &=  \Sigma\cup\{\pi\cdot(\beta'\trans{\Ed_1,\Ed_2}\beta)\mid \pi\in\Perm\}\\ 
  \Sigma@\beta &= \{\pi\cdot(\beta'\trans{\Ed_1,\Ed_2}\beta'')\mid\pi\in\Perm\land\exists K_1,K_2.\ \beta\toSigs{K_1,K_2}\beta'\}
\end{align*}
      Moreover, entry points and continuation graphs can be left-right inverted as follows:
\begin{align*}
  \diamond^{-1}=\diamond\,,\;\;
  \betais{C_1,C_2}^{-1} = 
\betais{C_2,C_1}\,,\;\;
  \Sigma^{-1} = \{\beta'^{-1}\trans{\Ed_2,\Ed_1}\beta^{-1}\mid \beta'\toSig{\Ed_1,\Ed_2}\beta\}.
\end{align*}
\end{definition}
Bisimulations for stackless configurations will involve tuples of the form defined next.

\begin{definition}[Compatible (bi)simulation tuples] \label{def:compatible-conf}\label{def:compatible}
  Let us call configurations $C_1,C_2$ \emph{compatible} whenever $\{\botcconf\}\subsetneq\{C_1,C_2\}$, or
  $C_1,C_2$ have the same polarity,
    $\dom{C_1.\Gamma}=
    \dom{C_2.\Gamma}$ and $(C_1.\Ed=\diamond\iff C_2.\Ed=\diamond)
    $.\\
Tuple $(C_1,C_2,\Sigma,\beta)$ is  {compatible} if
  $C_1,C_2$ compatible, 
    $\Sigma@\beta=\Sigma$ and
  for $j\in\{1,2\}$:
  \begin{compactenum}
  \item[($\bot$)] if $\beta.j=\bot$ then $C_j=\botcconf$;
    \item[($\diamond$)] if $\beta=\diamond$ then $C_j.\Ed=\diamond$ or $\an{C_j.e}=\emptyset$; and if $C_j.\Ed=\diamond$ then $\beta=\diamond$.
    \end{compactenum}
\end{definition}

Condition~($\diamond$) above says that a top-level $\beta$ ($\diamond$) is only allowed in opponent configurations with top-level continuation ($\diamond$) and in initial proponent configurations evaluating the top-level term (e.g.\ $\pconf{\emptyA}{\emptyG}{\emptyS}{M}$ in \cref{ex:SL-dummy}); dually, any top-level continuation ($\diamond$) needs a top-level $\beta$ ($\diamond$).
We now give the definition of (bi)simulation for the stackless LTS.

\begin{definition}[PDNF Bisimulation]\label{def:SLbisim}
  A relation \bisim{R} with elements of the form $(C_1,C_2,\Sigma,\beta)$, and membership thereof denoted $C_1 \bisim*{R}_{\Sigma,\beta} C_2$, is called 
  \emph{weak simulation} when
  for all $C_1 \bisim*{R}_{\Sigma,\beta} C_2$ we have $(C_1,C_2,\Sigma,\beta)$ compatible
 and:
  \begin{enumerate}
    \item[0.]
      if $C_1\downarrow$
      then $C_2\downarrow$
    \item
      if $C_1 \realtrans{\lopret{D[\vec\alpha]}} C_1'$ with $\vec\alpha\fresh C_2,\beta$ 
      then $C_2 \trans{\lopret{D[\vec\alpha]}} C_2'$ such that
      $C_1' \bisim*{R}_{\Sigma,\beta} C_2'$
    \item
      if $C_1 \realtrans{\eta} C_1'$
      then $C_2 \wtrans{\eta} C_2'$ and
      $C_1' \bisim*{R}_{\Sigma,\beta} C_2'$, for $\eta\in\{\tau,\lpropapp{\alpha}{D}\}$
    \item
      if $C_1 \realtrans{\lopapp{i}{D\hole[\vec \alpha]}}  C_1'$ with
 $\vec\alpha\fresh C_2,\beta$ then $C_2 \trans{\lopapp{i}{D\hole[\vec\alpha]}}  C_2'$
such that      $C_1' \bisim*{R}_{\Sigma',\beta'} C_2'$
      with $\beta' = \betais{C_1',C_2'}$ and $\Sigma' = \Sigma[\beta'\trans{C_1.\Ed, C_2.\Ed}\beta]$
    \item
      if $C_1 {\realtrans{\lpropret{D}}} C_1'$
      and
      $\beta{\toSig{\Edb_1,\Edb_2}}\beta'$
      then $C_2 {\wtrans{\lpropret{D}}} C_2'$
      and
      $C_1'\sub{\chi}{\Edb_1} \bisim*{R}_{\Sigma@\beta',\beta'} C_2'\sub{\chi}{\Edb_2}$.
    \end{enumerate}
  Similarity $(\simil)$ is the largest weak simulation.
  Relation \bisim{R} is a \emph{weak bisimulation} when \bisim{R} and $\bisim{R}^{-1}$ are weak simulations, where $R^{-1}=\{(C_2,C_1,\Sigma^{-1},\beta^{-1})\mid (C_1,C_2,\Sigma,\beta)\in R\}$.
  Bisimilarity $(\bisimil)$ is the largest weak bisimulation.
    Expressions $\vdash e_1,e_2: T$ 
  with {$\an{e_i}=\emptyset$}
are \emph{PDNF bisimilar}, written $e_1 \bisimilPD e_2$,
when 
  $\pconf{\emptyA}{\emptyG}{\emptyS}{e_1} \bisimil_{\Sigma_\diamond,\diamond}
  \pconf{\emptyA}{\emptyG}{\emptyS}{e_2}$.
\end{definition}

\begin{remark}
  The definition above assumes that $\Sigma$ is well-defined, i.e.\ it satisfies the conditions of \Cref{def:betaSigma} and tuples satisfy the compatibility conditions of \Cref{def:compatible}. These conditions are preserved by the bisimulation, thus making them merely initial conditions for the construction of the relations.
\\
  In particular $\Sigma$ is well-defined when extending $\Sigma$ (case~3) as $\beta'$ is well-defined and:
  \begin{compactitem}
  \item Reachability and nominal closure are preserved by construction.
  \item Top follows from the fact that  $C_1.\Ed=\diamond$ iff $\beta=\diamond$ (by definition) and $(C_1.\Ed=\diamond\land C_2.\Ed\neq\diamond)$ is not possible due to compatibility.
  \item For divergence, given that $C_1$ does not diverge, it suffices to check the conditions for $j=2$. For the first one, we need to verify that we cannot have $C_2'=\botcconf\not=C_2$, which is indeed the case as $C_1$ can make the move $\lopapp{i}{D[\vec\alpha]}$ and $C_1,C_2$ are compatible.
For the second condition, if $\beta.2=\bot$ then by compatibility we have $C_2=\botcconf$, and therefore $C_2.\Ed=\bot$.
\end{compactitem}
On the other hand, by restricting $\Sigma$ (case~4) we get a smaller graph with all entry points reachable from $\beta'$, and its validity follows from the validity of $\Sigma$.\\
Furthermore,  compatibility is preserved in the target tuple in each case. For
each $j$,
 if $C_j=\botcconf$ then $C_j'=\botcconf$ (this is vacuously true for $j=1$).
Moreover, if $\beta=\diamond$ then:
  \begin{compactitem}
  \item if $C_j$ an opponent configuration then $C_j.\Ed=\diamond$ and  $\lopret{D[\vec\alpha]}$ is not possible;
    \item if $C_j$ is proponent then $\an{C_j.e}=\emptyset$ and $\lpropapp{\alpha}{D}$ is not possible;
    \end{compactitem}
        thus, in either case, $\beta$ is replaced by $\beta'$ and the validity of $\Sigma$ implies validity of the resulting tuple involving $C_1',C_2'$. Finally, if case~4 takes place and $C_j'\sub\chi{\Ed_j}.\Ed=\diamond$ then $\Ed_j=\diamond$ and therefore $\beta'=\diamond$.
\end{remark}


\subsection{Decidability of PDNF Bisimulation}\label{sec:decidability}

We previously mentioned  that PDNF bisimulation can be used to decide equivalence of context-free terms {(\cref{def:cxt-free-term})}. The first step in proving this is to show {bounded} PDNF bisimilarity decidable.

Given $k\in\mathbb N$ and edge $\beta'\trans{\Ed_1,\Ed_2}\beta$ (of $\Sigma$), we say that the edge is \emph{$k$-bounded} if:
\[ \max(\|\beta'\|,\|\Ed_1\|,\|\Ed_2\|,\|\beta\|)\leq k,\;\text{where }\|\betais{C_1,C_2}\|=\max(\|C_1\|,\|C_2\|)\text{ and }\|\diamond\|=1.
\]
Accordingly, 
tuple $(C_1,C_2,\Sigma,\beta)$ is {$k$-bounded} if all elements of $\Sigma$ are $k$-bounded and
$ \max(\|C_1\|,\|C_2\|,\|\beta\|)\leq k$. Finally, candidate weak bisimulation relation $\bisim R$ is $k$-bounded if all its elements are. Note below that $\pi\cdot\bisim R=\{\pi\cdot x\mid x\in\bisim R\}$.

\begin{lemma}
  If $\bisim R$ is $k$-bounded and equivariant then it is orbit-finite.
\end{lemma}
\begin{proof}
  Note first that, for each $k\in\mathbb N$, the set of configurations with size at most $k$
  is orbit-finite. Similarly for the set 
  of continuations with size at most $k$. Accordingly, since (in nominal sets~\cite{PittsNominal}) orbit-finiteness is closed under cartesian products and equivariant subsets~\cite{BojanczykKL14}, the set of all $k$-bounded edges is also orbit-finite.
  As orbit-finiteness is also closed under equivariant powerset, the set of all $k$-bounded $\Sigma$'s is also orbit-finite. Since orbit-finiteness is closed under cartesian products, disjoint union and equivariant subsets, any $k$-bounded equivariant candidate weak bisimulation is orbit-finite. \qed
\end{proof}
\begin{theorem}\label{thm:decide}
  Given $e_1,e_2$ context-free {according to \cref{def:cxt-free-term}} with bound $k$,
  $e_1\, {\bisimilPD}\, e_2$ iff
  $\pconf{\emptyA}{\emptyG}{\emptyS}{e_1} \bisim R_{\Sigma_\diamond,\diamond}
  \pconf{\emptyA}{\emptyG}{\emptyS}{e_2}$
  for some $k$-bounded equivariant weak bisimulation $\bisim R$. Therefore, $\bisimilPD$ is decidable for context-free expressions.
\end{theorem}
\begin{proof}
  We observe that, since $e_1,e_2$ are context-free with bound $k$, in the bisimulation game starting from $(\pconf{\emptyA}{\emptyG}{\emptyS}{e_1},\pconf{\emptyA}{\emptyG}{\emptyS}{e_2},\Sigma_\diamond,\diamond)$ we only reach $k$-bounded tuples $(C_1,C_2,\Sigma,\beta)$. This is due to the fact that $C_1,C_2$ , all configurations contained in $\beta$ and in the vertices of $\Sigma$, and all continuations found in edges of $\Sigma$ are all sourced from components of (stacked) configurations found in the set:
  \[
    \{C\mid \exists i,t.\ \oldpconf{\emptyA}{\emptyG}{\emptyK}{\emptyS}{e_i}\trans{t}C\}
  \]
  which, by assumption, have size at most $k$. Thus,
  if $e_1  \bisimilPD e_2$ then 
  $\pconf{\emptyA}{\emptyG}{\emptyS}{e_1} \bisim R_{\Sigma_\diamond,\diamond}
  \pconf{\emptyA}{\emptyG}{\emptyS}{e_2}$ with $\bisim R$ being the restriction of $\bisimilPD$ to $k$-bounded elements. \\
  Now, given context-free expressions $e_1,e_2$ with $\vdash e_1,e_2:T$ and $\an{e_1,e_2}=\emptyset$, to decide whether $e_1 \bisimilPD e_2$ we pick a bound $k\in\mathbb{N}$ and try the weak bisimulation conditions on all $k$-bounded equivariant candidate weak bisimulation relations $\bisim R$.
  The examination of all such relations is possible due to orbit-finiteness.
  If in our examination we are led to a tuple $(C_1,C_2,\Sigma,\beta)$ that is not $k$-bounded, we restart with $k=k+1$.
  If a weak bisimulation is found, we Accept. If no weak bisimulation is found, we Reject.
\qed
\end{proof}

\subsubsection*{Example revisited}
We now show how PDNF bisimulation applies to the example from \cref{sec:motivating-example}. To simplify presentation, and relying on determinacy and a simple up-to beta reduction technique (see \cref{app:simple}), we shall restrict our attention to bisimulation challenges of the form:
\[
  C_1 \transs{\tau}\trans{\eta} C_1'
\]
where $\eta\neq\tau$, and only present the part of the bisimulation containing the corresponding configurations at the beginning and end of such transition sequences.

\begin{example}\label{ex:revisited}
  Recall again the equivalent terms $M$ and $N$ from \cref{ex:dummy}, as well as their reducts $V_l,V',e_{l,\alpha},e'_\alpha$ from \cref{ex:SL-dummy}. We construct a PDNF bisimulation $\bisim R$ by including:
  \begin{align*}
    \bisim R_0 &=\{(\pconf{\emptyA}{\emptyG}{\emptyS}{M},\pconf{\emptyA}{\emptyG}{\emptyS}{N},\Sigma_\diamond,\diamond)\}
\cup      
\orb{\{(\oconf{\emptyA}{\Gamma_l}{s_l}{\diamond},\oconf{\emptyA}{\Gamma'}{\emptyS}{\diamond},\Sigma_\diamond,\diamond)\}}\\
    \bisim R_1 &= \orb{\{(\pconf{A\uplus\{\alpha\}}{\Gamma_l}{s_l}{e_{l,\alpha}},\pconf{A\uplus\{\alpha\}}{\Gamma'}{\emptyS}{e'_\alpha},{\Sigma_1,\beta_{l,\alpha}}),(\oconf{A\uplus\{\alpha\}}{\Gamma_l}{s_l}{E_l},\oconf{A\uplus\{\alpha\}}{\Gamma'}{\emptyS}{E'},{\Sigma_1,\beta_{l,\alpha}})\}}\\
        \bisim R_2 &= \orb{\{(\pconf{A\uplus\{\alpha\}}{\Gamma_l}{s_l}{e_{l,\alpha'}},\pconf{A\uplus\{\alpha\}}{\Gamma'}{\emptyS}{e'_{\alpha'}},{\Sigma_2,\beta_{l,\alpha'}}),(\oconf{A\uplus\{\alpha\}}{\Gamma_l}{s_l}{E_l},\oconf{A\uplus\{\alpha\}}{\Gamma'}{\emptyS}{E'},{\Sigma_2,\beta_{l,\alpha'}}) \}}
  \end{align*}
for some $l\in\Loc$ and $\alpha\neq\alpha'\in\Nam$,  with
$s_l=\{l\mapsto 0\}$,
  $\beta_{l,\alpha}=\betais{(\Gamma_{l},s_l,e_{l,\alpha}),(\Gamma',\emptyS,e'_\alpha)}$ and:
\[
  \begin{aligned}
    \Gamma_l &=\maps{1}{V_l} & E_l&=\hole;{!l} & \Sigma_1&=\Sigma_\diamond[\beta_{l,\alpha}\trans{\diamond,\diamond}\diamond]\\
     \Gamma'&=\maps{1}{V'} & E'& =\hole;0 &\Sigma_2&=\Sigma_1[\beta_{l,\alpha'}\trans{E_l,E'}\beta_{l,\alpha}]
   \end{aligned}\left(
\begin{aligned}
  V_l &=  \lambda f.\ f(); !l
&e_{l,\alpha} &=\alpha(); !l
  \\
V' &= \lambda f. f(); 0
&
e_\alpha' &=\alpha(); 0
\end{aligned}\right)
\]
and taking $\bisim R=\bisim R_1\cup\bisim R_2\cup\bisim R_2$. Hence, $M \bisimilPD N$.
  \end{example}

  Note that in the example above the relation we build is infinite, due to the accumulation of edges e.g.\ in $\Sigma_2$:
  \[
    \cdots \beta_{l,\alpha_i}\trans{E_l,E'}\cdots
    \beta_{l,\alpha_3}\trans{E_l,E'}
    \beta_{l,\alpha_2}\trans{E_l,E'}
    \beta_{l,\alpha_1}\trans{\diamond,\diamond}\diamond
    \!\!\!\raisebox{-4pt}{  \begin{tikzpicture}
      \node[] (A) at (0,0) {$\text{ }$};
      \path[->] (A) edge [loop right, looseness=20] node [right] {$\diamond,\diamond$} (A);
    \end{tikzpicture}}
  \]
  Nonetheless, $\Sigma_2$ is {orbit-finite}, as it is the closure under permutation of this finite graph:
  \[
    \beta_{l,\alpha_2}\trans{E_l,E'}
    \beta_{l,\alpha_1}\trans{\diamond,\diamond}\diamond
    \!\!\!\raisebox{-4pt}{  \begin{tikzpicture}
      \node[] (A) at (0,0) {$\text{ }$};
      \path[->] (A) edge [loop right, looseness=20] node [right] {$\diamond,\diamond$} (A);
    \end{tikzpicture}}
  \]
  In fact, as seen by its definition, the bisimulation that we built above is also orbit-finite.

\subsection{Soundness}\label{sec:soundness}

\newcommand\toS[1]{\widetilde{(#1)}}
\newcommand\tC{\tilde{C}}
\newcommand{\tsR}{\ma{\mathrel{\tilde{\mathcal{R}}}}}
\newcommand{\tR}{\ma{\tilde{\mathcal{R}}}}

We next show that PDNF bisimulation is sound with respected to (standard) NF bisimulation. We will prove that if stackless configurations $C_1,C_2$ are related by a weak bisimulation $\bisim R$ then we can construct a weak bisimulation $\tR$ on stacked configurations containing $C_1,C_2$ with appropriate stacks attached.

Recall in \Cref{def:lts} the {stacked} LTS.
For notational convenience,
in this and the next section, 
we shall denote the configurations of the stacked LTS by $\tilde C$ and variants (and $C$ and variants is reserved for stackless). We move from stackless to stacked configurations by adding compatible stack components.

\begin{definition}
Given a configuration $C$ and stack $K$, we set:
\[
  \toS{C,K} = \begin{cases}
    \oldpconf{A}{\Gamma}{K}{s}{e} &\text{if }C=\pconf{A}{\Gamma}{s}{e}\land K\not=\bot\\
    \oldoconf{A}{\Gamma}{K}{s}{\Ed} &\text{if }C=\oconf{A}{\Gamma}{s}{\Ed}\land K\not=\bot 
    \land(\Ed=\diamond\iff K=\noe)\\
    \botconf &\text{if }C=\botcconf\\
    \text{undefined} &\text{otherwise}
    \end{cases}
  \]
\end{definition}


\begin{lemma}\label{lem:toS}
  For any $\beta\toSigsD{K_1,K_2}$ and $C_1,C_2$, if $\beta=\betais{C_1,C_2}$ or $(C_1,C_2,\Sigma,\beta)$ compatible
  then $\toS{C_1,K_1},\toS{C_2,K_2}$ are defined.
\qed\end{lemma}

Soundness can be shown by the following result (cf.~\Cref{app:simple-lemmas}).

\begin{lemma}[Soundness]\label{lem:soundness}
  If $\bisim R$ is a weak simulation then so is:
  \begin{align*}
     \tR = \{
           (\toS{C_1,K_1},\toS{C_2,K_2}) \mid&\; 
                                                         \exists\Sigma,\beta.\
  C_1\bisim*{R}_{\Sigma,\beta}C_2
\land
                                                         \beta\toSigs{K_1,K_2}\diamond
\}
  \end{align*}
  Moreover, if $\bisim R$ is a weak bisimulation then so is $\tR$.
\qed\end{lemma}

\subsection{Completeness}\label{sec:completeness}
In order to derive a pushdown bisimulation $\bisim R$ from a (standard) bisimulation $\tR$, we shall 
define an LTS that follows the bisimulation game in the stackless LTS but is faithful to the stack discipline.
 
 \begin{definition}\label{def:sat}
   The  \emph{saturated simulation LTS} contains transitions of the form:
  \[
    (C_1,C_2)\satrans{\eta}(C_1',C_2') 
  \]
  where
$C_1,C_2$ compatible and
  $\eta$ is either $\varepsilon$ or a pair $(\Edb_1,\Edb_2)$. 
  The rules for $\satrans{}$ are given in \cref{fig:Xlts} (note use of stackless LTS). 
We call a continuation graph $\Sigma$ \emph{sat-connected} if whenever $\betais{C_1',C_2'}\toSig{\Ed_1,\Ed_2}\betais{C_1,C_2}$ then
$
 (C_1,C_2)\satrans{\varepsilon}\cdot\satrans{(\Ed_1,\Ed_2)}(C_1', C_2').
$
\end{definition}

\begin{remark}
The LTS defined above is an adaptation of the saturation procedure for pushdown systems presented in~\cite{FinkelWW97}.
The intended meaning of
$    (C_1,C_2)\satrans{\eta}(C_1',C_2') $
is that, using synchronisation on visible moves and 
the stacked LTS, $(\toS{C_1,\noe},\toS{C_2,\noe})$ reduces to $(\toS{C_1',K_1},\toS{C_2',K_2})$  following the simulation game and:
  \begin{itemize}
  \item if $\eta=\varepsilon$ then $K_1=K_2=\noe$;
  \item if $\eta=(\Ed_1,\Ed_2)$ then $(K_1,K_2)=(\Ed_1,\Ed_2)$.
  \end{itemize}
\end{remark}

  
Standard (stacked) similarity is  closed under transitions in the saturated LTS.

\begin{figure*}[t] 
  \[\begin{array}{cc}
      \irule[Tau][tau]{
      C_i\realtrans{\tau}C_i'
      \\
      C_{3-i}=C_{3-i}'
      }{
      (C_1,C_2)
      \satrans{\varepsilon} 
      (C_1',C_2')
      }
      &
            \irule[PropCall][propapp]{
           C_1\realtrans{\lpropapp{\alpha}{D}}C_1'
      \\
      C_2\trans{\lpropapp{\alpha}{D}}C_2'
      }{
      (C_1,C_2)
      \satrans{\varepsilon} 
      (C_1',C_2')
           }
      \\[7mm]
        \irule[OpRet][opret]{
      C_1\realtrans{\lopret{D\hole[\vec\alpha]}}C_1'
      \;\;
      C_2\trans{\lopret{D\hole[\vec\alpha]}}C_2'
      }{
      (C_1,C_2)
      \satrans{\varepsilon} 
      (C_1',C_2')
        }&
              \irule[OpCall][opapp]{%
      C_1\realtrans{\lopapp{j}{D\hole[\vec\alpha]}}C_i'
           \;\;
           C_2\trans{\lopapp{j}{D\hole[\vec\alpha]}}C_2'
           }{
      (C_1,C_2)
      \satrans{(C_1.\Ed,C_2.\Ed)} 
      (C_1',C_2')
        }
      \\[7mm]
      \multicolumn{2}{c}{%
    \irule[PropRet][propret]{
      (C_1,C_2)
      \satrans{(\Edb_1,\Edb_2)} 
      (C_1',C_2')
      \satrans{\varepsilon} 
      (C_1'',C_2'')
      \\
      C_1''\realtrans{\lpropret{D}}C_1'''
      \\
      C_2''\trans{\lpropret{D}}C_2'''
      }{
      (C_1,C_2)
      \satrans{\varepsilon} 
      (C_1'''\sub{\chi}{\Ed_1},C_2'''\sub{\chi}{\Ed_2})
      }}
      \\[7mm]
      \irule[Refl][refl]{
      }{
      (C_1,C_2)
      \satrans{\varepsilon} 
      (C_1,C_2)
      }
      &
      \irule[Trans][trans]{
      (C_1,C_2)
      \satrans{\varepsilon} 
      (C_1',C_2')
      \satrans{\varepsilon} 
      (C_1'',C_2'')
      }{
      (C_1,C_2)
      \satrans{\varepsilon} 
        (C_1'',C_2'')
      }
   \end{array}\]
  \hrule
  \caption{The Saturated Simulation Labelled Transition System.}\label{fig:Xlts}
\end{figure*}


\begin{lemma}\label{lem:sat}
  Given $\toS{C_1,K_1}\simil \toS{C_2,K_2}$ and
  $(C_1,C_2)\satrans{\eta}(C_1',C_2')$:
  \begin{compactitem}
  \item if $\eta=\varepsilon$ then $\toS{C_1',K_1}\simil \toS{C_2',K_2}$,
      \item if $\eta=(\Ed_1,\Ed_2)$ then $\toS{C_1',(\Ed_1,K_1)}\simil \toS{C_2',(\Ed_2,K_2)}$.
  \end{compactitem}
\end{lemma}
\begin{proof}
  We use rule induction.
The case of \textsc{Refl} is trivial, while that of \textsc{Trans} follows directly from induction hypothesis. 
For rule \textsc{Tau} we use \Cref{lem:closures}. For \textsc{PropCall} we  use determinacy of the stacked LTS. For \textsc{OpCall, OpRet}, suppose $C_i\trans{\eta}C_i'$ for $i=1,2$ and $\eta$ having abstract names $\vec\alpha\fresh C_1,C_2$.
Then, by hypothesis and determinacy,  $\tC_1'\simil\tC_2'$.
For \textsc{PropRet}, the induction hypothesis gives us $\toS{C_1'',(\Ed_1,K_1)}\simil \toS{C_2'',(\Ed_2,K_2)}$. Combining this with hypotheses
      $C_1''\realtrans{\lpropret{D}}C_1'''$ and
$C_2''\trans{\lpropret{D}} C_2'''$ (and using the fact that $C_2''$ has no other transitions),
we obtain $\toS{C_1'''\sub{\chi}{\Ed_1},K_1}\simil \toS{C_2'''\sub{\chi}{\Ed_2},K_2}$.
\qed
\end{proof}

The main result is the following (cf.\ \Cref{app:completeness}). Note $\simil$ is standard (stacked) similarity.

\begin{lemma}\label{lem:completeness}
The following is a weak (pushdown) simulation:
  \begin{align}
    \bisim{R} = \{\ &(C_1,C_2,\Sigma,\beta) \mid
\Sigma\text{ sat-connected}\land
                 (C_1,C_2,\Sigma,\beta)\text{ compatible}\notag\\
                    &{}\land\forall K_i.\ \beta\toSigsD{K_1,K_2}
                      \implies
                                                \toS{C_1,K_1}\simil\toS{C_2,K_2}\tag{A}\label{cond:a}\\
              & {}\land\forall C_i',K_i.\  \betais{C'_1,C'_2}\toSigsD{K_1,K_2}
                \implies
                \toS{C_1',K_1}\simil\toS{C_2',K_2}\tag{A$^*$}\label{cond:a1}\\
              & {}\land\beta\neq\diamond\implies\exists C_1',C_2'.\ \beta=\betais{C_1',C_2'}\land 
 (C_1',C_2')\satrans{\varepsilon}(C_1,C_2)\tag{B}\label{cond:b}
\    \}\\[-11mm]\notag
  \end{align}\qed
\end{lemma}

We can now prove full abstraction.   Recall $\Sigma_\diamond=\{(\diamond,\diamond,\diamond,\diamond)\}$.

  \begin{theorem}\label{thm:complete}
  $\oldpconf{\emptyA}{\emptyG}{\emptyK}{\emptyS}{e_1} \simil \oldpconf{\emptyA}{\emptyG}{\emptyK}{\emptyS}{e_2}$
  iff 
  $\pconf{\emptyA}{\emptyG}{\emptyS}{e_1} \simil_{\Sigma_\diamond,\diamond} \pconf{\emptyA}{\emptyG}{\emptyS}{e_2}$.
\end{theorem}
\begin{proof}
We let $C_i=\pconf{\emptyA}{\emptyG}{\emptyS}{e_i}$ and $\tC_i=\oldpconf{\emptyA}{\emptyG}{\emptyK}{\emptyS}{e_i}$, for $i=1,2$.
  Note first that $\toS{C_i,\noe}=\tC_i$.
  The right-to-left direction follows from \cref{lem:soundness}. For the converse, suppose  $\tC_1\simil\tC_2$. By  \cref{lem:completeness} there is weak simulation $\bisim{R}$ defined as in the lemma. We claim that $C_1\bisim*R_{\Sigma_\diamond,\diamond}C_2$. We have that $\Sigma_\diamond$ is sat-connected, $(C_1,C_2,\Sigma_\diamond,\diamond)$ is compatible, conditions \cref{cond:a1} and \cref{cond:b} are  vacuously true, while 
 \cref{cond:a} is simplified to $\tC_1\simil\tC_2$. Thus, $C_1\simil_{\Sigma_\diamond,\diamond} C_2$.\qed
\end{proof}
\begin{corollary}
  Contextual equivalence is decidable for context-free expressions.
\end{corollary}
\begin{proof}
  Follows from \Cref{thm:decide,thm:complete}.\qed
\end{proof}

  \section{Up-to Techniques}
  \label{sec:up-to}
  
PDNF bisimulation supports the standard techniques: 
up to identity, up to garbage collection, up to beta reductions and up to name permutations
(see Appendix~\ref{app:simple}).
Here we present an \emph{up to name reuse}, which is important for finitising examples such as those in \cref{sec:motivating-example}, and a redesigned \emph{up to separation} technique from \cite{KoutavasLT22}, which is effective in finitising the bisimulation game of many examples.
We develop our up-to techniques using the theory of bisimulation enhancements from \cite{PousS11,PousCompanion}, which is based on \emph{weak progression}, summarised below.


\begin{definition} \label{def:wp} 
  We write $\WP(R)$ for the monotone functional derived from \cref{def:SLbisim}.
%
\end{definition} 
\begin{definition}[Progressions ($\progress$)]\label{def:progression} 
  ~
  \begin{itemize}
    \item 
      $\bisim{R}$ \emph{weakly progresses} to $\bisim{S}$, and we write 
      ${\bisim{R}}\progress[\WP]{\bisim{S}}$ when ${\bisim{R}}\subseteq\WP(\bisim{S})$.

    \item For monotone functions 
      $f, g$ 
    we write $f \progress[\WP] g$ 
      when $f\comp\WP \sqsubseteq \WP\comp g$.
\defqed
  \end{itemize}
\end{definition} 

\begin{lemma}
  \bisim{R} is a weak simulation when ${\bisim{R}}\progress[\WP]{\bisim{R}}$.
  Also, $({\simil}) = (\gfp\WP)$.
  \qed
\end{lemma}
The following gives the definition of an up-to technique, what it means to be sound, and the stronger notion of compatibility.
\begin{definition}~
  \begin{itemize}
    \item\emph{Simulation up-to:}
      \bisim{R} is a \emph{weak simulation up to $f$} when ${\bisim{R}} \progress[\WP] f(\bisim{R})$.

    \item\emph{Sound up-to technique:}
      Function $f$ is \emph{$\WP$-sound} when $\gfp{\WP\comp f} \subseteq \gfp\WP$.

  \end{itemize}
\end{definition}

\begin{lemma}[\cite{PousS11}, Thm. 6.3.9] \label{lem:compat-sound} 
  If $f\progress[\WP]f$ then it is $\WP$-sound.
  \qed
\end{lemma}

\subsection{Up to Name Reuse}
\label{sec:up-to-nr}

In this section we define an up-to technique that allows us to reuse a single abstract function name in all opponent calls to a function in the knowledge environment, provided that the function does not contain any higher-order references.
In such cases, it is guaranteed that the function being called will not contain names from past calls.
We first define a name substitution that is only defined under these conditions.

\begin{definition}[$\mathsf{nr}$-Substitution]
  The partial substitution operation $(\cdot)\subnr{\alpha}{\alpha'}$ is defined only for terms that do not contain $\alpha$, or have no higher-order references and no occurrences of $\alpha'$. Note below $\Phi$ ranges over $\Ed$ and $e$.
  \begin{gather*}
  \begin{array}{@{}r@{\;}l@{}}
    \Phi\subnr{\alpha}{\alpha'}      & \defeq \left\{
      \begin{array}{l@{\quad\,}l}
        \Phi\sub{\alpha}{\alpha'} & \text{if $\alpha'\fresh\Phi$ and $\Phi$ contains no HO references} \\
        \Phi                      & \text{if $\alpha\fresh\Phi$}
      \end{array}
      \right. 
   \\[1ex]
   \Gamma\subnr{\alpha}{\alpha'}(i) & \defeq 
     (\Gamma(i))\subnr{\alpha}{\alpha'} 
   \\[1ex]
   C\subnr{\alpha}{\alpha'} &\defeq 
\cconf{C.A}{\Gamma\subnr{\alpha}{\alpha'}}{s}{\Phi\subnr{\alpha}{\alpha'}}
 \qquad\text{if }
C = \cconf{C.A}{\Gamma}{s}{\Phi}
   \\[1ex]
   \beta\subnr{\alpha}{\alpha'} & \defeq 
      \betais{C_1\subnr{\alpha}{\alpha'}, C_2\subnr{\alpha}{\alpha'}}
   \\[1ex]
    \beta\subnr{\alpha}{\alpha'} &\toSig[\Sigma\subnr{\alpha}{\alpha'}]{\Ed_1\subnr{\alpha}{\alpha'}, \Ed_2\subnr{\alpha}{\alpha'}} \beta'\subnr{\alpha}{\alpha'} 
    \qquad\quad\text{if }
    \beta\toSig{\Ed_1, \Ed_2}\beta'
  \end{array}
  \end{gather*}
\end{definition}

\begin{definition}[Up to Name Reuse]
  The function $\utnr{}$ on relations is defined as:
  \begin{gather*}
    C_1 \utnr*{\bisim{R}}_{\Sigma,\beta} C_2
    \quad\text{when}\quad
    \exists\alpha,\alpha'.\,C\subnr{\alpha}{\alpha'} \bisim*{R}_{\Sigma\subnr{\alpha}{\alpha'}@\beta\subnr{\alpha}{\alpha'},~\beta\subnr{\alpha}{\alpha'}} C_2\subnr{\alpha}{\alpha'}
  \end{gather*}
\end{definition}

This technique is useful when opponent applies the same higher-order function more than once, e.g. to names $\alpha_1$, $\alpha_2$, $\alpha_3$, etc. Immediately after the calls with arguments $\alpha_i$, $i > 1$, we can apply the substitution $\subnr{\alpha_i}{\alpha_1}$ and prove (bi)simulation of the resulting configurations, effectively using the same opponent name on all calls to the same function. 
  The following lemma shows that (bi)simulation shown after applying such a substitution implies (bi)simulation of the configurations before applying the substitution.

\begin{lemma}\label{lem:nrsound}
  Function $\utnr{}$ is a sound up-to technique.
\end{lemma}
\begin{proof}
  We prove this by showing that $\utnr{} \progress[\WP] \utnr{}$, that is $\utnr{}\comp\WP(\bisim{R}) \subseteq \WP\comp \utnr{\bisim{R}}$ unfolding the definition of $\WP$.
  Note that we only need to prove this for a single substitution $\subnr{\alpha}{\alpha'}$.
  Proponent calls and returns may extend the knowledge environments with function containing the substitution, resulting in configurations captured by $\utnr{\bisim{R}}$.
  Proponent returns in particular involves showing that graph reachability $\Sigma@\beta$ is invariant to name substitution.
  Opponent returns from configurations in $\utnr{}\comp\WP(\bisim{R})$ will produce the same configurations as the same transitions from $\WP(\bisim{R})$, modulo the single substitution $\subnr{\alpha}{\alpha'}$.
  Opponent calls are a bit more involved as they extend the call graph with all permutations of the new edge of the call graph; this however is captured by $\Sigma\subnr{\alpha}{\alpha'}@\beta\subnr{\alpha}{\alpha'}$ in the above definition.
  Reductions preserve the conditions of the substitution and termination and compatibility are unaffected by it.\qed
\end{proof}

The lemma below, proven similarly to \Cref{lem:nrsound}, shows that up to name reuse is a complete technique. Namely, if after applying a name substitution the (bi)simulation conditions are broken then the configurations before the substitution are inequivalent.

\begin{lemma}
  The function $\utnrinv{}$ defined below is a sound up-to technique.
  \begin{gather*}
    C\subnr{\alpha}{\alpha'} \utnrinv*{\bisim{R}}_{\Sigma\subnr{\alpha}{\alpha'}@\beta\subnr{\alpha}{\alpha'},\beta\subnr{\alpha}{\alpha'}} C_2\subnr{\alpha}{\alpha'}
    \;\;\text{when}\;\;
    C_1 \bisim*{R}_{\Sigma,\beta} C_2\\[-11mm]
  \end{gather*}\qed
\end{lemma}


\subsection{Up to Separation}
  \label{sec:up-to-sep}

  We next develop an adaptation of up to separation from~\cite{KoutavasLT22}. This is an effective technique for reducing the state-space of the bisimulation exploration in our verification tool.
The intuition of this technique is that if different functions operate on disjoint parts of the store, they can be explored by bisimulation independently, removing interleaving of their calls. In cases where a function does not contain free locations, the effect of this technique is to allow bisimulation to apply it only once, as two copies of the function will not interfere with each other, even if the create new locations when run.

To define up to separation we need a separating conjunction for configurations.

\begin{definition}[Separating Conjuction]
  We define the partial function $\sepconj$ on stores and knowledge environments as:
  \begin{gather*}
    s_1 \sepconj[L] s_2 = s_1,s_2
    \qquad
    \Gamma_1 \sepconj[I,L] \Gamma_2 = \Gamma_1,\Gamma_2
  \end{gather*}
     when $\dom{s_1}\cap\dom{s_2} = \emptyset$ and $\fl{s_2}\subseteq\dom{s_2}=L$ and $\fl{s_1}\subseteq\dom{s_1}$, and
     when $\dom{\Gamma_1}\cap\dom{\Gamma_2} = \emptyset$ and $\dom{\Gamma_2}=I$ and $\fl{\Gamma_2}\subseteq L$ and $\fl{\Gamma_1}\cap L = \emptyset$.
    Moreover we write
    $ s_1 \sepconj s_2$ and $\Gamma_1 \sepconj \Gamma_2$ when there exist $I$, $L$ such that $ s_1 \sepconj[L] s_2$ and $\Gamma_1 \sepconj[I,L] \Gamma_2$, respectively.

  \begin{gather*}\begin{array}{l@{\;}c@{\;}l@{\;}c@{\;}l}
    \oconf{A_1}{\Gamma_1\sepconj[I,L]\Gamma}{s_1\sepconj s}{\Ed}
    &\sepconj[I][0]&
    \oconf{A_2}{\Gamma_2\sepconj[I,L]\Gamma}{s_2\sepconj[L] s}{\Ed}
    &=&
    \oconf{A_1\cup A_2}{\Gamma_1\sepconj\Gamma_2\sepconj\Gamma}{s_1\sepconj s_2\sepconj s}{\Ed}
    \\[1ex]
    \pconf{A_1}{\Gamma_1\sepconj[I,L]\Gamma}{s_1\sepconj[L] s}{e}
    &\sepconj[I][0]&
    \pconf{A_2}{\Gamma_2\sepconj[I,L]\Gamma}{s_2\sepconj[L] s}{e}
    &=&
    \pconf{A_1\cup A_2}{\Gamma_1\sepconj\Gamma_2\sepconj\Gamma}{s_1\sepconj s_2\sepconj s}{e}
    \\[1ex]
    \oconf{A_1}{\Gamma_1\sepconj[I,L]\Gamma}{s_1\sepconj[L] s}{\Ed_1}
    &\sepconj[I][j]&
    \oconf{A_2}{\Gamma_2\sepconj[I,L]\Gamma}{s_2\sepconj[L] s}{\Ed_2}
    &=&
    \oconf{A_1\cup A_2}{\Gamma_1\sepconj\Gamma_2\sepconj\Gamma}{s_1\sepconj s_2\sepconj s}{\Ed_j}
    \\[1ex]
    \pconf{A_1}{\Gamma_1\sepconj[I,L]\Gamma}{s_1\sepconj[L] s}{e_1}
    &\sepconj[I][1]&
    \oconf{A_2}{\Gamma_2\sepconj[I,L]\Gamma}{s_2\sepconj[L] s}{\Ed_2}
    &=&
    \pconf{A_1\cup A_2}{\Gamma_1\sepconj\Gamma_2\sepconj\Gamma}{s_1\sepconj s_2\sepconj s}{e_1}
    \\[1ex]
    \oconf{A_1}{\Gamma_1\sepconj[I,L]\Gamma}{s_1\sepconj[L] s}{\Ed_1}
    &\sepconj[I][2]&
    \pconf{A_2}{\Gamma_2\sepconj[I,L]\Gamma}{s_2\sepconj[L] s}{e_2}
    &=&
    \pconf{A_1\cup A_2}{\Gamma_1\sepconj\Gamma_2\sepconj\Gamma}{s_1\sepconj s_2\sepconj s}{e_2}
    \\[1ex]
    \hfill C_1 &\sepconj[I][0]& 
    \multicolumn{3}{@{}l@{}}{
        C_2
        =
        \botcconf
        \text{ when }
        C_1 = C_2 = \botcconf
    }
    \\[1ex]
    \hfill C_1 &\sepconj[I][j]&
    \multicolumn{3}{@{}l@{}}{
        C_2
        =
        \botcconf
        \text{ when } i\in\{1,2\} \text{ and }
        C_1 = \botcconf \text{ or } C_2 = \botcconf
    }
  \end{array}\end{gather*}
  Where
    $j\in\{1,2\}$,
    $\fl{e,\Ed} \subseteq \dom{s} = L$,
    $\fl{e_i,\Ed_i}\subseteq \dom{s_i}$ ($i\in\{1,2\}$),
\end{definition}

The intuition here is that instead of exploring bisimulation with the composite configuration, we can instead explore it only with the smaller, constituent configurations. Note that these configurations are allowed to contain a common store $s$ and knowledge environment $\Gamma$. This makes the technique possible in intermediate configurations where some state has already been allocated and functions in $\Gamma$ can access it.

The definition of up to separation shall use a product construction on continuation graphs and a dual merging operation. The key intuition is that (bi-)simulation is preserved by these operations for graphs and relations (cf.~\cref{sec:pair-bisim}).

\begin{definition}[Pair Entry Points and Pair Continuation Graphs]
  \begin{align*}
    \mathsf{PEPoint} \ni \ \bbeta &::=\ (\beta_1, \beta_2, k)  \qquad (k\in\{0,1,2\})\\
    \mathsf{PCGrph}\ni \ \bSigma &\subseteq_{\rm fin}^{\neq\emptyset} \mathsf{PEPoint}\times\Kont\times\Kont\times\mathsf{PEPoint}
  \end{align*}
 and each $\bSigma$ must satisfy the conditions: 
  \begin{compactitem}
  \item{} \emph{Reachability}.~For all $\bbeta\in\dom{\bSigma}$ there are stacks $K_1,K_2$ such that $\bbeta\toSigs[\bSigma]{K_1,K_2}\diamond$

  \item \emph{Top and Divergence}.~For all $\bbeta'^{k'}\toSig[\bSigma]{\Ed_1,\Ed_2}\bbeta^{k}$ and $j\in\{1,2\}$:
  \begin{align*}
    &    (\diamond,\diamond,0)\in\dom{\bSigma}\land (\bbeta'=(\diamond,\diamond,0)\implies\bbeta=(\diamond,\diamond,0)) \land
    (\bbeta=(\diamond,\diamond,0)\iff\Ed_j=\diamond) \\
&{}\land    (\bbeta'.k.j=\bot\iff \Ed_j=\bot) \land
(\bbeta.k.j=\bot\implies
    \Ed_j=\bot)
  \end{align*}
    where, if $\bbeta = (\beta_1,\beta_2, k)$ and $i \in\{1,2\}$, we write $\bbeta.i.1=\bot$ when $\beta_i=\betais{\botcconf,C}$ 
    and $\bbeta.i.2=\bot$ when $\beta_i=\betais{C,\botcconf}$.
  \item \emph{Nominal closure}.~For  all $\bbeta'\toSig[\bSigma]{\Ed_1,\Ed_2}\bbeta$ and permutations $\pi$, \
    $\pi\cdot\bbeta'\toSig[\bSigma]{\pi\cdot\Ed_1,\pi\cdot\Ed_2}\pi\cdot\bbeta$.
\end{compactitem}
Finally, we lift \cref{def:sigma-ext-restr} to pair graphs obtaining continuation graph extension
  $\bSigma[\bbeta'\mapsto(\Ed_1,\Ed_2,\bbeta'')]$
and restriction $\bSigma@\bbeta$.
We will write $\bbeta.i$ to mean $\beta_i$ ($i\in\{1,2\}$), and $\bbeta^m$ to mean $k=m$, when $\bbeta = (\beta_1,\beta_2,k)$.
\end{definition}

\begin{definition}
  Given two continuation graphs $\Sigma_1$, $\Sigma_2$ we construct the product graph:
  \begin{gather*}
    \irule{
    }{
      (\diamond,\diamond,0)\toSig[\Sigma_1 \otimes \Sigma_2]{ \diamond, \diamond} (\diamond, \diamond, 0)
    }
    \\
    \irule{
      (\beta_1,\beta_2,k) \in \dom{\Sigma_1\otimes\Sigma_2}
      \;\;
      \forall i \in \{1,2\}.\ \beta_i'\toSig[\Sigma_i]{\Ed, \Ed'}\beta_i
      \;\;
      \forall j \in \{1,2\}.~\beta'_1.j.e=\beta'_2.j.e
    }{
      (\beta_1',\beta_2',0)\toSig[\Sigma_1 \otimes \Sigma_2]{ \Ed, \Ed'} (\beta_1, \beta_2, k)
    }
    \\
    \irule{
      (\beta_1,\beta_2,k) \in \dom{\Sigma_1\otimes\Sigma_2}
      \\
      \beta_i'\toSig[\Sigma_i]{\Ed, \Ed'} \beta_i
      \\
      \beta_{\co{i}} = \beta_{\co{i}}'
      \\
      i\in\{1,2\}
    }{
      (\beta_1',\beta_2',i)\toSig[\Sigma_1 \otimes \Sigma_2]{ \Ed, \Ed'} (\beta_1, \beta_2, k)
    }
  \end{gather*}
  \end{definition}
  \begin{lemma}
    Suppose  $\Sigma_1, \Sigma_2$ well-formed continuation graphs; then $\Sigma_1 \otimes \Sigma_2$ is a well-formed pair continuation graph.
 \qed \end{lemma}

\begin{definition}[Merging]
  Suppose
  \begin{align*}
    \beta_1 &= ((\Gamma_1 \sepconj[I,L] \Gamma_3), (\Gamma_1' \sepconj[I,L'] \Gamma_3'), (s_1 \sepconj[L] s_3), (s_1' \sepconj[L'] s_3'), e_1)\\
    \beta_2 &= ((\Gamma_2 \sepconj[I,L] \Gamma_3), (\Gamma_2' \sepconj[I,L'] \Gamma_3'), (s_2 \sepconj[L] s_3), (s_2' \sepconj[L'] s_3'), e_2)
  \end{align*}
  We define the partial merging function $\bmerge{\cdot}$ for pair nodes as:
  \begin{align*}
    \bmerge{ (\beta_1,~ \beta_2,~ k) }
    \defeq
      (\Gamma_1 \sepconj \Gamma_2 \sepconj \Gamma_3, \Gamma_1' \sepconj \Gamma_2' \sepconj \Gamma_3', s_1 \sepconj s_2 \sepconj s_3, s_1' \sepconj s_2' \sepconj s_3', e)
  \end{align*}
provided that when $k=0$ then $e_1=e_2=e$ and when $k\in\{1,2\}$ then $e_i = e$.
We extend merging to well-formed pair continuation graphs $\bSigma$:
  \begin{align*}
     \bmerge{\bbeta'}\toSig[\bmerge\bSigma]{\Ed_1, \Ed_2} \bmerge{\bbeta} 
    \quad\text{when}\quad
    \bbeta'\toSig[\bSigma]{\Ed_1, \Ed_2} \bbeta 
  \end{align*}

\end{definition}

\begin{definition}[Up to Separation] \label{def:uptosep} 
  The partial function $\utsepconj{}$ provides the up to separation technique:
  \[
    C_1 \sepconj[I][k] C_2 \utsepconj*{\bisim{R}}_{\Sigma,\bmerge{\bbeta}} C_1' \sepconj[I][k] C_2'
  \]
  when
  $C_i \bisim*{R}_{\Sigma_i,\beta_i} C_i'$ ($i\in\{1,2\}$)
  and
  $\Sigma=\bmerge{(\Sigma_1\otimes\Sigma_2)@\bbeta}$
  and
  $\bbeta= (\beta_1,\beta_2,k)$.
\end{definition} 

\begin{lemma}\label{lem:utsepconj-sound}
  Function \utsepconj{} is a sound up-to technique.
\qed\end{lemma}

This technique is also complete, which is important for our tool as it allows us to use it without backtracking (in contrast to 
up to weakening and garbage collection). 

\begin{lemma}\label{lem:utsepconj-complete}
  Suppose that
  $C_1 \sepconj[I][0] C_2 \simil_{\Sigma,\beta} C_1' \sepconj[I][0] C_2'$.
  Then, there exist  $\Sigma'$ and $\beta'$ such that
  $C_i \simil_{\Sigma',\beta'} C_i'$, for $i\in\{1,2\}$.
\qed\end{lemma}

\section{Implementation and Evaluation}
\label{sec:imp}
We implemented PDNF bisimulation in a prototype tool called \SLHobbit{} that checks programs written in an ML-like syntax for \lang. The tool was developed by replacing the LTS and bisimulation definition in \Hobbit{}\cite{KoutavasLT22} with an implementation of our Stackless LTS (Fig. \ref{fig:SLlts}) and a Bounded Symbolic Execution of our PDNF bisimulation (Def. \ref{def:SLbisim}). As such, the tools share the same front-end, enhancement techniques, reduction semantics, and symbolic execution routine (calling Z3 to resolve constraints). They {otherwise} differ in the implementation of the LTS and bisimulation game, as well as in \Hobbit's up to reentry, which the \SLHobbit{} cannot use as it lacks a stack.

As a symbolic execution tool, \SLHobbit{} is \textit{sound} (reports only true positives and negatives) and \textit{bounded-complete} (exhaustively and precisely explores all paths up to a bound). {The bound used here is different than the bound in the decidability result in \cref{sec:decidability}, and it is intended to be used as a more straightforward timeout.} In \SLHobbit{}, we bound the number of proponent calls and both opponent calls and returns along an execution path, whereas \Hobbit{} bounds only calls. This is done because {the saturation procedure in PDNF bisimulation may lead to cycles in the continuation graph, which when explored by \SLHobbit{} lead to unbounded returns without the same number of corresponding calls}. We accumulate SAT/SMT constraints by extending the LTS with a \textit{symbolic environment} $\sigma$ for \textit{symbolic constants} $\kappa$ and reductions involving any $\kappa$; we branch on symbolic conditions as is standard of symbolic execution.
The exploration is performed over \textit{configuration pairs} $\langle C_1, C_2, \Sigma, \beta, \sigma, k_{\textsf{call}} , k_{\textsf{ret}} , k_{\textsf{int}} \rangle$ of related term configurations $C_1$ and $C_2$, continuation graph $\Sigma$, current call entry point $\beta$, {symbolic environment $\sigma$} and given bounds $k_{\textsf{call}}$ for calls, $k_{\textsf{ret}}$ for returns and $k_{\textsf{int}}$ for internal reductions.
As with \Hobbit{}, we make use of enhancements that help finitise the bisimulation exploration in some examples: explore-set memoisation {to discover cycles}; normalisation; store garbage collection; $\sigma$ garbage collection and simplification; up to separation (Sec. \ref{sec:up-to-sep}); and up to name reuse (Sec. \ref{sec:up-to-nr}). In addition, a normalisation procedure is implemented to ensure $\Sigma$ is effectively closed under permutation by capturing the complete orbit of every edge in $\Sigma$ via a canonical representation of the abstract names in said edge.

\subsubsection{Evaluation}
\newcommand{\titleAC}{\multicolumn{1}{@{\;\quad}c@{\;\quad}}{PDNF}}
\newcommand{\titleB}{\multicolumn{1}{c}{\Hobbit{}}}
\newcommand{\titleC}{\multicolumn{1}{c|}{H+reentry}}
\newcommand{\titleE}{\multicolumn{1}{c|}{\Hobbit{}}}
\newcommand{\tabitem}{~~\llap{\textbullet}~~}

\begin{table}[t]
\begin{center}
  \begin{tabular}{|cl|ccc|cc|}
\hline
   && \multicolumn{3}{c|}{~\Hobbit{} tests: 129 eq's and 78 ineq's~} & \multicolumn{2}{c|}{~\SLHobbit{} tests: 12 eq's} \\[1mm]\cline{3-5}\cline{6-7}
  &         & \titleAC & \titleB  & \titleC & \titleAC & \titleE \\ \hline
\multicolumn{1}{|r}{Eq.}
& Proven  & 72 & 62 & 67 & 11 & 0 \\\hline
\multicolumn{1}{|r}{Ineq.}
& Proven  & 77 & 78 & 78 & {\scriptsize N/A} &{\scriptsize N/A}           \\ \hline
  \end{tabular}\\[4mm]
      \begin{tikzpicture}[scale=0.75]
  \begin{axis}[
    width=0.8\textwidth,
    height=0.4\textwidth,
    xmode=log,
    ymode=log,
    title={},
    grid=major,
    legend pos=north west 
  ]
  \addplot[only marks, mark size=1pt] table [x=Hobbit, y=SLBisim, col sep=comma] {data.csv};
  \draw [gray] (rel axis cs:0,0) -- (rel axis cs:1,1);
  \legend{Time (s)}
  \end{axis}
\end{tikzpicture}\\
(X) \Hobbit{} vs. (Y) \SLHobbit{} over \Hobbit{}'s test suite
\end{center}
\caption{Summary of experiments comparing \SLHobbit{} to \Hobbit{}}
\label{table:results}\vspace{-5mm}
\end{table}

We evaluate here our tool against \Hobbit{} as a reference implementation of
the standard (stacked) bisimulation and because of its favourable comparison to
other tools in the higher-order program equivalence
landscape~\cite[Sec.~9]{KoutavasLT22}. Both tools were executed over two test
suites: (1) \Hobbit{}'s suite of 129 equivalences and 78 inequivalences; and
(2) our own suite of 12 equivalences (11 inspired by Event Handlers in
Android~\cite{android_events,android_timer}, JavaScript~\cite{mdn_web_docs},
Java Swing~\cite{java_swing}, jQuery~\cite{jquery}, and the DOM
Framework~\cite{dom_standard}; and 1 based on a simplification of a CDMA-WLAN
handoff protocol\cite{DBLP:journals/ejwcn/KimSSQ08}). Combined, the test suites
total 6701LoC\,---\,viz. 6182LoC in (1) with 3802LoC in equivalences and
2380LoC in inequivalences, and 519LoC in (2). For this comparison, we are interested in three scenarios: the performance of both approaches as fully automatic
techniques, for which all invariant, reentry and synchronisation
annotations were removed from the \Hobbit{} testsuite; \SLHobbit{} against \Hobbit{} assuming the reentry annotations have been placed correctly, to measure how PDNF
bisimulation fares in comparison to NF bisimulation with up to reentry to finitise reentrant calls; and \SLHobbit{} against \Hobbit{} on our own test suite, which aims to showcase the
difficulty of dealing with reentrant functions in the presence of changing
state. The tools were evaluated on an Intel Core i7 1.90GHz machine with 32GB
RAM running OCaml 4.10.0 and Z3 4.8.10 on Ubuntu 23.04. We record the results
of our comparison in Table~\ref{table:results}. Execution of each example was
capped to a 150-second timeout.

Firstly, \SLHobbit{} {verified 72 equivalences, which contain all 62 equivalences that \Hobbit{} verified}; without up to reentry \Hobbit{} {did not prove any examples that \SLHobbit{} could not prove.} Execution times were also not significantly different ($r = 0.74$). We can thus conclude that for equivalences \SLHobbit{} supersedes \Hobbit{} at fully-automatic verification by proving 9 additional examples on \Hobbit{}'s own test suite (without manual annotations) with minimal difference in performance. Note, however, that \Hobbit{} is more mature as a semi-automatic tool, and (from testing) is able to prove up to 95 examples when invariant, reentry and synchronisation annotations are appropriately used (albeit requiring significant effort and experience in formalising equivalence annotations). Additionally, one inequivalence is not proven by \SLHobbit{}. This example (c.f. {\small \textsf{invariants-4}}) is particularly difficult as no up-to techniques apply, memoisation is unable to finitise the path exploration, and the failing trace exhibits sequences of sub-traces that nest deeply. Both \SLHobbit{} and \Hobbit{} are able to solve this example on small parameters and both encounter an exponentially growing number of configurations, but \Hobbit{}'s more elementary transition system leads to a faster exploration. Bar implementation concerns, the slower analysis may be explained by a higher branching factor due to graph-based returns, which are additionally able to expand more deeply than in \Hobbit{} as returns can occur without the same number of corresponding calls.

Secondly, we observe in Table~\ref{table:results} that, \Hobbit{} was able to prove an additional 5 examples by turning on the up to reentry technique (but not the rest of the manual up-to techniques) and carefully adding reentry annotations in the right functions, leaving 5 examples from \Hobbit's testsuite that can be exclusively proven by \SLHobbit.
We can thus conclude that for our second scenario on equivalences, \SLHobbit{} still supersedes \Hobbit{} with semi-automatic reentry annotations.

Finally, on our own test suite, \SLHobbit{} is clearly superior to \Hobbit{} on higher-order stateful programs that feature reentrant calls with changing state as is common in, but not limited to, higher-order data structures, event-driven programming, and various protocols. In these, the stackless approach was able to quickly saturate the graph and prove equivalence, whereas \Hobbit{} was unable to finitise the bisimulation game and eventually timed out. Lastly, one of the examples in our test suite is provable by neither \SLHobbit{} nor \Hobbit{}. We include this example to illustrate a current limitation of our technique: it cannot finitise exploration concerning infinite state. To achieve this would require adapting \Hobbit{}'s invariant annotations technique to our framework.


\section{Related and Future Work}
\label{sec:outro}

Theorems of closed instantiation of uses (CIU theorems) were amongst the first operational techniques that reduced the contexts considered by contextual equivalence in languages with state \cite{mason_talcott_1991,felleisen1987calculi,gordon_hankin_lassen_1999}.
Applicative bisimulation \cite{applicative90} was the first application of bisimulation to a (pure) higher-order programming language, which reduces contexts further by considering applying top-level term functions to identical closed arguments.
Logical relations \cite{JaberTabareau15,state-dependent,HurDNV12} can be viewed as similarly reducing the examined contexts applying functions to related arguments.
Environmental bisimulation \cite{SumiiPierce07a,SumiiPierce07b,KoutavasW06,environmental} introduces stratification of bisimulations based on state and opponent knowledge, providing an effective proof technique due to being amenable to up-to techniques~\cite{BiernackiLP19,KoutavasW06,environmental,PousS11}, while applying functions to closed arguments derived by the congruence of the bisimulation.
Game semantics~\cite{AJM,HO,Nickau}, provides fully abstract denotational semantics for a range of higher-order languages, and in particular languages with higher-order state \cite{AbramskyHM98,Laird07,MurawskiT21}. Algorithmic interpretations thereof give rise to decision procedures for contextual equivalence for restricted language fragments~\cite{GhicaM03,hector,Dimovski14,coneqct}.
The \tool{SyTeCi} tool~\cite{syteci} combines notions from game semantics and logical relations, and manages to overcome some of the language restrictions of game-semantics tools.
Normal form bisimulation, discussed in the introduction, treats context-generated code symbolically, entirely removing quantification over context-generated code and leading to sound but not complete techniques, with the notable exception of the case of higher-order languages with: sequential control and state~\cite{StorvingLassenPOPL07}, state-only \cite{BiernackiLP19,KoutavasLT22}, and no effects \cite{KoutavasLT23}.
It has been shown~\cite{LassenL07,LassenLevy:nfbisimpoly,KoutavasLT22} that NF bisimulation relates to operational game semantics models where opponent-generated terms are also represented by names~\cite{Laird07,TzeGhica,Jaber15}.
The closest work to ours is~\cite{KoutavasLT22}, which combines game semantics and techniques from environmental bisimulations and up-to techniques to produce a fully abstract LTS suitable for NF bisimulation.

Unlike prior approaches, our treatment of the stack stems from model checking pushdown systems~\cite{BouajjaniEM97,FinkelWW97,Schwoon02} and exact-stack control-flow analyses of higher-order functional languages~\cite{P4F,CFA2-jour,PDCFA,AAC}, and allows us to
eliminate the need for a term/context call stack without loss of precision.
Our approach is related to~\cite{P4F}, where the use of a continuation graph
is proposed (called \emph{continuation store}). 
The reachability analysis of procedural code using pushdown systems and saturation techniques was first considered in~\cite{EsparzaKS01,Schwoon02}. Saturation typically relies on the fact that the underlying control state space is finite, which is not the case in our NF bisimulation games. We therefore follow an on-the-fly forward saturation procedure which over-approximates the saturation procedure devised in~\cite{FinkelWW97}. While this over-approximation is generally unsound (cf.~\cref{ex:dummy}), it is sound for reachability.
 
In conclusion, in this work we created a novel
fully abstract technique for contextual equivalence and implement a
bounded-complete prototype verification tool. Our tool is able to verify equivalence in a number of examples which were out of reach in previous work.
%
In the future we believe that our work can lead to useful verification tools, for example for regression verification~\cite{regressionverif,rvt,reve,symdiff} in higher-order languages with state, relational verification of assertion reachability in code, or even (single-program) contextual model checking in settings such as blockchain smart contracts.


%
%


{\let\thefootnote\relax\footnote{{For the purpose of Open Access, the author/s has/have applied a CC BY public copyright licence to any Author Accepted Manuscript version arising from this submission.}}}

\bibliographystyle{ACM-Reference-Format}
  \bibliography{references}

   \clearpage
   \emph{This appendix is provided for the benefit of the reviewers, and will not appear in a final version of this paper.}
   \appendix
%

\section{The Stackless LTS}\label{sec:stackless-lts-full}
\begin{figure*}[h] 
  \[\begin{array}{@{}ll@{\;\;}l@{}}
    \irule![PropCall][propappf]{
      (D,\Gamma') \in \ulpatt(v)
    }{\nbox{
      \pconf{A}{\Gamma}{s}{E[\app {\alpha} v]}
      \trans{\lpropapp{\alpha}{D}} 
      \oconf{A}{\Gamma,\Gamma'}{s}{E}
    }}
    \\
    \irule![PropRet][propretf]{
      (D,\Gamma') \in \ulpatt(v)
    }{
      \pconf{A}{\Gamma}{s}{v}
      \trans{\lpropret{D}}
      \oconf{A}{\Gamma,\Gamma'}{s}{\chi}
    }
    \\
     \irule![OpCall][opappf]{
       \nbox{
      \vec\alpha\fresh\Gamma,s,\Ed\land        (D,\vec \alpha) \in \ulpatt(T)
      \\{\hspace{-5mm}}\land
            \Sigma_s \vdash \Gamma(i) : T \arrow T'
\land
        \app {\Gamma(i)} {D\hole[\vec \alpha]} \funred e
      }
    }{
      \oconf{A}{\Gamma}{s}{\Ed}
      \trans{\lopapp{i}{D\hole[\vec \alpha]}}
      \pconf{A\uplus\vec\alpha}{\Gamma}{s}{e}
    }
    \\
    \irule![OpRet][opretf]{
      \vec\alpha\fresh\Gamma,s,\Ed\land  (D,\vec \alpha) \in \ulpatt(T)
    }{
      \oconf{A}{\Gamma}{s}{E\hole_T}
      \trans{\lopret{D\hole[\vec\alpha]}}
      \pconf{A\uplus\vec\alpha}{\Gamma}{s}{E\hole[D[\vec\alpha]]}
    }
    \\
    \irule![Tau][tau]{
      \redconf{s}{e} \red \redconf{s'}{e'}
    }{
      \pconf{A}{\Gamma}{s}{e}
      \trans{\tau}
      \pconf{A}{\Gamma}{s'}{e'}
    }
    \\
    \irule![Response][dummy]{
      \nbox{ \eta \not= \tau
      \text{ and } 
      C \not\wtrans{\eta} 
      \text{ from other rules}
    }
    }{
      C \trans{\eta} \botcconf
    }
   \end{array}\]
  \hrule
  \caption{The Stackless Labelled Transition System. Relation $\funred$ is defined as in \Cref{fig:lts}.}\label{fig:SLlts}
\end{figure*}

   \section{Typing rules of \lang}\label{sec:typing}

\newcommand\qweqwe{\quad}
\newcommand\preG[1][]{\Delta#1;\varSigma\vdash}
\newcommand\atype[1]{\mathsf{#1}}
\newcommand\vars{Vars}
\newcommand\refs{Refs}
\newcommand\meths{Meths}

\begin{gather*}
{\infer{\alpha\in\Nam_{T\to T'}}{ \preG \alpha : T\to T' }}
\qweqwe
{\infer{c\text{ cons.\ of type }T}{ \preG c : T }}
\qweqwe
\infer{ (x:T)\in\Delta }{ \preG x : T }
\qweqwe
\infer{ \preG e_1: T_1\quad\dots\quad \preG e_n: T_n }{ \preG (e_1,\dots,e_n) : T_1*\dots*T_n }
\\[1mm]
\infer{ op:\vec T\to T\quad  \preG (\vec e):\vec T }{ \preG \arithop{\vec e} : T }
\qweqwe
\infer{ \preG e : \atype{bool}\qquad  \preG (e_1,e_2) : T*T }{ \preG \cond{e}{e_1}{e_2} : T }
\\[1mm]
\infer{\preG v : T \quad  \Delta;\varSigma,l:T\vdash e : T' }{ \preG \new l v e : T' }
\qweqwe 
\infer{ (l:T)\in\varSigma }{ \preG {!l} : T }
\qweqwe
\infer{ (l:T)\in\varSigma \quad \preG e : T }{\preG l := e : \atype{unit} }
\\[1mm]
\infer
{ \preG e: T \rightarrow T' \quad \preG e' : T }
{ \preG ee' : T' }
\qweqwe 
\infer{ \preG[,f:T\to T',x:T] e : T' }{ \preG \lam{x}{e} : T\to T'}
\\[1mm]
\infer{ \preG[,x_1:T_1,\dots,x_n:T_n] e : T \quad  \preG e' : \vec T }{ \preG \elet{(\vec x)}{e'}{e} : T }
\end{gather*}


   \section{Simple lemmas}\label{app:simple-lemmas}

\paragraph{\cref{lem:Ks}}
  For any $\beta\toSigsD{K_1,K_2}$, we have that $K_1,K_2$ are defined and:
  \begin{itemize}
    \item $\beta=\diamond$ and $K_1=K_2=\noe$, or
    \item $\beta.1,\beta.2\neq\bot$ and 
      $|K_1|=|K_2|$, or
      \item $\exists j\in\{1,2\}.\ \beta.j=\bot \land K_j=\bot\neq K_{3-j}$
\end{itemize}
where $|K|$ is the length of $K$ (if $K\neq\bot$). 
\begin{proof}
  By rule induction. The base case is clear. Suppose now
  \[
      \irule{
      \beta'\toSigsD{K_1,K_2}
      \\
      \beta\toSig{\Ed_1,\Ed_2}\beta'
    }{
      \beta\toSigsD{(\Ed_1,K_1),(\Ed_2,K_2)}
    }
  \]
  By induction hypothesis, either $\beta'.1,\beta'.2\neq\bot$ and $|K_1|=|K_2|$, or (WLOG)
$\beta'.1=\bot$ and 
$K_1=\bot\neq K_2$. Suppose
the former is the case.
If $\beta'\neq\diamond$ then $\Ed_j\neq\diamond$ and $\Ed_j,K_j$ is defined (for $j=1,2$), hence
$\beta.1,\beta.2\neq\bot$ and
$|\Ed_1,K_1|=|\Ed_2,K_2|$, or (by divergence) there is exactly one $j$ such that  $\Ed_j,K_j=\bot$ and $\beta.j=\bot$. If $\beta'=\diamond$ then by definition $K_1=K_2=\noe$ and $\Ed_1=\Ed_2=\diamond$ and the claim follows. It remains to check the case $K_1=\bot\neq K_2$ and $\beta'.1=\bot$. By \Cref{def:betaSigma} (divergence) we have that $\Ed_1=\beta.1=\bot$. It remains to check that $\Ed_2,K_2$ is defined and not $\bot$. The former is clear as $\Ed_2\neq\diamond$ (as $\beta'\neq\diamond$) and $K_2\neq\bot$. For the latter, observe that $\Ed_1=\bot\neq\Ed_2$. \qed
\end{proof}

\paragraph{\cref{lem:toS}}
  For any $\beta\models_\Sigma(K_1,K_2)$ and $C_1,C_2$, if $\beta=\betais{C_1,C_2}$ or $(C_1,C_2,\Sigma,\beta)$ compatible
  then $\toS{C_1,K_1},\toS{C_2,K_2}$ are defined.
\begin{proof}
  By symmetry, it suffices to show that $\toS{C_1,K_1}$ is defined. WLOG assume $C_1\neq\botcconf$, so $\beta.1\neq\bot$. Then, by \cref{lem:Ks} we have that $K_1\neq\bot$. Finally, we need to check that if $C_1$ is an O-configuration and $(C_1,C_2,\Sigma,\beta)$ compatible then $C_1.\Ed=\diamond\iff K_1=\cdot$. By \Cref{lem:Ks}, $K_1=\cdot$ iff $\beta=\diamond$. By compatibility,  $\beta=\diamond$ iff $C_1.\Ed=\diamond$.
Thus,  $\toS{C_1,K_1}$ is defined. \qed
\end{proof}

\paragraph{\Cref{lem:soundness}}
  If $\bisim R$ is a weak simulation then so is:
  \begin{align*}
     \tR = \{
           (\toS{C_1,K_1},\toS{C_2,K_2}) \mid&\; 
                                                         \exists\Sigma,\beta.\
  C_1\bisim*{R}_{\Sigma,\beta}C_2
\land
                                                         \beta\toSigs{K_1,K_2}\diamond
\}
  \end{align*}
  Moreover, if $\bisim R$ is a weak bisimulation then so is $\tR$.
\begin{proof}
  Let $\Sigma,\beta,C_1,C_2,K_1,K_2$ be such that
  $\beta\toSigs{K_1,K_2}\diamond$
and
$C_1\bisim*{R}_{\Sigma,\beta}C_2$, and $\tilde C_i=\toS{C_i,K_i}$. We assume WLOG that $\tC_1\neq\botconf$ as otherwise the simulation conditions are vacuously true. If $\tC_1\lterm$ then $K_1=\noe$ and $C_1.\Ed=\diamond$, and hence $C_1\lterm$ and by 
$  C_1\bisim*{R}_{\Sigma,\beta}C_2$
we obtain $C_2\lterm$. 
Since   $\beta\toSigs{K_1,K_2}\diamond$
 by \cref{lem:Ks}
we have $K_2=\noe$ or $K_2=\beta.2=\bot$. The latter case is excluded by $C_2\neq\botcconf$, hence $\tC_2\lterm$.
Suppose now
$C_i=\cconf{A}{\Gamma_i}{s_i}{\Phi_i}$, for $i=1,2$,
so
 $\tilde C_i=\conf{A}{\Gamma_i}{K_i}{s_i}{\Phi_i}$. Cases:
\begin{asparaitem}
\item $\tC_1\realtrans{\lopapp{i}{D[\vec\alpha]}}\tC_1'$ with $\vec\alpha\fresh\tC_2$. By \Cref{lem:closures} and equivariance of the LTS, we can assume WLOG that $\vec\alpha\fresh\beta$, so $\vec\alpha\fresh C_2,\beta$.
  Then, $\Phi_i=\Ed_i$ ($i=1,2$) and 
  $C_1\realtrans{\lopapp{i}{D[\vec\alpha]}}C_1'$ with $\tC_1'=\toS{C_1',(\Ed_1,K_1)}$. Thus,
    $C_2\trans{\lopapp{i}{D[\vec\alpha]}}C_2'$ and $C_1'\bisim*{R}_{\Sigma',\beta'}C_2'$ with $\beta'=\betais{C_1',C_2'}$ and $\Sigma'=\Sigma[\beta'\trans{\Ed_1,\Ed_2}\beta]$.
Hence, $\tC_2\trans{\lopapp{i}{D[\vec\alpha]}}\tC_2'$ with $\tC_2'=\toS{C_2',(E_2,K_2)}$. The claim follows from the fact that $\beta'\toSigsD[\Sigma']{(\Ed_1,K_1),(\Ed_2,K_2)}$.
\item $\tC_1\realtrans{\lpropret{D}}\tC_1'$.
Then, $\Phi_i=v_i$ ($i=1,2$) and, assuming $K_i=\Ed_i,K_i'$ ($i=1,2$) and $\tC_1'=\oldoconf{A}{\Gamma_1'}{K_1'}{s_1}{\Ed_1}$,
we have $C_1\realtrans{\lpropret{D}}C_1'$ with $C_1'=\oconf{A}{\Gamma_1'}{s_1}{\chi}$.
Now,
if $\Ed_1,\Ed_2\not=\diamond$ then 
$\beta\toSigsD{K_1,K_2}$ must be due to some $\beta'\toSigsD{K_1',K_2'}$ and $\beta\toSig{\Ed_1,\Ed_2}\beta'$. Otherwise, $\Ed_1=\Ed_2=\diamond$ and $K_i=K_i'=\noe$ ($i=1,2$). In either case, there is $\beta'$ such that $\beta\toSig{\Ed_1,\Ed_2}\beta'$ and $\beta'\toSigsD[\Sigma']{K_1',K_2'}$, where $\Sigma'=\Sigma@\beta'$.
For that $\beta'$, by bisimilarity conditions, we have $C_2\wtrans{\lpropret{D}}C_2'$ and 
$C_1'\sub{\chi}{\Edb_1} \bisim*{R}_{\Sigma',\beta'} C_2'\sub{\chi}{\Edb_2}$, and hence
$\tC_2\wtrans{\lpropret{D}}\tC_2'$ with $\tC_2'=\oldoconf{A}{\Gamma_2'}{K_2'}{s_2}{\Ed_2}$ and
$C_2'=\oconf{A}{\Gamma_2'}{s_2}{\chi}$ for some $\Gamma_2'$ and $s_2$. Observing that $\tC_i'=\toS{C_i'\sub\chi{\Ed_i},K_i'}$ we obtain $\tC_1'\tsR\tC_2'$.
\item $\tC_1\realtrans{\tau}\tC_1'$.
Then, $\Phi_i=e_i$ ($i=1,2$) and, assuming $\tC_1'=\oldpconf{A}{\Gamma_1}{K_1}{s_1'}{e_1'}$, we have $C_1\realtrans{\tau} C_1'$ with $C_1'=\pconf{A}{\Gamma_1}{s_1'}{e_1'}$. Thus, by $C_1\bisim*R_{\Sigma,\beta}C_2$, we have $C_2\wtrans{\tau}C_2'$ with $C_1'\bisim*R_{\Sigma,\beta}C_2'$ for some $C_2'=\pconf{A}{\Gamma_2}{s_2'}{e_2'}$. But then $\tC_2\wtrans{\tau}\tC_2'=\oldpconf{A}{\Gamma_2}{K_2}{s_2'}{e_2'}$ and, by definition, $\tC_1'\tsR\tC_2'$.
\item 
  $\tC_1\realtrans{\lpropapp{\alpha}{D}}\tC_1'$ or $\tC_1\realtrans{\lopret{D[\vec\alpha]}}\tC_1'$. Similar to the previous cases.
\end{asparaitem}
For the case where $C_2=\botcconf$ we can  replay the same arguments as in the cases above.
%
\\
  Finally, if $\bisim R$ is a weak bisimulation then $\bisim R,\bisim R^{-1}$ are weak simulations. Moreover:
  \begin{align*}
\tR^{-1} &=  \{
    (\toS{C_2,K_2},\toS{C_1,K_1}) \mid \exists\Sigma,\beta.\ \beta\toSigsD{K_1,K_2}
\land
  C_1\bisim*{R}_{\Sigma,\beta}C_2
           \}\\
     &=  \{
    (\toS{C_2,K_2},\toS{C_1,K_1}) \mid \exists\Sigma,\beta.\ \beta^{-1}\toSigsD[\Sigma^{-1}]{K_2,K_1}
\land
  C_2\bisim*{R}_{\Sigma^{-1},\beta^{-1}}^{-1}C_1
       \}\\
    &=  \{
    (\toS{C_1,K_1},\toS{C_2,K_2}) \mid \exists\Sigma,\beta.\ \beta\toSigsD{K_1,K_2}
\land
  C_1\bisim*{R}_{\Sigma,\beta}^{-1}C_2
           \}
  \end{align*}
  and, hence, both $\tR,\tR^{-1}$ are weak simulations.
\qed
\end{proof}  


   \section{Proof of completeness}\label{app:completeness}

\paragraph{\cref{lem:completeness}}
The following is a weak (pushdown) simulation:
  \begin{align}
    \bisim{R} = \{\ &(C_1,C_2,\Sigma,\beta) \mid\;
\Sigma\text{ sat-connected}\land
                 (C_1,C_2,\Sigma,\beta)\text{ compatible}\notag\\
                    &{}\land\forall K_i.\ \beta\toSigsD{K_1,K_2}
                      \implies
                                                \toS{C_1,K_1}\simil\toS{C_2,K_2}\tag{A}\label{cond:a}\\
              & {}\land\forall C_i',K_i.\  \betais{C'_1,C'_2}\toSigsD{K_1,K_2}
                \implies
                \toS{C_1',K_1}\simil\toS{C_2',K_2}\tag{A$^*$}\label{cond:a1}\\
              & {}\land\beta\neq\diamond\implies\exists C_1',C_2'.\ \beta=\betais{C_1',C_2'}\land 
 (C_1',C_2')\satrans{\varepsilon}(C_1,C_2)\tag{B}\label{cond:b}
\    \}
  \end{align}
\begin{proof}
  Let $C_1\bisim*R_{\Sigma,\beta}C_2$, for $C_1\not=\botcconf$, and pick some $\beta\toSigsD{K_1,K_2}$.
By \cref{cond:a} we have $\tilde C_1\simil\tilde C_2$.
If $C_1\lterm$ then $C_1.\Ed=\diamond$ and, by compatibility, $\beta=\diamond$. But then $K_1=\noe$ and therefore $\tC_1\lterm$, so $\tC_2\lterm$. Again by compatibility we get $C_2.\Ed=\diamond$ hence $C_2\lterm$.
Suppose now $C_1\realtrans{\eta}C_1'$.
Cases:
\begin{asparaitem}
\item $\eta=\tau$. Then, $\tC_1\realtrans{\tau}\tC_1'=\toS{C_1',K_1}$ and hence, by $\tau$-closure, $\tC_1'\simil \tC_2$.
  It suffices to show that $C_1'\bisim*R_{\Sigma,\beta}C_2$. Condition \cref{cond:a1} holds as $\Sigma$ has not changed, while \cref{cond:b}
 follows from the hypothesis and \Cref{def:sat}.
Finally, for \cref{cond:a} we use \cref{lem:sat}.
\item $\eta=\lpropapp{\alpha}{D}$. Then, $\tC_1\realtrans{\lpropapp{\alpha}{D}}\tC_1'=\toS{C_1',K_1}$ and hence
$\tC_2\wtrans{\lpropapp{\alpha}{D}}\tC_2'$ with
$\tC_1'\simil\tC_2'$. We obtain $C_2\wtrans{\lpropapp{\alpha}{D}}C_2'$, with $\tC_2'=\toS{C_2',K_2}$, and it suffices to show conditions \cref{cond:a,cond:a1,cond:b}.
These are shown as in the previous case above.
\item $\eta=\lopret{D[\vec\alpha]}$ with $\vec\alpha\fresh C_2$.
Let $\vec\alpha'\fresh K_1,K_2,\vec\alpha$ be of the same length as $\vec\alpha$.
  Then, $\tC_1\realtrans{\lopret{D[\vec\alpha']}}\tC_1'$ and hence
$\tC_2\trans{\lopret{D[\vec\alpha']}}\tC_2'$ with
$\tC_1'\simil\tC_2'$. We obtain $C_2\trans{\lopret{D[\vec\alpha]}}C_2'$,
where $\tC_i'=\toS{\pi\cdot C_i',K_i}$ and $\pi=(\vec \alpha\ \vec\alpha')$ (permute component-wise $\vec \alpha$ with $\vec\alpha'$).
We need to show that $C_1'\bisim*R_{\Sigma,\beta}C_2'$, and in particular
 show conditions \cref{cond:a,cond:a1,cond:b}. Condition \cref{cond:a1} still holds as $\Sigma$ has not changed. For \cref{cond:a}, we use \Cref{lem:sat}. For \cref{cond:b}, assuming $\beta\neq\diamond$, by hypothesis there are $C_1'',C_2''$ such that $\beta=\betais{C_1'',C_2''}$ and $(C_1'',C_2'')\satrans{\varepsilon}(C_1,C_2)$. By \Cref{def:sat} we have $( C_1, C_2)\satrans{\varepsilon}( C_1', C_2')$, and we then use rule \textsc{Trans}.
\item $\eta=\lopapp{i}{D[\vec\alpha]}$ with $\vec\alpha\fresh C_2$.
  Let $\vec\alpha'\fresh K_1,K_2,\vec\alpha$ be of the same length as $\vec\alpha$.
Then, setting $\Ed_i=C_i.\Ed$, $\tC_1\realtrans{\lopapp{i}{D[\vec\alpha']}}\tC_1'$ and, hence, $\tC_2\trans{\lopapp{i}{D[\vec\alpha']}}\tC_2'$ with $\tC_1'\simil\tC_2'$.
We obtain $C_2\trans{\lopapp{i}{D[\vec\alpha]}}C_2'$ with 
 $\tC_i'=\toS{(\vec\alpha\ \vec\alpha')\cdot C_i',(\Ed_i,K_i)}$. We set $\beta'=\betais{C_1',C_2'}$ and $\Sigma'=\Sigma[\beta'\trans{\Ed_1,\Ed_2}\beta]$.
To verify 
that $\Sigma'$ is sat-connected, by equivariance of $\Sigma$ (and of the sat-LTS) it suffices to show that the edge $\beta'\trans{\Ed_1,\Ed_2}\beta$ preserves it.
By condition \cref{cond:b} (on $\Sigma$) we have $\beta=\betais{C_{10},C_{20}}$  and $(C_{10},C_{20})\satrans{\varepsilon}( C_1, C_2)$ and, by \Cref{def:sat}, $(C_1,C_2)\satrans{(\Ed_1,\Ed_2)}(C_1',C_2')$, so \
$
(C_{10},C_{20})\satrans{\varepsilon}\cdot\satrans{(\Ed_1,\Ed_2)}( C_1', C_2')
$ \
as required.\\
It remains to show $C_1'\bisim*R_{\Sigma',\beta'}C_2'$, and in particular that conditions \cref{cond:a,cond:a1,cond:b} hold. 
For \cref{cond:b}, we have that $\beta'=\betais{C_1',C_2'}$ so
we require that $(C_1',C_2')\satrans{\varepsilon}(C_1',C_2')$, which is trivial.
We also note that \cref{cond:a} follows from \cref{cond:a1} since $\beta'=\betais{C_1',C_2'}$.
For \cref{cond:a1}, let $\beta'',C_i'',K_i''$ be such that $\beta''=\betais{C_1'',C_2''}\toSigsD[\Sigma']{K_1'',K_2''}$. We show that $\toS{C_1'',K_1''}\simil\toS{C_2'',K_2''}$ using rule induction on 
$\beta''\toSigsD[\Sigma']{K_1'',K_2''}$. Since $\beta''\not=\diamond$, we must have:
  \[
    \irule{
      \hat\beta\toSigsD[\Sigma']{\hat K_1,\hat K_2}
      \\
      \beta''\toSig[\Sigma']{\hat\Ed_1,\hat\Ed_2}\hat\beta
    }{
      \beta''\toSigsD[\Sigma']{K_1'',K_2''}
    }
  \]
  for some $\hat\beta,\hat K_i,\hat\Ed_i$ with $K_i''=\hat\Ed_i,\hat K_i$. In the base case, $\hat\beta=\hat\Ed_i=\diamond$ and $\hat K_i=K_i''={\noe}$.
  If $\beta''\toSig{\hat\Ed_1,\hat\Ed_2}\hat\beta$ then $\beta''\toSigsD{K_1'',K_2''}$ and the claim follows by \cref{cond:a1} applied on $\Sigma$ (as $C_1\bisim*R_{\Sigma,\beta}C_2$). Otherwise, we must have $\beta''=\pi\cdot\beta'$ (for some $\pi$) and $\beta=\Ed_i=\diamond$, and so the claim follows from $\tC_1'\simil\tC_2'$  and closure under $\pi$ (as $C_i'=C_i''$).
  Suppose now $\hat\beta\neq\diamond$. By sat-connectedness, $\hat\beta =\betais{\hat C_1,\hat C_2}$
for some $\hat C_i$ such that 
$(\hat C_1,\hat C_2)\satrans{\varepsilon}\cdot\satrans{(\hat\Ed_1,\hat\Ed_2)}( C_1'', C_2'')$.
  By induction hypothesis we have $\toS{\hat C_1,\hat K_1}\simil\toS{\hat C_2,\hat K_2}$, so by \cref{lem:sat} we have $\toS{ C_1'', K_1''}\simil\toS{ C_2'', K_2''}$. 
\item $\eta=\lpropret{D}$ and $\beta\toSig{\Ed_1,\Ed_2}\beta'$.
  Note first that, for any $K_i'$ such that $\beta'\toSigsD{K_1',K_2'}$, we get $A\cup\tC_1\trans{\lpropret{D}}\tC_1'$ and hence $A\cup\tC_2\wtrans{\lpropret{D}}\tC_2'$ with $\tC_1'\simil\tC_2'$, where  $\tC_i=\toS{C_i,(\Ed_i,K_i')}$ and
 $\tC_1'=A\cup\toS{C_1'\sub{\chi}{\Ed_1},K_1'}$ (by compatibility, these are all defined).
Picking now any such $K_1',K_2'$,
we obtain $C_2\wtrans{\lpropret{D[\vec\alpha]}}C_2'$ for $C_2'$ such that $\tC_2'=A\cup\toS{C_2'\sub{\chi}{\Ed_2},K_2'}$, so it remains to show that $C_1'\sub{\chi}{\Ed_1}\bisim*R_{\Sigma',\beta'} C_2'\sub{\chi}{\Ed_2}$ for $\Sigma'=\Sigma@\beta'$. Condition \cref{cond:a1} still holds as $\Sigma'$ is a subgraph of $\Sigma$; while we saw above that condition \cref{cond:a} holds ($\tC_1'\simil\tC_2'$ is true for any valid choice of $K_i'$).
  For \cref{cond:b}, if 
  $\beta'\not=\diamond$
then $\beta\neq\diamond$. Hence, condition \cref{cond:b} on $C_1\bisim*R_{\Sigma,\beta}C_2$ implies 
$
(C_{10},C_{20})\satrans{\varepsilon}(C_1,C_2)
$
for $\beta=\betais{C_{10},C_{20}}$. Since $\Sigma$ is sat-connected,
assuming  $\beta'=\betais{C_1'',C_2''}$, we get 
\[
(C_{1}'',C_{2}'')\satrans{\varepsilon}\cdot\satrans{(\Ed_1,\Ed_2)}( C_{10}, C_{20})
\]
and by \Cref{def:sat} 
we obtain $(C_{1}'',C_{2}'')\satrans{\varepsilon}( C_1'\sub\chi{\Ed_1}, C_2'\sub\chi{\Ed_2})$. 
\qed
\end{asparaitem}
\end{proof}


   \section{Theory of Enhancements}\label{sec:enhancements}

Here we present additional definitions and results from \cite{PousS11,PousCompanion} omitted from \cref{sec:up-to}.
The main result of this section is a set of proof obligations with which we can proof an up-to technique sound, shown in
\cref{lem:prfs}.

\begin{definition}
  Consider monotone functions $f,g: \pow{X}\rightarrow \pow{X}$ on some set $X$.
  We write $f\comp g$ for the composition of $f$ and $g$, and
  $f \sqcup g$ for the function $\bisim{S} \mapsto f(\bisim{S}) \sqcup g(\bisim{S})$.
  For any set $F$ of functions, we write $\bigsqcup F$ for the function $\bisim{S} \mapsto \bigcup_{f\in F}f(\bisim{S})$.
  We also write $\constf{X}$ to be the constant function with range $\{X\}$.
  We let $f^0 \defeq \utId{}$ and $f^{n+1} \defeq f\comp f^{n}$. Moreover, we write
  $f^{\omega}$ to mean $\bigsqcup_{k<\omega} f^{k}$.
  We write $f\sqsubseteq g$ when, for all $\bisim{S}\in\pow{X}$, $f(\bisim{S}) \subseteq g(\bisim{S})$.
\end{definition}

\begin{lemma}[\cite{PousS11}, Lem. 6.3.12] \label{lem:compat-alternative}
  $f\progress[\WP]g$ if and only if 
  for all ${\bisim{R}}\progress[\WP]{\bisim{S}}$ we have $f\comp\WP(\bisim{R}) \subseteq \WP\comp g(\bisim{S})$.
  \qed
\end{lemma}

\begin{lemma}[\cite{PousS11}, Prop. 6.3.11 and 6.3.12]\label{lem:pous-basic}
  The following functions are $\WP$-compatible:
  \begin{itemize}
    \item the reflexive \utRefl{} and identity \utId{} functions;
    \item $f \comp g$, for any $\WP$-compatible monotone functions $f$, $g$;
    \item $\bigsqcup F$, for any set $F$ of $\WP$-compatible monotone functions.\qed
    \end{itemize}
\end{lemma}

Pous \cite{PousCompanion} extends the theory of enhancements with the notion of companion of $\WP$, the largest 
 $\WP$-compatible function.
\begin{definition}[Companion]
  $\companion[\WP] \defeq \bigsqcup\{ f:\mathcal{P}(X)\to\mathcal{P}(X) \mid  {f\progress[\WP]f} \}$.
\end{definition}

\begin{lemma}[\cite{PousCompanion}]\label{lem:companion-props}~
  \begin{enumerate}
    \item\label{lem:companion-props-1}
      \protect{\companion[\WP]} is $\WP$-compatible: $\companion[\WP]\progress\companion[\WP]$;
    \item\label{lem:companion-props-2}
      $\WP$ is $\WP$-compatible: $\WP\sqsubseteq\companion[\WP]$;
    \item\label{lem:companion-props-3}
      \protect{\companion[\WP]} is idempotent: ${\utId{}}\sqsubseteq\companion[\WP]$ and 
      $\companion[\WP]\comp\companion[\WP]\sqsubseteq\companion[\WP]$;
    \item\label{lem:companion-props-5}
      \companion[\WP] is $\WP$-sound: $\gfp{\WP\comp\companion[\WP]}\subseteq \gfp\WP$.
      \qed
  \end{enumerate}
\end{lemma}

This gives rise a proof technique for proving up-to techniques sound.
\begin{lemma}\label{lem:subseteq-companion-sound}
  Let $f\sqsubseteq\companion[\WP]$. Then $f$ is $\WP$-sound.
\end{lemma}
\begin{proof}
  By showing that $f\cup\companion[\WP]\progress[\WP]f\cup\companion[\WP]$ and using \cref{lem:compat-sound}.
\qed\end{proof}

\begin{lemma}[Function Composition Laws]\label{lem:fcl}
  Consider monotone functions $f,g,h: \pow{X}\rightarrow \pow{X}$ and set ${\bisim{S}}\in\pow{X}$. We have
  \begin{enumerate}
    \item $\constf{\bisim{S}} \comp f = \constf{\bisim{S}}$
    \item $(f\sqcup g)\comp h = (f\comp h) \sqcup (g \comp h)$
    \item $h\comp(f\sqcup g) = (h\comp f) \sqcup (h\comp g)$
    \item $(f\sqcup g) \sqsubseteq (f\sqcup h)$ and
          $(f\comp g) \sqsubseteq (f\comp h)$ and
          $(g\comp f) \sqsubseteq (h\comp f)$, when $g\sqsubseteq h$.
    \item
      $f \sqsubseteq f^\omega$ and $f \comp f^\omega = f^\omega \comp f \sqsubseteq f^\omega \comp f^\omega \sqsubseteq f^\omega$.
      \item $f^\omega\circ g = \bigsqcup_{i<\omega}(f^i\circ g)$.
      \qed
  \end{enumerate}
\end{lemma}

We distil this up-to technique to the following three proof obligations, each sufficient for proving the soundness of up-to techniques.
\begin{lemma}[POs for Up-To Soundness]\label{lem:prfs}
  Let $f$ be a monotone function and $\bisim{R}$ be a weak simulation; $f$ is $\WP$-sound when one of the following holds:
  \begin{enumerate}
    \item \label{lem:prfs-1}
      $f\progress[\WP]f$; or

    \item \label{lem:prfs-1.2}
      $f\progress[\WP]f^\omega$; or

    \item \label{lem:prfs-2}
      $f\progress[\WP](f\comp g)$, for some $g\sqsubseteq\companion[\WP]$; or


    \item \label{lem:prfs-3}
      $f = \bigsqcup_{f_i\in F}f_i\comp{\constf{\gfp\WP}}$,
      where $F$ is a set of monotone functions and, 
      for all $f_i\in F$, there exists $g_i\sqsubseteq\companion[\WP]$ such that
      $f_i\comp{\constf{\gfp\WP}}\progress[\WP](f\sqcup g_i)^{\omega}\comp{\constf{\gfp\WP}}$.

 \end{enumerate}
\end{lemma}
\begin{proof}~
  \begin{enumerate}
    \item By \cref{lem:compat-sound}.
    \item 
      By \cref{lem:subseteq-companion-sound}, it suffices to show $f \sqsubseteq \companion[\WP]$.
      Because
        $f \sqsubseteq f^\omega$,
      it suffices to show that 
      $f^\omega \progress[\WP] f^\omega$.
      This is proven by showing that for all $k$,
      \begin{equation}\tag{$P(k)$}
        f^k \progress[\WP] f^\omega.
        \end{equation}
      We proceed by induction on $k$. The base case is straightforward:
      \begin{align*}
        f^0\comp\WP&=
         \utId{}\comp\WP =
         \WP\comp \utId{} =
         \WP\comp f^0 \sqsubseteq
         \WP\comp f^\omega
      \end{align*}
In the inductive case we assume $P(k)$ and prove $P(k+1)$ as follows:
      \begin{align*}
        f^{k+1}\comp\WP
        &=
        f\comp f^{k}\comp\WP
        \\&
        \sqsubseteq
        f\comp\WP\comp f^\omega
        &(P(k))
        \\&
        \sqsubseteq
        \WP\comp f^\omega\comp f^\omega
        &(f\progress[\WP]f^\omega)
        \\&
        =
        \WP\comp f^\omega
      \end{align*}

    \item
      By \cref{lem:subseteq-companion-sound}, it suffices to show $f \sqsubseteq \companion[\WP]$.
      Because $f\sqsubseteq f \comp(\utId{}\sqcup\companion[\WP])\sqsubseteq f\comp\companion[\WP]$,
      it suffices to show $f\comp\companion[\WP]\progress[\WP]f\comp\companion[\WP]$ by 
      unfolding definitions and the premise:
      \begin{align*}
        f \comp \companion[\WP] \comp \WP \sqsubseteq
        f \comp \WP \comp \companion[\WP] \sqsubseteq
        \WP\comp f \comp g \comp \companion[\WP] \sqsubseteq
        \WP\comp f \comp \companion[\WP].
      \end{align*}

    \item 
      Let $g=\bigsqcup_{f_i\in F}g_i$.
      By \cref{lem:subseteq-companion-sound}, it suffices to show $f \sqsubseteq \companion[\WP]$.
      Because
      \begin{align*}
        f_i\comp{\constf{\gfp\WP}}
        &=
        f_i\comp{\constf{\gfp\WP}}\comp{\constf{\gfp\WP}}
        \sqsubseteq
        f\comp{\constf{\gfp\WP}}
        \sqsubseteq
        (f\comp{\constf{\gfp\WP}}) \sqcup (\companion[\WP]\comp{\constf{\gfp\WP}})
        \\&=
        (f\sqcup\companion[\WP])\comp{\constf{\gfp\WP}}
        \sqsubseteq
        (f\sqcup\companion[\WP])^\omega\comp{\constf{\gfp\WP}}
      \end{align*}
      it suffices to show that 
      $(f\sqcup\companion[\WP])^\omega\comp{\constf{\gfp\WP}} \progress[\WP] (f\sqcup\companion[\WP])^\omega\comp{\constf{\gfp\WP}}$.
      This is proven by showing that for all $k$,
      \begin{equation}\tag{$P(k)$}
        (f\sqcup\companion[\WP])^k\comp{\constf{\gfp\WP}} \progress[\WP] (f\sqcup\companion[\WP])^\omega\comp{\constf{\gfp\WP}}.
        \end{equation}
      We proceed by induction on $k$. The base case is straightforward:
$$
\utId{}\circ\constf{\gfp\WP}\circ\WP =\constf{\gfp\WP} = \WP\circ \utId{}\circ\constf{\gfp\WP}\sqsubseteq
\WP\circ(f\sqcup\companion[\WP])^\omega\comp{\constf{\gfp\WP}}
$$
      In the inductive case we assume $P(k)$ and prove $P(k+1)$ as follows:
      \begin{align*}
        &(f\sqcup\companion[\WP])^{k+1}\comp{\constf{\gfp\WP}}\comp\WP
        \\&=
        (f\sqcup\companion[\WP])\comp(f\sqcup\companion[\WP])^{k}\comp{\constf{\gfp\WP}}\comp\WP
        \\&
        \sqsubseteq
        (f\sqcup\companion[\WP])\comp\WP\comp h
        &\text{($P(k),~h=(f\sqcup\companion[\WP])^\omega\comp{\constf{\gfp\WP}}$)}
        \\&
        \sqsubseteq
        (f\comp\WP\comp h) \sqcup (\companion[\WP]\comp\WP\comp h) 
        &\text{(\cref{lem:fcl})}
        \\&
        =
        \left(\bigsqcup_{f_i\in F}(f_i\comp{\constf{\gfp\WP}}\comp\WP\comp h)\right) \sqcup (\companion[\WP]\comp\WP\comp h)
        &\text{(definition of $f$ and \cref{lem:fcl})}
        \\&
        \sqsubseteq
        \left(\bigsqcup_{f_i\in F}(\WP\comp(f\sqcup g_i)^\omega\comp{\constf{\gfp\WP}}\comp h)\right)
        \sqcup (\companion[\WP]\comp\WP\comp{h})
        &\text{(premise)}
        \\&
        \sqsubseteq
        \left(\bigcup_{f_i\in F}(\WP\comp(f\sqcup \companion[\WP])^\omega\comp{\constf{\gfp\WP}})\right)
        \sqcup (\companion[\WP]\comp\WP\comp h)
        &\text{(\cref{lem:fcl} and premise on $g_i$)}
        \\&
        \sqsubseteq
        (\WP\comp(f\sqcup \companion[\WP])^\omega\comp{\constf{\gfp\WP}}) \sqcup (\WP\comp\companion[\WP]\comp h) 
        &\text{(\cref{lem:companion-props}~(\ref{lem:companion-props-1}))}
        \\&
        \sqsubseteq
        (\WP\comp\utId{}\comp h) \sqcup (\WP\comp\companion[\WP]\comp h) 
        &\text{(definition of $h$)}
        \\&
        =
        \WP\comp(\utId{}\sqcup\companion[\WP])\comp h
        &\text{(\cref{lem:fcl})}
        \\&
        = \WP\comp\companion[\WP]\comp h
        &\text{(\cref{lem:companion-props}~(\ref{lem:companion-props-3}))}
        \\&
        \sqsubseteq
        \WP\comp(f\sqcup\companion[\WP])\comp h
        &\text{(\cref{lem:fcl})}
        \\&
        \sqsubseteq
        \WP\comp(f\sqcup\companion[\WP])^\omega\comp {\constf{\gfp\WP}}
        &\text{(\cref{lem:fcl} and definition of h )}
         \\&&
        \qedhere
      \end{align*}
  \end{enumerate}\qed
\end{proof}

As we are only interested in weak progression, in the following we drop the $\WP$ annotation from progressions, compatibility and companion.

   \section{Simple Up-To Techniques}\label{app:simple}
We develop our up-to techniques using the theory of bisimulation enhancements from \cite{PousS11,PousCompanion} (see~Appendix~\ref{sec:enhancements}).
We start by presenting three straightforward up-to techniques which are needed to reduce the configurations considered by bisimulation, achieving finite LTSs in many examples.
These techniques are \emph{up to permutations}, \emph{beta reductions}, \emph{garbage collection}, and \emph{weakening of knowledge environments}.
To present these techniques we first need the following definitions.

\subsection{Up to Permutations}
\begin{definition}[Permutations]
We consider permutations of store locations, $\pi_l$, abstract names, $\pi_\alpha$ and environment indices, $\pi_i$, respectively.
We use juxtaposition to denote permutation composition.

When applying a permutation $\pi_i$ to an environment $\Gamma$, it only acts on its domain; other types of permutations only act on the codomain of $\Gamma$.
When applying a permutation $\pi_l$ to a store $s$, the former acts on both the domain and range of the latter; a permutation $\pi_\alpha$ acts on the codomain of $s$, and a permutation $\pi_i$ leaves $s$ unaffected.
When $\pi_1=\pi_{l_1}\pi_\alpha\pi_i$ and $\pi_2=\pi_{l_2}\pi_\alpha\pi_i$
and $\beta=\betais{C_1,C_2}$, we define
  \begin{gather*}
    (\pi_1,\pi_2)\cdot\beta = \betais{\pi_1 C_1, \pi_2 C_2}
    \\
    (\pi_1,\pi_2)\cdot\Sigma =
      \{((\pi_1,\pi_2)\cdot\beta', \pi_1\cdot\Ed_1,\pi_2\cdot\Ed_2,(\pi_1,\pi_2)\cdot\beta) \where (\beta',\Ed_1,\Ed_2,\beta)\in\Sigma\}
  \end{gather*}

  Moreover if $\pi_1=\pi_2=\pi$ we write $\pi\cdot\beta$ and $\pi\cdot\Sigma$ to mean $(\pi,\pi)\cdot\beta$ and $(\pi,\pi)\cdot\Sigma$, respectively.
\end{definition}
Note that call graphs $\Sigma$ are closed under $\pi_\alpha$ permutations, and thus are unaffected by such. However they are affected by $\pi_i$ and $\pi_l$ permutations.

\begin{lemma}[Permutation Invariance for Reductions]\label{lem:red-perm}
  Let $\pi_l,\pi_\alpha,\pi_i$ be permutations on locations,  abstract names and indices respectively, and $\pi=\pi_l\pi_\alpha\pi_i$.
  If $\redconf{s}{e} \redbase \redconf{s'}{e'}$ then
  $\pi\cdot\redconf{s}{e} \redbase \pi\cdot\redconf{s'}{e'}$.
  Moreover,
  if $\redconf{s}{e} \red \redconf{s'}{e'}$ then
  $\pi\cdot\redconf{s}{e} \red \pi\cdot\redconf{s'}{e'}$.
\end{lemma}
\begin{proof}
  By nominal sets reasoning (all reduction rules are closed under permutation).
\qed\end{proof}
\begin{lemma}\label{lem:lts-perm}
  Let $\pi_l$, $\pi_\alpha$, and $\pi_i$ be
  permutations on locations,  abstract names, and indices, respectively,
 and $\pi=\pi_l\pi_\alpha\pi_i$.
  If $C \trans{\eta} C'$ then
  $\pi\cdot C \trans{\pi_\alpha\pi_i\cdot \eta} \pi\cdot C'$.
\end{lemma}
\begin{proof}
    By nominal sets reasoning (all transition rules are closed under permutation).
\qed\end{proof}
\begin{lemma}\label{lem:lts-permute}
  Let $C \trans{\eta} C'$; then
  for all finite $L_0,A_0,I_0$ there exists $\pi$ such that
 $C \trans{\pi\cdot\eta} \pi\cdot C'$
  and 
  \begin{align*}
    (\an{\pi\cdot C'}\sdif \an{C})\cap A_0 &= (\dom{\pi\cdot C'.s}\sdif\dom{C.s})\cap L_0 \\
    &= (\dom{\pi\cdot C'.\Gamma}\sdif\dom{C.\Gamma})\cap I_0=\emptyset
  \end{align*}
\end{lemma}
\begin{proof}
  By \cref{lem:lts-perm}, picking permutations $\pi$ that rename new names in $C'$ that do not exist in $C$ to fresh ones, and observing that $\pi\cdot C=C$.
\qed\end{proof}
\begin{corollary}\label{lem:lts-permute-weak}
  Let $C \wtrans{\eta} C'$; then
  for all finite $L_0,A_0,I_0$ there exists $\pi$ such that
 $C \trans{\pi\cdot\eta} \pi\cdot C'$
  and 
  \begin{align*}
    (\an{\pi\cdot C'}\sdif \an{C})\cap A_0 &= (\dom{\pi\cdot C'.s}\sdif\dom{C.s})\cap L_0 \\
    &= (\dom{\pi\cdot C'.\Gamma}\sdif\dom{C.\Gamma})\cap I_0=\emptyset
  \end{align*}
\end{corollary}
\begin{proof}
  By induction on the length of the transition from $C$, using \cref{lem:lts-permute}.
\qed\end{proof} 
\begin{figure*}[t]
  \[\begin{array}{@{}c@{}}
    \irule*[UpToPerm][uptoperm]{
      C_1
      \bisim*{R}_{\Sigma,\beta}
      C_2
      \\
      \pi_1=\pi_{l_1}\pi_\alpha\pi_i
      \\
      \pi_2=\pi_{l_2}\pi_\alpha\pi_i
    }{
      \pi_1\cdot C_1
      \utperm*{\bisim{R}}_{(\pi_1,\pi_2)\cdot\Sigma,(\pi_1,\pi_2)\cdot\beta}
      \pi_2\cdot C_2
    }
  \end{array}\]
  \vspace{-1mm}
  \hrule
  \vspace{-1mm}
  \caption{Up to Permutations.}\label{fig:utperm}
\end{figure*}
\begin{lemma}\label{lem:upto-perm}
   Function $\utperm{}$ is a sound up-to technique.
\end{lemma}
\begin{proof}
  From \cref{lem:prfs}~(\ref{lem:prfs-1}), it suffices to show that $\utperm{}$ is compatible; i.e.,
  ${\utperm{\WP(\bisim{R})}} \subseteq \WP(\utperm{\bisim{R}})$, for any configuration relation \bisim{R}.

  Let $C_1 \mathrel{\WP(\bisim{R})}_{\Sigma,\beta} C_2$ and 
  $\pi_1\cdot C_1 \utperm*{\bisim{R}}_{(\pi_1,\pi_2)\cdot\Sigma,(\pi_1,\pi_2)\cdot\beta} \pi_2\cdot C_2$,
  where $\pi_1=\pi_{l1}\pi_\alpha\pi_i$ and $\pi_2=\pi_{l2}\pi_\alpha\pi_i$.
  Moreover, let $\pi_1\cdot C_1 \realtrans{\eta} C_1'$.
  Because of $\pi_1 \pi_1 = \mathsf{id}$ and \cref{lem:lts-perm} we get
  $C_1 \realtrans{\pi_\alpha\pi_i\cdot\eta} C_1'\pi_1$.
  We proceed by definition of $\WP(\bisim{R})$, taking cases on $\eta$.

  When $\eta\in\{\tau,\lpropapp{\alpha}{D},\lopret{D[\vec\alpha]} \where \vec\alpha\fresh C_2\}$,  there exists $C_2'$ such that
  $C_2 \wtrans{\pi_\alpha\pi_i\cdot\eta} C_2'$ and $\pi_1\cdot C_1' \bisim*{R}_{\Sigma,\beta} C_2'$. By \cref{lem:lts-perm},
  $\pi_2\cdot C_2 \wtrans{\eta} \pi_2\cdot C_2'$, and by definition of $\utperm{\bisim{R}}$:
  $C_1' \utperm*{\bisim{R}}_{(\pi_1,\pi_2)\cdot\Sigma,(\pi_1,\pi_2)\cdot\beta} \pi_2\cdot C_2'$.

  When $\eta=\lopapp{i}{D\hole[\vec \alpha]}$, there exists $C_2'$ such that
  $C_2 \wtrans{\pi_\alpha\pi_i\cdot\eta}  C_2'$
  and
  $\pi_1\cdot C_1' \bisim*{R}_{\Sigma',\beta'} C_2'$
  with $\beta' = \betais{\pi_1\cdot C_1',C_2'}$ and $\Sigma' = \Sigma[\beta' \mapsto (\pi_1\cdot C_1.\Ed, C_2.\Ed, \beta)]$.
  By \cref{lem:lts-perm},
  $\pi_2\cdot C_2 \wtrans{\eta}  \pi_2\cdot C_2'$
  and by definition of $\utperm{\bisim{R}}$:
  $C_1' \bisim*{R}_{(\pi_1,\pi_2)\cdot\Sigma',(\pi_1,\pi_2)\cdot\beta'} \pi_2\cdot C_2'$.

  When $\eta=\lpropret{D}$, there exists 
  $(\Edb_1,\Edb_2,\beta') \in \Sigma(\beta)$
  and exists
  $C_2'$ such that
  $C_2 \wtrans{\pi_\alpha\pi_i\cdot\eta} C_2'$ and
  $\pi_1\cdot C_1'\sub{\chi}{\Edb_1} \bisim*{R}_{\Sigma',\beta'} C_2'\sub{\chi}{\Edb_2}$
  with $\Sigma' = \Sigma@\beta'$.
  By \cref{lem:lts-perm},
  $\pi_2\cdot C_2 \wtrans{\eta} \pi_2\cdot C_2'$, and by definition of $\utperm{\bisim{R}}$:
  $$C_1'\sub{\chi}{\pi_1\cdot\Edb_1}=\pi_1\cdot (\pi_1\cdot C_1'\sub{\chi}{\Edb_1})
    \utperm*{\bisim{R}}_{(\pi_1,\pi_2)\cdot\Sigma',(\pi_1,\pi_2)\cdot\beta'}
    \pi_2\cdot(C_2'\sub{\chi}{\Edb_2})=\pi_2\cdot C_2'\sub{\chi}{\pi_2\cdot\Edb_2}$$
  with $(\pi_1,\pi_2)\cdot\beta' = \betais{C_1',\pi_2\cdot C_2'}$ and
  $$(\pi_1,\pi_2)\cdot\Sigma' = (\pi_1,\pi_2)\cdot\Sigma[(\pi_1,\pi_2)\cdot\beta' \mapsto (\pi_1\cdot C_1.\Ed, \pi_2\cdot C_2.\Ed, (\pi_1,\pi_2)\cdot\beta)]$$
\qed
\end{proof}

\subsection{Up to Beta Moves}
\begin{lemma}[beta-move]\label{def:beta-move} 
  Any $\tau$-transition $C \realtrans{\tau} C'$ is called a \emph{beta-move}, and
  we write $C \betatrans C'$, because for all transitions $C \realtrans{\eta} C''$,
    we have $\eta=\tau$ and $C'=C''\pi_l$, for some location permutations $\pi_l$ on the locations in $\fl{C',C''}\setminus\fl{C}$.
\end{lemma}
\begin{proof}
  By the deterministic nature of the reduction semantics.
\qed\end{proof}
\begin{figure*}[t]
  \[\begin{array}{@{}c@{}}
     \irule*[UpToBeta][uptobeta]{
      C_1'
      \bisim*{R}_{\Sigma,\beta}
      C_2'
      \\\\
       C_1 \betatrans C_1' \text{ or } C_1 = C_1'
      \\\\
      C_2 \betatrans C_2' \text{ or } C_2 = C_2'
    }{
      C_1
      \utbeta*{\bisim{R}}_{\Sigma,\beta}
      C_2
    }
  \end{array}\]
  \vspace{-1mm}
  \hrule
  \vspace{-1mm}
  \caption{Up to Beta Moves.}\label{fig:utbeta}
\end{figure*}
\begin{lemma}\label{lem:upto-beta}
  Function $\utbeta{}$ is a sound up-to technique.
\end{lemma}
\begin{proof}
  From \cref{lem:prfs}~(\ref{lem:prfs-1}), it suffices to show that $\utbeta{}\progress[\WP]\utbeta{}$.
  Let $C_1 \utbeta*{\WP(\bisim{R})}_{\Sigma,\beta} C_2$
  and $C_1' \mathrel{\WP(\bisim{R})}_{\Sigma,\beta} C_2'$
  and $C_1 \betatrans C_1'$ or $C_1 = C_1'$
  and $C_2 \betatrans C_2'$ or $C_2 = C_2'$.
  
  When
  $C_1 \betatrans C_1'$, then by \cref{def:beta-move} $C_1$ can only take this real transition.
  This can be matched with the transition $C_2 \wtrans{\tau}C_2'$. Moreover,
  $C_1' \utbeta*{\WP(\bisim{R})}_{\Sigma,\beta} C_2'$ by definition.
  
  When $C_1 = C_1'$, let $C_1\realtrans{\eta} C_1''$, thus $C_1'\realtrans{\eta} C_1''$.
  By definition of $\WP(\bisim{R})$,
  there exists $C_4$,$\Sigma'$,$\beta',\Ed_1,\Ed_2$ such that
  $C_2' \wtrans{\eta} C_2''$
  and $C_1'' \bisim*{R}_{\Sigma',\beta'} C_2''$
  or $C_1''\sub{\chi}{\Ed_1} \bisim*{R}_{\Sigma',\beta'} C_2''\sub{\chi}{\Ed_2}$.
  Moreover, 
  $C_2 \wtrans{\tau} C_2' \wtrans{\eta} C_2''$.
  and $C_1'' \utbeta*{\bisim{R}}_{\Sigma',\beta'} C_2''$
  or $C_1''\sub{\chi}{\Ed_1} \utbeta*{\bisim{R}}_{\Sigma',\beta'} C_2''\sub{\chi}{\Ed_2}$ by definition.
\qed\end{proof}

\subsection{Up to Garbage Collection}

\begin{definition}[Garbage Collection] \label{def:asymp}
  Given a set of location names $S$,
  we define the following total operation on configurations:
  \[
    \cconf{A\uplus A_g}{\Gamma}{s,s_g}{\Phi}
    \gcrelinv_{S}
    \left\{\begin{array}{ll}
      \cconf{A}{\Gamma}{s}{\Phi}
      & \text{if }
      S=\dom{s_g},~
      \dom{s_g}\cap\fl{\Gamma,s,\Ed}=\emptyset
      \\
      \cconf{A\uplus A_g}{\Gamma}{s,s_g}{\Phi}
      &\text{otherwise}
    \end{array}\right.
  \]
  Moreover ${\asympgc_S}={\gcrel_S} \cup {\gcrelinv_S}$.
Given $S_1,S_2,\Sigma$ with $S_i\cap\fl{\Sigma.\Ed_i}=\emptyset$, 
  we also define $\Sigma^{gc}_{S_1,S_2}$ as follows:
  \begin{align*}
    \betais{C'_{11},C'_{21}}\toSig[\Sigma^{gc}_{S_1,S_2}]{\Ed_1,\Ed_2}\betais{C'_{12},C'_{22}}
  \end{align*}
  when
    $\betais{C_{11},C_{21}}\toSig{\Ed_1,\Ed_2}\betais{C_{12},C_{22}}$ and,
      for $i,j\in\{1,2\}$, we have $C_{ij}' \gcrel_{S_i} C_{ij}$.
\end{definition}
\begin{lemma}\label{lem:asymp-perm}
  If $C \asympgc_{S} C'$ then $\pic C \asympgc_{\pic S} \pic C'$.
  \qed
\end{lemma}
\begin{lemma}\label{lem:asymp-bisim}\label{lem:asymp-bisim-weak}
  Given $C_1\asympgc_{S} C_2$ and       $(\fl{C_1'}\sdif\fl{C_1})\cap \fl{S} = \emptyset$:
  \begin{itemize}
    \item if
 $C_1 \trans{\eta} C_1'$ then
  $C_{2} \trans{\eta} C_{2}'$ and $C_1' \asympgc_{S} C_2'$ 

\item if $C_1 \wtrans{\eta} C_1'$ then
  $C_2 \wtrans{\eta} C_2'$ and $C_1' \asympgc_{S} C_2'$.
\end{itemize}
\end{lemma}
\begin{proof}
First part by induction on the derivation of $C_1 \asympgc_{S} C_2$ and case analysis on the transition from $C_1$. Second part 
  by induction on the length of the transition from $C_1$ and using first part.
\qed\end{proof}
\begin{figure*}[t]
  \[\begin{array}{@{}c@{}}
    \irule*[UpToGC][uptogc]{
      C_1
      \gcrelinv_{S_1}
      \bisim*{R}_{\Sigma^{gc}_{S_1,S_2},\betais{C_3',C_4'}}
      \gcrel_{S_2}
      C_2
      \\
      C_3'\gcrel_{S_1}C_3
      \\
      C_4'\gcrel_{S_2}C_4
    }{
      C_1
      \utgc*{\bisim{R}}_{\Sigma,\betais{C_3,C_4}}
      C_2
    }
  \end{array}\]
  \vspace{-1mm}
  \hrule
  \vspace{-1mm}
  \caption{Up-to garbage collection.}\label{fig:utsimple}
\end{figure*}
\begin{lemma}\label{lem:utgc-sound} 
  Function $\utgc{}$ is a sound up-to technique.
\end{lemma}
\begin{proof}
  By showing $\utgc{}\progress\utgc{}\comp\utperm{}$ and
  \cref{lem:prfs}~(\ref{lem:prfs-2}) and \cref{lem:asymp-bisim,lem:upto-perm}. 
\qed\end{proof}

\subsection{Up to Opponent Knowledge Weakening}

\begin{lemma}\label{lem:lts-weaken}
  Let $C_1$ and $C_2$ be well formed configurations with $C_2 = C_1,\maps i v$, meaning that $C_2$ is identical to $C_1$ except it contains an additional value $v$ indexed by $i$ in $C_2.\Gamma$.
  Then the following hold:
  \begin{enumerate}
    \item
      If $C_1 \trans{\eta} C_1'$,
      where $\eta\not\in\{\lpropapp{\alpha}{i},\lpropret{i}\where \text{any } \alpha\}$,
      then
      $C_2 \trans{\eta} C_1',\maps i v$.
    \item
      If $C_2 \trans{\eta} C_2',\maps i v$,
      where $\eta\not\in\{\lopapp{i}{\alpha}\where \text{any } \alpha\}$,
      then $C_1 \trans{\eta} C_2'$.

  \end{enumerate}
\end{lemma}
\begin{proof}
By case analysis on the transitions.
\qed\end{proof}
\begin{definition}[Opponent Knowledge Weakening] \label{def:gamma-weak}
  We let $(\weakrel)$ as follows:
  \begin{itemize}

    \item $\oconf{A}{\Gamma}{s}{\Ed} \weakrel_i \oconf{A}{\Gamma,\maps{i}{v_1}}{s}{\Ed}$

    \item $\pconf{A}{\Gamma}{s}{e} \weakrel_i \pconf{A}{\Gamma,\maps{i}{v_1}}{s}{e}$

  We also define $\Sigma^{wk}_i$:
  \begin{align*}
    (\Ed_1,\Ed_2,\betais{C'_{12},C'_{22}})\in \Sigma^{wk}_i(\betais{C'_{11},C'_{21}})
  \end{align*}
  when
  \begin{itemize}
    \item
      $(\Ed_1,\Ed_2,\betais{C_{12},C_{22}})\in \Sigma(\betais{C_{11},C_{21}})$; and
    \item
      for $k,j\in\{1,2\}$, we have $C_{kj}' \weakrel_{i} C_{kj}$.
  \end{itemize}

  \end{itemize}
  \defqed
\end{definition}
\begin{figure*}[t]
  \[\begin{array}{@{}c@{}}
    \irule*[UpToWeakening][uptoweakening]{
      C_1 \weakrel_i \bisim*{R}_{\Sigma,\betais{C_3',C_4'}} \weakrelinv_i C_2
      \\
      C_3\weakrel_i C_3'
      \\
      C_4\weakrel_i C_4'
    }{
      C_1
      \utweak[i]*{\bisim{R}}_{\Sigma^{wk}_i,\betais{C_3,C_4}}
      C_2
    }
  \end{array}\]
  \vspace{-1mm}
  \hrule
  \vspace{-1mm}
  \caption{Up to Weakening of the Opponent Knowledge.}\label{fig:utweak}
\end{figure*}

\begin{lemma}\label{lem:utweak-sound}
  Function $\utweak[i]{}$ is a sound up-to technique.
\end{lemma}
\begin{proof}
  By showing $\utweak{}\progress[\WP]\utweak{}\comp\utperm{}$ and
  \cref{lem:prfs}~(\ref{lem:prfs-2}), using \cref{lem:lts-weaken,lem:upto-perm}.
\qed\end{proof}

%



\section{Pair (Bi-)Simulation}\label{sec:pair-bisim}

In order to define and prove our up to separation technique, we need to extend the stackless LTS to pairs of configurations and  define a notion of bisimulation over it. 

\begin{definition}[Pair Configuration]
  We define pair configurations $\bconf{C_1}{C_2}$ for all stackless LTS configurations $C_1$, $C_2$.
  To enable symmetric reasoning we define $\co{1} \defeq 2$ and $\co{2}\defeq 1$.
  We let $\bC$ range over pair configurations, and write $\bC.i$ to get the $i$'th inner configuration.
  We write $\bC\downarrow$ when $\bC.1\downarrow$ and $\bC.2\downarrow$.
\end{definition}

\begin{definition}[Pair LTS]
  We extend the LTS of \cref{fig:SLlts} to pair configurations as follows:

  \begin{align*}
    \begin{array}{rlrlll}
    \irule!{
      C_i \trans{\eta} C_i'
      \text{ and }
      C_{\co{i}} = C'_{\co{i}}
      \text{ an opponent configuration and }
      i\in\{1,2\}
    }{
      &\bconf{C_1}{C_2} \trans{\eta,i} \bconf{C_1'}{C_2'}
    }
    \\
    \irule!{
      \nbox{
      C_1 \trans{\eta} C_1'
      \text{ and }
      C_2 \trans{\eta} C_2'
      \text{ and }
      C_1.e = C_2.e \text{ or } C_1.\Ed = C_2.\Ed 
      \\\hfill\text{ and }
      C_1'.e = C_2'.e \text{ or } C_1'.\Ed = C_2'.\Ed 
    }
    }{
      &\bconf{C_1}{C_2} \trans{\eta,0} \bconf{C_1'}{C_2'}
    }
    \end{array}
  \end{align*}
  We also write $\bC \realtrans{\eta,k} \bC'$ when the transition is derived from the above LTS, without using the \iref{dummy} rule of the stackless LTS in \cref{fig:SLlts}.
\end{definition}

The following definition lifts \cref{def:betaSigma} to the pair LTS.

\begin{definition}[Pair Entry Points and Pair Continuation Graphs]
  \begin{align*}
    \mathsf{PEPoint} \ni \ \bbeta &::=\ (\beta_1, \beta_2, k)  \qquad (k\in\{0,1,2\})\\
    \mathsf{PCGrph}\ni \ \bSigma &\subseteq_{\rm fin}^{\neq\emptyset} \mathsf{PEPoint}\times\Kont\times\Kont\times\mathsf{PEPoint}
  \end{align*}
 and each $\bSigma$ must satisfy the conditions: 
  \begin{compactitem}
  \item{} \emph{Reachability}.~For all $\bbeta\in\dom{\bSigma}$ there are evaluation stacks $K_1,K_2$ such that $\bbeta\toSigs[\bSigma]{K_1,K_2}\diamond$

  \item \emph{Top and Divergence}.~For all $\bbeta'^{k'}\toSig[\bSigma]{\Ed_1,\Ed_2}\bbeta^{k}$ and $j\in\{1,2\}$:
  \begin{align*}
    &    (\diamond,\diamond,0)\in\dom{\bSigma}\land (\bbeta'=(\diamond,\diamond,0)\implies\bbeta=(\diamond,\diamond,0)) \land
    (\bbeta=(\diamond,\diamond,0)\iff\Ed_j=\diamond) \\
&{}\land    (\bbeta'.k.j=\bot\iff \Ed_j=\bot) \land
(\bbeta.k.j=\bot\implies
    \Ed_j=\bot)
  \end{align*}
    where, if $\bbeta = (\beta_1,\beta_2, k)$ and $i \in\{1,2\}$, we write $\bbeta.i.1=\bot$ when $\beta_i=\betais{\botcconf,C}$ 
    and $\bbeta.i.2=\bot$ when $\beta_i=\betais{C,\botcconf}$.
  \item \emph{Nominal closure}.~For  all $\bbeta'\toSig[\bSigma]{\Ed_1,\Ed_2}\bbeta$ and permutations $\pi$, \
    $\pi\cdot\bbeta'\toSig[\bSigma]{\pi\cdot\Ed_1,\pi\cdot\Ed_2}\pi\cdot\bbeta$.
\end{compactitem}

Finally, we lift \cref{def:sigma-ext-restr} to pair graphs obtaining continuation graph extension
  $\bSigma[\bbeta'\mapsto(\Ed_1,\Ed_2,\bbeta'')]$
and restriction $\bSigma@\bbeta$.
We will write $\bbeta.i$ to mean $\beta_i$ ($i\in\{1,2\}$), and $\bbeta^m$ to mean $k=m$, when $\bbeta = (\beta_1,\beta_2,k)$.
\end{definition}

We define simulation on compatible pair configurations.

\begin{definition}[Compatible Pair Configurations and Pair (Bi)simulation Tuples]
  Configurations $\bconf{C_1}{C_2}$ and $\bconf{C_1'}{C_2'}$ are compatible, when $C_j$ and $C_j'$ are compatible according to \cref{def:compatible-conf}, for $j\in\{1,2\}$.

  A tuple $(\bconf{C_1}{C_2}, \bconf{C_1'}{C_2'}, \bSigma, (\beta_1,\beta_2,k))$ is compatible if 
  $\bconf{C_1}{C_2}$ and $\bconf{C_1'}{C_2'}$ compatible, $\bSigma@(\beta_1,\beta_2,k)=\bSigma$ and
  for $i,j\in\{1,2\}$:
  \begin{compactenum}
  \item[($\bot$)] if $\beta_i.1=\bot$ then $C_i=\botcconf$; if $\beta_i.2=\bot$ then $C_i'=\botcconf$;
  \item[($\diamond$)] if $\beta_i=\diamond$ then $C_i.\Ed=C_i'.\Ed=\diamond$ or $\an{C_i.e}=\an{C_i'.e}=\emptyset$; and if $C_i.\Ed=\diamond$ then $\beta_i=\diamond$.\footnote{And also $C_i'.\Ed=\diamond$ by compatibility of $C_i, C_i'$.}
    \end{compactenum}
    and moreover:
  \begin{itemize}
    \item if $C_i$ is a proponent and $C_{\co{i}}$ an opponent configuration, then $k=i$;
    \item if $C_1$ and $C_2$ are proponent configurations, then $k=0$;
    \item if $k=0$ then $\beta_1.e=\beta_2.e$ and $C_1.e = C_2.e$ or $C_1.\Ed = C_2.\Ed$, and the same for $C_1'$, $C_2'$.
  \end{itemize}
\end{definition}

\begin{definition}[Weak Pair (Bi)Simulation]
  A relation \bbisim{R} with elements of the form $(\bC_1,\bC_2,\bSigma,\bbeta)$, and membership thereof denoted $\bC_1 \bbisim*{R}_{\bSigma,\bbeta} \bC_2$, is called 
  \emph{weak pair simulation} when
  for all $\bC_1 \bbisim*{R}_{\bSigma,\bbeta^k} \bC_2$ we have $(\bC_1,\bC_2,\bSigma,\bbeta^k)$ compatible
 and:
  \begin{enumerate}
    \item[0.]
      if $\bC_1\downarrow$
      then $\bC_2\downarrow$
    \item
      if $\bC_1 \realtrans{\lopret{D[\vec\alpha]},k} \bC_1'$ with  $\vec\alpha\fresh \bC_2$ 
      then $\bC_2 \wtrans{\lopret{D[\vec\alpha]},k} \bC_2'$ and
      $\bC_1' \bbisim*{R}_{\bSigma,\bbeta^k} \bC_2'$
    \item
      if $\bC_1 \realtrans{\eta,k} \bC_1'$
      then $\bC_2 \wtrans{\eta,k} \bC_2'$ and
      $\bC_1' \bbisim*{R}_{\bSigma,\bbeta^k} \bC_2'$, for $\eta\in\{\tau,\lpropapp{\alpha}{D}\}$
    \item
      if $\bC_1 \realtrans{\lopapp{i}{D\hole[\vec \alpha]},j}  \bC_1'$ with  $\vec\alpha\fresh \bC_2$ 
      then $\bC_2 \trans{\lopapp{i}{D\hole[\vec \alpha]},j}  \bC_2'$
      and
      $\bC_1' \bbisim*{R}_{\bSigma',\bbeta'} \bC_2'$
      and $\bSigma' = \bSigma[\bbeta' \mapsto (\Ed_1, \Ed_2, \bbeta^k)]$
      with 
      \begin{gather*}\begin{array}{lllll}
        \bbeta' = (\betais{\bC_1'.1,\bC_2'.1}, \betais{\bC_1'.2,\bC_2'.2}, 0) & \Ed_1 = \bC_1.1.\Ed = \bC_1.2.\Ed, & \Ed_2 = \bC_2.1.\Ed = \bC_2.2.\Ed & \text{ if } j = 0; \text{ or}\\
        \bbeta' = (\betais{\bC_1'.1,\bC_2'.1}, \bbeta.2, 1)                   & \Ed_1 = \bC_1.1.\Ed                & \Ed_2 =\bC_2.1.\Ed                & \text{ if } j = 1; \text{ or}\\
        \bbeta' = (\bbeta.1, \betais{\bC_1'.2,\bC_2'.2}, 2)                   & \Ed_2 = \bC_1.2.\Ed                & \Ed_2 =\bC_2.2.\Ed                & \text{ if } j = 2\phantom{; \text{ or}}
      \end{array}\end{gather*}
    \item
      if $\bC_1 \realtrans{\lpropret{D},k} \bC_1'$
      and
      $\bbeta^k{\toSig[\bSigma]{\Edb_1,\Edb_2}}\bbeta'$
      then $\bC_2 \wtrans{\lpropret{D},k} \bC_2'$
      and
      $\bC_1'\sub{\chi}{\Edb_1} \bbisim*{R}_{\bSigma',\bbeta'} \bC_2'\sub{\chi}{\Edb_2}$ with $\bSigma'=\bSigma@\bbeta'$.
    \end{enumerate}
  Pair Similarity $(\bsimil)$ is the largest weak simulation.
  Relation \bbisim{R} is a \emph{weak pair bisimulation} when \bbisim{R} and $\bbisim{R}^{-1}$ are weak pair simulations, where $R^{-1}=\{(\bC_1,\bC_2,\bSigma^{-1},\bbeta^{-1})\mid (\bC_2,\bC_1,\bSigma,\bbeta)\in \bbisim{R}\}$.
  Pair bisimilarity $(\bbisimil)$ is the largest weak bisimulation.
\end{definition}

\begin{definition}
  Given two continuation graphs $\Sigma_1$, $\Sigma_2$ we construct the product graph:
  \begin{gather*}
    \irule{
    }{
      ((\diamond,\diamond,0), \diamond, \diamond, (\diamond, \diamond, 0)) \in \Sigma_1 \otimes \Sigma_2
    }
    \\
    \irule{
      (\beta_1,\beta_2,k) \in \dom{\Sigma_1\otimes\Sigma_2}
      \\
      \forall i \in \{1,2\}.~(\beta_i', \Ed, \Ed', \beta_i) \in \Sigma_i
      \\
      \forall j \in \{1,2\}.~\beta'_1.j.e=\beta'_2.j.e
    }{
      ((\beta_1',\beta_2',0), \Ed, \Ed', (\beta_1, \beta_2, k)) \in \Sigma_1\otimes\Sigma_2
    }
    \\
    \irule{
      (\beta_1,\beta_2,k) \in \dom{\Sigma_1\otimes\Sigma_2}
      \\\\
      (\beta_i',\Ed, \Ed', \beta_i) \in \Sigma_i
      \\
      \beta_{\co{i}} = \beta_{\co{i}}'
      \\
      i\in\{1,2\}
    }{
      ((\beta_1',\beta_2',i), \Ed, \Ed', (\beta_1, \beta_2, k)) \in \Sigma_1\otimes\Sigma_2
    }
  \end{gather*}
  \begin{lemma}
    Suppose  $\Sigma_1, \Sigma_2$ well-formed continuation graphs; then $\Sigma_1 \otimes \Sigma_2$ is a well-formed pair continuation graph.
  \end{lemma}
  \begin{proof}
    By induction on the construction of $\Sigma_1 \otimes \Sigma_2$, reasoning about each condition separately and using the induction hypothesis for proving reachability and the second part of closure. Because $\Sigma_i$ ($i\in\{1,2\}$) are closed under permutations of their tuples, so is the product graph.
\qed  \end{proof}
  Moreover, given two stackless configuration relations $\bisim{R}_1$,$\bisim{R}_2$, we define the pair relation $\bisim{R}_1\otimes\bisim{R}_2$ by induction:
  \begin{gather*}
    \irule{
      C_1 \bisim*{R}_{1\,\Sigma_1,\beta_1} C_1'
      \\
      C_2 \bisim*{R}_{2\,\Sigma_2,\beta_2} C_2'
      \\\\
      C_1.\Phi = C_2.\Phi
      \\
      C_1'.\Phi = C_2'.\Phi
      \\\\
      \bbeta = (\beta_1,\beta_2,0) \in \dom{\Sigma_1\otimes\Sigma_2}
    }{
      \bconf{C_1}{C_2}
      \mathrel{(\bisim{R}_1\otimes\bisim{R}_2)_{(\Sigma_1\otimes\Sigma_2)@\bbeta,\;\bbeta}}
      \bconf{C_1'}{C_2'}
    }
    \\
    \irule{
      C_1 \bisim*{R}_{1\,\Sigma_1,\beta_1} C_1'
      \\
      C_2 \bisim*{R}_{2\,\Sigma_2,\beta_2} C_2'
      \\\\
      \bbeta = (\beta_1,\beta_2,i) \in \dom{\Sigma_1\otimes\Sigma_2}
      \\
      C_{\co{i}},C_{\co{i}}' \text{ opponent configurations}
    }{
      \bconf{C_1}{C_2}
      \mathrel{(\bisim{R}_1\otimes\bisim{R}_2)_{(\Sigma_1\otimes\Sigma_2)@\bbeta,\;\bbeta}}
      \bconf{C_1'}{C_2'}
    }
  \end{gather*}
\end{definition}

\begin{lemma}\label{lem:bsound}
  Let $\bisim{R}_1$ and $\bisim{R}_2$ are stackless simulations; then $\bisim{R}_1\otimes\bisim{R}_2$ is a pair simulation.
\end{lemma}
\begin{proof}
  By showing that $(\eta,i)$ transitions are matched because $\bisim{R}_i$  is a simulation (for $i\in\{1,2\}$), and $\eta,0$ transitions are matched because both $\bisim{R}_1$, $\bisim{R}_2$ are simulations and the constituent configurations can perform the same moves.
  In addition we use simple lemmas to show that the extension and reachability operations of the pair continuation graph is captured by the construction of the product relation (cf.~\cref{sec:pair-lemmas}).
\qed\end{proof}


\subsection{Pair Configuration Lemmas}
\label{sec:pair-lemmas}

\begin{lemma}
  Suppose $\Sigma_i = \Sigma_i@\beta_i (i\in\{1,2\})$ are valid call graphs.
  Then $(\beta_1,\beta_2,k)    \in \dom{\Sigma_1\otimes\Sigma_2}$.
\qed\end{lemma}
\begin{lemma}
  Suppose
  \begin{align*}
    \Sigma_i'              & = \Sigma_i[\beta_i' \mapsto (E_a, E_b, \beta_i)] & (i\in\{1,2\}) \\
    \Sigma_i               & = \Sigma_i@\beta_i                               & (i\in\{1,2\})
  \end{align*}
  Then
  \begin{align*}
    (\beta_1',\beta_2',k')                             & \in \dom{\Sigma_1'\otimes\Sigma_2'} \qquad\text{and}\\
    (\Sigma_1'\otimes\Sigma_2')@(\beta_1',\beta_2',k') & =
    ((\Sigma_1\otimes\Sigma_2)@(\beta_1,\beta_2,k))[(\beta_1',\beta_2',k')\mapsto(E_a,E_b,k)]
  \end{align*}
\qed\end{lemma}
\begin{lemma}
  Suppose
  \begin{align*}
    \Sigma_i                & = \Sigma_i@\beta_i             & (i\in\{1,2\}) \\
    (\beta_i',E_{1a},E_{2b},\beta_i) & \in \Sigma_i          & (i\in\{1,2\}) \\
  \end{align*}
  Then
  \begin{align*}
    ((\Sigma_1\otimes\Sigma_2)@(\beta_1',\beta_2',k'))@(\beta_1,\beta_2,k) =
    (\Sigma_1\otimes\Sigma_2)@(\beta_1,\beta_2,k)
  \end{align*}
\qed\end{lemma}
\begin{lemma}
  Suppose $\bbeta\in\dom\bSigma$; then:
  \begin{enumerate}
    \item $\bmerge{\bSigma}@\bmerge{\bbeta} = \bmerge{\bSigma@\bbeta}$
    \item $\bmerge{\bSigma}[\bmerge{\bbeta'} \mapsto (\Ed_1,\Ed_2,\bmerge{\bbeta})] = \bmerge{\bSigma[\bbeta' \mapsto (\Ed_1,\Ed_2,\bbeta)]}$\qed
  \end{enumerate}
\end{lemma}

\section{Up to Separation}

\subsection{Soundness of Up to Separation}
\label{appx:uptosep-sound}

\begin{proof}(\cref{lem:utsepconj-sound})
  Let
  \begin{gather*}
    C_1 \sepconj[I][k] C_2 \utsepconj*{\bisim{R}}_{\Sigma,\bmerge{\bbeta_{12}}} C_1' \sepconj[I][k] C_2'\\
    C_1 = \pconf{A_1}{\Gamma_1\sepconj[I,L]\Gamma}{s_1\sepconj[L] s}{e_1} \qquad
    C_2 = \oconf{A_2}{\Gamma_2\sepconj[I,L]\Gamma}{s_2\sepconj[L] s}{\Ed_2} \\
    C_1' = \pconf{A_1'}{\Gamma_1'\sepconj[I,L']\Gamma}{s_1'\sepconj[L'] s'}{e_1'} \qquad
    C_2' = \oconf{A_2'}{\Gamma_2'\sepconj[I,L']\Gamma}{s_2'\sepconj[L'] s'}{\Ed_2'}
  \end{gather*}
  and
  $C_i \simil_{\Sigma_i,\beta_i} C_i'$ ($i\in\{1,2\}$)
  and
  $\Sigma=\bmerge{(\Sigma_1\otimes\Sigma_2)@\bbeta_{12}}$
  and
  $\bbeta_{12}= (\beta_1,\beta_2,k)$.
  We show the case where $k=1$ and the two interesting subcases of $\WP$:
  
  \paragraph{$\bullet$ Case 
  $C_1 \sepconj[I][1] C_2 \realtrans{\eta} C_{32}$
  and
  $\eta=\lpropret{D}$
  and
  $(\Edb_3,\Edb_3',\bmerge{\bbeta_{32}}) \in \Sigma(\bmerge{\beta_{12}})$:}~\\
  We have
  $C_1 \realtrans{\eta} C_3$ and
  \begin{gather*}
    C_3 = \oconf{A_1}{\Gamma_1,\Gamma_3\sepconj[I,L]\Gamma}{s_1\sepconj[L] s}{\chi}
    \qquad
    C_{32} = C_3\sepconj[I][k] C_2
    \\
    (\Edb_3,\Edb_3',\beta_3) \in \Sigma_1(\beta_1)
    \qquad
    \bbeta_{32}=(\beta_3,\beta_2, k_{32})
  \end{gather*}
  We then have
  \begin{gather*}
    \bconf{C_1}{C_2} \realtrans{\eta,1} \bconf{C_3}{C_2}
    \qquad
    (\Edb_3,\Edb_3',\bbeta_{32}) \in (\Sigma_1\oplus\Sigma_2)(\bbeta_{12})
  \end{gather*}
  By \cref{lem:bsound}, and $\bconf{C_1}{C_2} \mathrel{(\simil \otimes \simil)_{(\Sigma_1\otimes\Sigma_2)@\bbeta_{12},~\bbeta_{12}}} \bconf{C_1'}{C_2'}$ :
  \begin{gather*}
    \bconf{C_1'}{C_2'} \wtrans{\eta,1} \bconf{C_3'}{C_2}
    \qquad
    C_3' = \oconf{A_1'}{\Gamma_1',\Gamma_3'\sepconj[I,L']\Gamma'}{s_1'\sepconj[L] s'}{\chi}
    \\
    \bconf{C_3}{C_2}\sub{\chi}{\Edb_3}=
    \bconf{C_3\sub{\chi}{\Edb_3}}{C_2}
    \mathrel{(\simil \otimes \simil)_{(\Sigma_1\otimes\Sigma_2)@\bbeta_{32},~\bbeta_{32}}}
    \bconf{C_3'\sub{\chi}{\Edb_3'}}{C_2}=
    \bconf{C_3'}{C_2}\sub{\chi}{\Edb_3'}
  \end{gather*}
  Moreover,
  \begin{gather*}
    {C_1'}\sepconj[I][1]{C_2'} \wtrans{\eta} {C_3'}\sepconj[I][1]{C_2'}
    \\
    C_3\sub{\chi}{\Edb_3} \sepconj[I][k_{32}] C_2
    \utsepconj*{\simil}_{\bmerge{(\Sigma_1\otimes\Sigma_2)@\bbeta_{32}},~\bmerge{\bbeta_{32}}}
    C_3'\sub{\chi}{\Edb_3} \sepconj[I][k_{32}] C_2'
  \end{gather*}

  \paragraph{$\bullet$ Case
  $C_1 \sepconj[I][1] C_2 \realtrans{\eta} C_{32}$
  and
  $\eta=\lopapp{i}{D\hole[\vec \alpha]}$:}~\\
  If $i\in\dom{\Gamma_1}$, we have
  $C_1 \realtrans{\eta} C_3$ and
  \begin{gather*}
    C_3 = \oconf{A_1}{\Gamma_1\sepconj[I,L]\Gamma}{s_1\sepconj[L] s}{\Ed_3}
    \qquad
    C_{32} = C_3\sepconj[I][1] C_2
  \end{gather*}
  We then have
  $
    \bconf{C_1}{C_2} \realtrans{\eta,1} \bconf{C_3}{C_2}
  $.
  By \cref{lem:bsound}, $\bconf{C_1}{C_2} \mathrel{(\simil \otimes \simil)_{(\Sigma_1\otimes\Sigma_2)@\bbeta_{12},~\bbeta_{12}}} \bconf{C_1'}{C_2'}$ :
  \begin{gather*}
    \bconf{C_1'}{C_2} \trans{\eta,1} \bconf{C_3'}{C_2}
    \qquad
    \beta_e = \betais{C_3, C_3'}
    \qquad
    \bbeta_{32}=(\beta_3,\beta_2, 1)
    \\
    \bSigma_{32} = (\Sigma_1\oplus\Sigma_2)[\bbeta_{32} \mapsto (\Edb_3,\Edb_3',\bbeta_{12})]
    \qquad
    \bconf{C_3}{C_2}
    \mathrel{(\simil \otimes \simil)_{\bSigma_{32},~\bbeta_{32}}}
    \bconf{C_3'}{C_2}
  \end{gather*}
  Moreover,
  \begin{gather*}
    {C_1'}\sepconj[I][1]{C_2'} \trans{\eta} {C_3'}\sepconj[I][1]{C_2'}
    \\
    C_3 \sepconj[I][1] C_2
    \utsepconj*{\simil}_{\bmerge{\bSigma_{32}@\bbeta_{32}},~\bmerge{\bbeta_{32}}}
    C_3' \sepconj[I][1] C_2'
  \end{gather*}

  If $i\in\Gamma$ the proof is similar, with the exception that we consider transition $\eta,0$ from the pair configurations.
\qed\end{proof}

\subsection{Completeness of Up to Separation}
\label{appx:uptosep-complete}

\begin{proof}(\cref{lem:utsepconj-complete})
  Let $\bisim{R}$ be a simulation and
  $C_1 \sepconj[I][0] C_2 \simil_{\Sigma,\beta} C_1' \sepconj[I][0] C_2'$.
  We will show that $C_1 \simil_{\Sigma',\beta'} C_1'$, for some $\Sigma',\beta'$ (the proof for $C_2,C_2'$ is symmetric).
  We unfold the definition for
  $C_1 \sepconj[I][0] C_2$ and
  $C_1' \sepconj[I][0] C_2'$, considering cases.
  The case where $C_1 \sepconj[I][0] C_2 = C_1=\botconf$ is trivial.
  The remaining two cases are similar and we only show one:
  \begin{align*}
    C_1&=\oconf{A_1}{\Gamma_1\sepconj[I,L]\Gamma}{s_1\sepconj s}{\Ed}
    \\
    C_2&=\oconf{A_2}{\Gamma_2\sepconj[I,L]\Gamma}{s_2\sepconj[L] s}{\Ed}
    \\
    C_1 \sepconj[I][0] C_2 &= 
    \oconf{A_1\cup A_2}{\Gamma_1\sepconj\Gamma_2\sepconj\Gamma}{s_1\sepconj s_2\sepconj s}{\Ed}
    \\
    C_1 \sepconj[I][0] C_2 &\mathrel{{\weakrelinv}^{(n)}}
    \oconf{A_1\cup A_2}{\Gamma_1\sepconj\Gamma}{s_1\sepconj s_2\sepconj s}{\Ed} = C_3
  \end{align*}
  where ${\weakrelinv}^{(n)} =\; \weakrelinv_{i_1}\ldots\weakrelinv_{i_n}$ and $\{i_1,\ldots,i_n\}=\dom{\Gamma_2}$.
  Similarly
  \begin{align*}
    C_1'&=\oconf{A_1'}{\Gamma_1'\sepconj[I,L']\Gamma'}{s_1'\sepconj s'}{\Ed'}
    \\
    C_2'&=\oconf{A_2'}{\Gamma_2'\sepconj[I,L']\Gamma'}{s_2'\sepconj[L'] s'}{\Ed'}
    \\
    C_1' \sepconj[I][0] C_2' &= 
    \oconf{A_1'\cup A_2'}{\Gamma_1'\sepconj\Gamma_2'\sepconj\Gamma'}{s_1'\sepconj s_2'\sepconj s'}{\Ed'}
    \\
    C_1' \sepconj[I][0] C_2' &\mathrel{{\weakrelinv}^{(n)}}
    \oconf{A_1'\cup A_2'}{\Gamma_1'\sepconj\Gamma'}{s_1'\sepconj s_2'\sepconj s'}{\Ed'} = C_3'
  \end{align*}
  By \cref{lem:utweak-sound}, we have
  $C_3 \simil_{\Sigma_3,\beta_3} C_3'$, for some $\Sigma_3,\beta_3$. Moreover,
  \begin{align*}
    C_3 & \gcrelinv_{A_2,\dom{s_2}} \oconf{A_1}{\Gamma_1\sepconj\Gamma}{s_1\sepconj s}{\Ed} = C_1
    \\
    C_3' & \gcrelinv_{A_2',\dom{s_2'}} \oconf{A_1'}{\Gamma_1'\sepconj\Gamma'}{s_1'\sepconj s'}{\Ed'} = C_1'
  \end{align*}
  By \cref{lem:utgc-sound}, we have
  $C_1 \simil_{\Sigma_4,\beta_4} C_1'$, for some $\Sigma_4,\beta_4$.
\qed\end{proof}


\end{document}